\documentclass[11pt]{article}
\usepackage[mathlines]{lineno}
\usepackage[letterpaper,margin=1in]{geometry}
\usepackage{amsmath,amsthm,amssymb,amsfonts,thmtools}
\usepackage[ruled,linesnumbered,noend]{algorithm2e}
\usepackage[dvipsnames]{xcolor}
\usepackage{appendix}
\usepackage{babel}
\usepackage{bm}
\usepackage{bbm}
\usepackage{mathtools}
\usepackage{mathrsfs}
\usepackage{comment}
\usepackage{graphicx}
\usepackage[inline]{enumitem}
\usepackage{thm-restate}

\theoremstyle{plain}
\newtheorem{theorem}{Theorem}[section]
\newtheorem{lemma}[theorem]{Lemma}

\newtheorem{corollary}[theorem]{Corollary}

\newtheorem{proposition}[theorem]{Proposition}

\newtheorem{observation}[theorem]{Observation}

\theoremstyle{definition}
\newtheorem{definition}[theorem]{Definition}

\theoremstyle{remark}

\newtheorem*{remark*}{Remark}

\usepackage[pdfusetitle]{hyperref}
\usepackage[capitalise,noabbrev]{cleveref}
\hypersetup{colorlinks=true,allcolors=blue}

\makeatletter
\AtBeginDocument{%
  \@ifundefined{theHtheorem}{%
  }{%
  }%
  \@ifundefined{theHlemma}{%
  }{%
  }%
  \@ifundefined{theHclaim}{%
  }{%
  }%
  \@ifundefined{theHcorollary}{%
  }{%
  }%
  \@ifundefined{theHquestion}{%
  }{%
  }%
  \@ifundefined{theHproposition}{%
  }{%
  }%
  \@ifundefined{theHfact}{%
  }{%
  }%
  \@ifundefined{theHobservation}{%
  }{%
  }%
  \@ifundefined{theHproperty}{%
  }{%
  }%
  \@ifundefined{theHdefinition}{%
  }{%
  }%
  \@ifundefined{theHexample}{%
  }{%
  }%
  \@ifundefined{theHproblem}{%
  }{%
  }%
  \@ifundefined{theHremark}{%
  }{%
  }%
}
\makeatother

\allowdisplaybreaks[4]

\crefname{AlgoLine}{line}{lines}
\Crefname{AlgoLine}{Line}{Lines}
\crefname{problem}{problem}{problems}
\Crefname{problem}{Problem}{Problems}
\crefname{claim}{claim}{claims}
\Crefname{claim}{Claim}{Claims}

\let\epsilon\varepsilon

\DeclareMathOperator*{\E}{\mathbb{E}}
\DeclareMathOperator{\ddim}{\operatorname{ddim}}
\DeclareMathOperator{\dist}{\mathbf{d}}
\DeclareMathOperator{\diam}{\operatorname{diam}}
\DeclareMathOperator{\cost}{\operatorname{cost}}
\DeclareMathOperator{\floc}{\cost_{\mathrm{fl}}}
\DeclareMathOperator*{\argmin}{argmin}

\DeclareMathOperator{\poly}{\operatorname{poly}}

\DeclareMathOperator{\ccostport}{\operatorname{ccost-port}}
\DeclareMathOperator{\ocost}{\operatorname{ocost}}
\DeclareMathOperator{\opt}{\operatorname{opt}}
\DeclareMathOperator{\Bad}{\operatorname{Bad}}
\DeclareMathOperator{\child}{\operatorname{Child}}
\DeclareMathOperator{\Nchild}{\operatorname{N-Child}}
\DeclareMathOperator{\Schild}{\operatorname{O-Child}}
\def \vectorset{\mathrm{set}}
\DeclareMathOperator{\feas}{D}
\DeclareMathOperator{\hmin}{a}
\DeclareMathOperator{\hmax}{b}

\newcommand{\decom}{\mathcal{H}}
\newcommand{\modifydecom}{\mathcal{P}}
\newcommand{\cT}{\mathcal{T}}
\newcommand{\cP}{\mathcal{P}}
\newcommand{\cC}{\mathcal{C}}
\newcommand{\cS}{\mathcal{S}}
\newcommand{\bX}{\mathbb{X}}
\newcommand{\NN}{\mathbb{N}}
\newcommand{\RR}{\mathbb{R}}
\newcommand{\braket}[1]{\langle #1 \rangle}
\newcommand{\veca}{\mathbf{a}}
\newcommand{\vecb}{\mathbf{b}}
\newcommand{\comp}{z}
\newcommand{\optkl}{\opt_{k,\ell}}
\newcommand{\optFL}{\opt_{\mathrm{fl}}}
\newcommand{\optkmed}{\opt_k}
\newcommand{\median}{\cost_{k}}
\newcommand{\epsddim}{\log(\ddim/\epsilon)}
\newcommand{\expepsddim}{\frac{\ddim}{\epsilon}}
\newcommand{\proj}[2]{\pi_{#1}(#2)}

\newcommand{\dportal}{\dist_{\mathrm{port}}}
\newcommand{\dportP}{\dportal^\modifydecom}
\newcommand{\dportH}{\dportal^\decom}
\newcommand{\dporthatH}{\widehat{\dist}_{\mathrm{port}}^\decom}
\def \eps{\epsilon}

\title{Near Linear Time Approximation Schemes for Clustering of Partially Doubling Metrics}

\author{
Anne Driemel\thanks{University of Bonn}
\and
Jan H\"ockendorff \thanks{University of Cologne} \thanks{Funded by the Deutsche Forschungsgemeinschaft (DFG, German Research Foundation) – Project Number 459420781}
\and
Ioannis Psarros \thanks{University of Athens}
\and
Christian Sohler \thanks{University of Cologne}
\and
Di Yue \thanks{University of Toronto. The work was done while visiting University of Cologne.}
}
\date{}

\begin{document}

\begin{titlepage}
\maketitle
\begin{abstract}
In the metric $k$-median problem we are given a finite metric space $(X\cup Y, \mathbf{d})$ and the objective is to compute a set of $k$ centers $C\subseteq Y$ that minimizes $\sum_{p\in X} \min_{c\in C} \mathbf{d}(p,c)$. In general metric spaces, the 
best polynomial time algorithm, which is due to Cohen-Addad, Grandoni, Lee, Schwiegelshohn, and Svensson \cite{CLSS25}, computes a $(2+\epsilon)$-approximation for arbitrary constant $\epsilon>0$. However, if the metric space has bounded doubling dimension, a near linear time $(1+\epsilon)$-approximation algorithm is known due to the work of Cohen-Addad, Feldmann, and Saulpic \cite{Cohen-AddadFS21}. 

In this paper, we show that the $(1+\epsilon)$-approximation algorithm can be generalized to the case when either $X$ or $Y$ has bounded doubling dimension (but the other set not).  
The case when $X$ has bounded doubling dimension is motivated by the assumption that even though $X$ is part of a high-dimensional space, it may be that it is close to a low-dimensional structure. 
The case when $Y$ has bounded doubling dimension is perhaps more natural. It is motivated by specific clustering problems where the centers are low-dimensional. Specifically, our work in this setting implies the first near linear time approximation algorithm for the  $(k,\ell)$-median problem under discrete Fr\'echet distance when $\ell$ is constant. The latter problem is a version of the $k$-median problem under Fr\'echet distance when the input consists of time series of $z$ reals and where the centers are time series of $\ell$ reals \cite{DKS16}.
Previously, for this problem no $(1+\epsilon)$-approximation algorithm with running time polynomial in $k$ was known. 
We also introduce a novel complexity reduction for time series of real values that leads to a similar result for the case of discrete Fr\'echet distance.

In order to solve the case when $Y$ has a bounded doubling dimension, we introduce a form of dimension reduction that replaces points from $X$ by sets of points in $Y$. 
To solve the case when $X$ has a bounded doubling dimension, we generalize Talwar's decomposition \cite{Talwar04} of doubling metrics to our setting. The running time of our algorithms is $2^{2^t} \tilde O(n+m)$ where $t=O(\mathrm{ddim} \log \frac{\mathrm{ddim}}{\epsilon})$ and where $\mathrm{ddim}$ is the doubling dimension of $X$ (resp.\ $Y$). 
The results
also extend to the metric (uncapacitated) facility location problem.
We believe that our techniques are likely applicable to other problems. 
\end{abstract}
\newpage
\tableofcontents
\newpage
\end{titlepage}

\section{Introduction}
\label{sec:intro}

Partitioning of data sets according to data characteristics is one of the most fundamental problems
in data analysis and optimization. Depending on the underlying problem there are 
many different variants of partitioning problems. In this paper, 
we consider the metric $k$-median problem and the (closely related) metric facility location problem. The former problem
belongs to the area of unsupervised learning and is formulated as follows. We are given a metric space $(X\cup Y, \dist)$, where the set $X$ consists of the data points and the set $Y$ of allowed center locations. The goal is to choose a set $C$ of $k$ centers from $Y$ such that the $\median(X,C) = \sum_{x\in X} \min_{c\in C} \dist(x,c)$ is minimized. The set $C$ induces a partition of $X$ into $k$ sets by assigning each point to its nearest center. The resulting partitioning of $X$ is also called a \emph{clustering}. The facility location problem is closely related even though it originates from a different setting. In facility location we have a set of clients $X$ and a set of possible facilities $Y$ both from a metric space $(X\cup Y, \dist)$ and with each $f\in Y$ there is an opening cost $\ocost(f)$ associated. Every client has to be assigned to an open facility and will pay the distance to the facility as connection cost. 
The objective is to find a set of facilities $F\subset Y$ such that the sum of opening and connection costs is minimized, that is, we want to minimize $\sum_{x\in X} \min_{f\in F} \dist(x,f) + \sum_{f\in F} \ocost(f)$. Thus, the main difference between the two problems is, that in $k$-median clustering
the number of centers is restricted to $k$, while in facility location we may open as many centers as we want, but we need to pay for it. Otherwise, both objective functions minimize the sum of distances of points from $X$ to their nearest centers.

Neither problem admits 
a polynomial time $(1+\epsilon)$-approximation algorithm 
for arbitrary small constant $\epsilon>0$
under standard complexity theoretical assumptions. In fact, the $k$-median problem cannot be approximated better than $1+2/e$
\cite{JMS02} and the facility location problem not better than $1.463$ \cite{GK99,JMS02}.
At the same time, both problems can be fairly well approximated. For the $k$-median problem the best possible approximation algorithm achieves a factor $2+\epsilon$ approximation \cite{CLSS25} and for the facility location problem, the best known approximation factor is 1.488 \cite{Li11}. Interestingly,
both problems can be approximated up to a factor of $(1+\epsilon)$, if the underlying metric space has bounded doubling dimension \cite{Cohen-AddadFS21}.

In this paper, we raise the question whether a $(1+\epsilon)$-approximation can also be achieved, if only one of the sets has bounded doubling dimension, but the other set is high-dimensional. This setting has been studied before in the context of nearest neighbor search \cite{IN07, Har-PeledK13} and Euclidean facility location \cite{HJKY25}. In this paper, we give the first near-linear time $(1+\epsilon)$-approximation algorithms for facility location and $k$-median in these settings.

We then observe that other known clustering problems have centers coming from a  space that is low-dimensional: An example is the 
 $(k,\ell)$-median clustering problem \cite{DKS16} of time series data
under the discrete Fréchet distance. The Fréchet distance is a standard distance measure for polygonal curves.
These curves can be represented as a time series of points. The Fréchet distance is especially suited for comparing  series of different complexity (number of points in the time series). We will present in Section \ref{sec:discreteFrechet} how one can solve the above problem using our algorithm for doubling metrics. In the $(k,\ell)$-median problem the length of the center time series is restricted to a constant $\ell$, which essentially means that the space of center time series is doubling (albeit it does not imply that the space of clients is doubling). 

We then apply our algorithm to get the first $(1+\epsilon)$-approximation algorithm for this problem with a near linear running time (for constant $\ell$) while all previous algorithms were exponential
in $k$~\cite{buchin2019hardness,nath2021k}. 
For the case of one-dimensional ambient space we also give a more direct algorithm that is based on a new complexity reduction method for the discrete Fr\'echet distance
that reduces the problem to the case of bounded doubling dimension.
  We further investigate the reverse setting when the data points are low-dimensional and the candidate center set is high dimensional. 
Such a setting may arise when high-dimensional data is located on or near a low-dimensional structure, which is 
a common assumption, for example, in the field of manifold learning. We show that also in the reverse setting we obtain an almost linear time approximation scheme.

This paper is a direct follow up to the work of Driemel et al. \cite{driemel2025nearlineartimeapproximationscheme} that introduced a complexity reduction for time series and presented a near linear time $(1+\varepsilon)$-approximation algorithm for $(k,\ell)$-median for time series. This paper strictly improves both results by giving an explicit bound to the complexity reduction and generalizing the $(1+\varepsilon)$-approximation to polygonal curves of arbitrary ambient dimension. 

\subsection{Problem Definitions}

In this section we introduce the problems we consider. 
We start by defining the metric facility location problem.
Let $(X \cup Y, \dist)$ be a metric space.
Input to the facility location problem is a set of $n$ data points (or clients) $X$ and a set of $m$ candidate facilities $Y$. Every facility $f \in Y$ is assigned with an opening cost $\ocost(f) > 0$.
The goal is to output a set of facilities $F \subseteq Y$, such that
\begin{equation*}
    \floc(X, F) := \sum_{x \in X} \dist(x, F) + \sum_{f \in F} \ocost(f)
\end{equation*}
is minimized, where $\dist(x, F) := \min_{f\in F} \dist(x,f)$.
The optimal facility location value is denoted by $\optFL(X, Y) := \min_{F \subseteq Y} \floc(X, F)$.

In the metric $k$-median problem, our input is also a set of $n$ data points $X$ and a set of $m$ candidate centers $Y$. The objective is to find a subset $F\subseteq Y$ of centers such that 
\begin{equation*}
    \median(X, F) := \sum_{x \in X} \dist(x,F)
\end{equation*}
is minimized over all sets $F\subseteq Y$
of cardinality $k$.
The optimal $k$-median value is denoted by $\optkmed(X, Y) := \min_{F \subseteq Y, |F| \leq k} \median(X, F)$.

Throughout this paper, unless stated otherwise, we assume that the distances in $X \cup Y$ can be computed in constant time.
We remark that in a setting where this is not the case and we require time $T$ to compute a distance between two points in $X \cup Y$ then we get an additional factor of $T$ in the running time.
It is known that
if the metric space $(X\cup Y,\dist)$ has bounded doubling dimension, we have a near linear time approximation algorithm \cite{Cohen-AddadFS21}. 
In this paper we study the problem variant, where just one of the two sets $X$ and $Y$ has bounded  doubling dimension. Perhaps surprisingly, we show that in these cases we can still obtain a $(1+\epsilon)$-approximation algorithm. 
Since one can $O(1)$-approximate the doubling dimension of a metric in $2^{O(\ddim)}n \log n$ time \cite{Har-PeledM06}, we 
will assume that our algorithms are given an upper bound $\ddim$ on the doubling dimension.

\subsubsection{Application: The $(k,\ell)$-Median Problem under Discrete Fr\'echet Distance}

We would like to apply our results to a geometric variant of $k$-median clustering that we will define in the following. 
We denote a polygonal curve $\pi$ with vertices $p_1,\ldots, p_\comp$ in $\RR^d$ as $\pi = \langle p_1, p_2, \dots, p_\comp\rangle$. The \emph{complexity} of a polygonal curve is the number of points of the sequence defining it. The set of all polygonal curves of complexity $\comp$ with vertices in $\RR^d$ is denoted by $\bX_{\comp}^d$. 
A \emph{traversal} $T$ between a polygonal curve of complexity $\comp$ and a polygonal curve of complexity $\ell$ is a sequence of index pairs
$T =(i_1, j_1), (i_2, j_2), \dots, (i_t, j_t)$ 
such that the following conditions hold:
\begin{enumerate*}[label=\roman*)]
\item $(i_1, j_1) = (1, 1)$, 
\item $(i_t, j_t) = (\comp, \ell)$, and for each  $1 \le r < t$, 
\item $i_{r+1} - i_r \in \{0, 1\}$, 
\item $j_{r+1} - j_r \in \{0, 1\}$, 
\item $(i_{r+1} - i_r) + (j_{r+1} - j_r) \ge 1$.
\end{enumerate*}
 Let $\cT_{\comp,\ell}$ be the set of all traversals between polygonal curves of complexity $\comp$ and polygonal curves of complexity $\ell$. 
\begin{definition}
The \emph{discrete Fréchet distance} between $\pi=\langle p_1,\ldots,p_\comp \rangle$ and $\tau =\langle q_1,\ldots,q_{\ell} \rangle $ is defined as:
$
\dist_{dF}(\pi, \tau) = \min_{T\in \cT_{\comp,\ell}} \max_{(i, j) \in T} \|p_i - q_j\|_2.
$
\end{definition}

We consider the $(k,\ell)$-median problem for clustering under the discrete Fr\'echet distance that has been introduced in \cite{DKS16} in the context of the continuous Fr\'echet distance. 
The problem is a variant of the $k$-median problem under Fr\'echet distance, where the complexity of the center polygonal curves is restricted to be at most $\ell$. 

\begin{definition}[($k,\ell)$-median clustering problem]
   Given a set of polygonal curves $\Pi\subset \bX_{\comp}^d$ and parameters $k, \ell \in \NN$, compute a set $\cC \subset \bX_{\ell}^d$, $|\cC|=k$,  
   that minimizes
    $\sum_{\pi \in \Pi}\min_{\tau \in \cC} \dist_{dF}(\pi,\tau).$ 
\end{definition}
We  remark that one may invariantly define the centers to have complexity at most $\ell$, instead of exactly $\ell$. This does not change the problem, since every polygonal curve with complexity fewer than $\ell$ can be extended to a polygonal curve with $\ell$ vertices by repeating the first element of the polygonal curve without affecting the discrete Fr\'echet distance.

\subsection{Our Results}

In the following section, we present our results. We first show that for the $k$-median and the facility location 
problem there are near-linear time $(1+\epsilon)$-approximation algorithms when the set of facilities $Y$ has bounded doubling dimension. This extends previous results that were restricted to the case that $X\cup Y$ has bounded doubling dimension.
The theorem below summarizes Theorems \ref{thm:FL_bounded_centers} and \ref{thm:kmedian_bounded_centers}  and is proven in the corresponding sections.
\begin{theorem}\label{thm:bounded_centers}
    There are randomized algorithms that, given as input $\epsilon \in (0, \tfrac{1}{2})$, $n, m, k \in \NN$ and  $(X \cup Y, \dist)$ with $|X| = n, |Y| = m, \ddim(Y) \leq \ddim$, compute a $(1 + \epsilon)$-approximation of the $k$-median and facility location problem in time 
    $2^{2^t} \cdot \tilde{O}(n + m)$
    with constant success probability, where
    \begin{equation*}
        t \in O\left(\ddim \log\frac{\ddim}{\epsilon}\right).
    \end{equation*}
\end{theorem}

We then show a similar result for the case when the set of clients $X$ has bounded doubling dimension. We keep the theorems separated since the underlying techniques are different. The following theorem summarizes Theorems \ref{thm:FL_bounded_clients} and \ref{thm:kmedian_bounded_clients} and is proven in the corresponding sections.

\begin{theorem}\label{thm:bounded_clients}
    There are randomized algorithms that, given as input $\epsilon \in (0, \tfrac{1}{2})$, $n, m, k \in \NN$ and  $(X \cup Y, \dist)$ with $|X| = n, |Y| = m, \ddim(X) \leq \ddim$, compute a $(1 + \epsilon)$-approximation of the $k$-median and facility location problem in time 
    $2^{2^t} \cdot \tilde{O}(n + m)$
    with constant success probability, where
    \begin{equation*}
        t \in O\left(\ddim \log\frac{\ddim}{\epsilon}\right).
    \end{equation*}
\end{theorem}

We apply our result to the $(k,\ell)$-median problem under discrete Fr\'echet distance that is the first near linear time $(1+\epsilon)$-approximation algorithm for this problem when $\ell$ and $d$ (the dimension of the ambient space) are constant. All prior algorithms were exponential in $k$.

\begin{restatable}{theorem}{thmklmedian}
\label{thm: klmedian_arbitrary_d}
    There is a randomized algorithm that, given as input $\epsilon \in (0, \tfrac{1}{2})$, $n, \comp, k ,\ell \in \NN$, $P \subset \bX_{\comp}^d$ with $|P|=n$
    computes a $(1 + \epsilon)$-approximate solution to the $(k,\ell)$-median problem in time $2^{2^t} \cdot \tilde{O}(d\ell nz)$ with constant success probability, where
    \[t \in O\left(d \ell \log\frac{d\ell}{\epsilon}\right).\]
\end{restatable}

\subsection{Technical Overview}

In the following we give an overview of the main ideas of our results. 
There are two conceptual ideas, one related to the case of low-dimensional centers, the other one related to low-dimensional clients. 
In the setting of low-dimensional centers, we develop a form of dimension reduction which represents every high-dimensional client by a set of points in the low-dimensional space.
In the setting of low-dimensional clients, our main contribution is a new hierarchical decomposition that generalizes Talwar's decomposition~\cite{Talwar04} to partially doubling metric spaces.
Both ideas are independent of the concrete problems we study and we therefore believe that there is high potential to apply them to other problems. Once we have established these main ideas, there is the technical challenge to integrate them into the dynamic programming approach from \cite{Cohen-AddadFS21}. This requires overcoming several technical problems to deal with our setting.

We start by describing the case of low-dimensional centers. We illustrate our ideas on the facility location problem. 
For simplicity, we will assume uniform opening costs.
The approach to $k$-median is similar. 
Finally, we discuss an application to the $(k,\ell)$-median clustering problem under discrete Fr\'echet distance.

\subsubsection{Low-dimensional Centers}
We consider the case where the point set $X$ is in an arbitrary metric space and $Y$ has bounded doubling dimension $\ddim$. A similar setting has been considered before in \cite{GK13} where the authors develop approximation algorithms for finding the smallest set to be removed to get a set of bounded doubling dimension and where they show that one can compute an approximate minimum spanning tree and other proximity structures when $O(\sqrt{n})$ points are not doubling.
A simple illustrative example is to have a point set $X$ in $\mathbb R^d$ with Euclidean distance and a fixed $2$-dimensional plane $Y$ that is supposed to contain the centers. 
In this example, the set $Y$ will be unbounded (one can usually discretize such a space to obtain a set of candidates that contains a $(1+\epsilon)$-approximation). 

\paragraph{Dimension reduction.}
A simple idea to solve this special case is to project all points from $X$ to the plane $Y$ and solve the resulting low-dimensional problem. Such an approach will result in a constant approximation as the distance to $Y$ as argued in the following. For the analysis we can think of the projection as moving all points to $Y$. By the triangle inequality, this will change the cost of any solution by at most the sum of distances the points have been moved. Since this sum is also a lower bound on the connection cost of any solution to the high-dimensional input, we get that any solution on the projection is a constant approximation. 
For similar reasons approximating the distances from $x$ to $y\in Y$ by $\dist(x,\pi_Y(x)) + \dist(\pi_Y(x), y)$, where $\pi_Y(x)$ denotes the point from $Y$ closest to $x$, only gives a constant approximation. Instead of replacing $x$ by a single point, our new idea is to replace $x$ by a set of points $N_x$ (which we call \emph{proxies}), which is an $\epsilon \dist(x,\pi_Y(x))$-net of a ball $B_Y(\pi_Y(x), \dist(x,\pi_Y(x))/\epsilon) \subseteq Y$ around $\pi_Y(x)$ of radius $\dist(x,\pi_Y(x))/\epsilon$. 
(Formally, a $\rho$-net of a set $B$ is a subset $N \subseteq B$, such that 1) the interpoint distances in $N$ are at least $\rho$, and 
2) every point in $B$ has a nearby point in $N$ within distance $\rho$.
)
Taking $\min_{u\in N_x} \dist(x,u) + \dist(u,y)$ for $y\in Y$ we approximate all distances within a factor of 
$(1+\epsilon)$ as summarized in the following lemma.
\begin{restatable}{lemma}{LemmaDimReduction}
\label{lemma:dim-reduction-simple}
Let $\epsilon \in (0,\frac{1}{2})$. For every $x\in X$ let 
  $N_x$ be an $\epsilon \cdot \dist(x,Y)$-net of $B_Y(\pi_Y(x), \dist(x,Y)/\epsilon)$ with $\pi_Y(x)\in N_x$.
For all $x\in X$ and $y\in Y$ define 
  \begin{equation}\label{eqn:dim_reduction}
    \widehat{\dist}(x,y) = \min_{u\in N_x} \dist(x,u) + \dist(u,y).
  \end{equation}
  Then we have
  \[
 \dist(x,y) \le \widehat{\dist}(x,y) \le (1+4 \epsilon) \cdot \dist(x,y). 
  \]
\end{restatable}

We combine the above dimension reduction idea with 
the algorithm of \cite{Cohen-AddadFS21}.
We start from a standard hierarchical decomposition for doubling metrics by Talwar~\cite{Talwar04}, which is an analogue to the randomly-shifted quadtree in doubling metrics.
The decomposition $\decom$ is constructed on top of $Y$, and has $L = O(\log m)$ levels.
For each $0 \leq \ell \leq L$, level $\ell$ of $\decom$, denoted by $\decom_\ell$, is a partition of $Y$ into clusters of diameter at most $2^\ell$,
and $\decom_{\ell - 1}$ is a refinement of $\decom_\ell$. 
$\decom$ can be represented as a tree, where each node corresponds to a cluster and has $2^{O(\ddim)}$ children.
In~\cite{Cohen-AddadFS21}, a dynamic program is run on $\decom$ to compute a $(1 + \epsilon)$-approximation for facility location.
In our setting, one immediate issue is that clients in $X$ are not directly defined on $\decom$, but are replaced by a proxy set $N_x \subseteq Y$.
Therefore, it is not immediately clear that one can follow a similar approach as~\cite{Cohen-AddadFS21}.

For the sake of presentation, 
let us first have a brief review of 
the argument of \cite{Cohen-AddadFS21}, for the case $X = Y$.
They use the notion of \emph{portals}, which originates from~\cite{Arora98}.
Roughly speaking, the portal set $P_C$ for a cluster $C \in \decom_\ell$ is an $\epsilon 2^\ell$-net of $C$.
The level $\ell$ portals are the union of $P_C$ over $C \in \decom_\ell$. 
For a pair of points $x, y \in Y$, the actual distance between $x$ and $y$ will be replaced by the \emph{portal-respecting distance} $\dportH(x, y)$, which is the length of the \emph{portal-respecting path} between $x$ and $y$.
Specifically, let $\ell$ be the highest level where $\{x, y\}$ is \emph{cut} w.r.t. $\decom$, i.e., the highest level where $x$ and $y$ fall into different clusters.
The portal-respecting path between $x$ and $y$ starts from $x$ (a portal at level $0$), each step connecting the current portal to the closest portal at one level up, until reaching a portal $p$ at level $\ell$.
Portal $p$ is then connected to another portal $q$ at the same level, and the path goes all the way down from $q$ to $y$.
It can be shown that the portal-respecting distance $\dportH(x, y)$ is upper bounded by $\dist(x, y) + O(\epsilon) 2^\ell$.
A number of techniques are developed in~\cite{Cohen-AddadFS21} to bound $\ell$, the level where $\{x, y\}$ is \emph{cut}.
Finally, a dynamic program is run on $\decom$ w.r.t. portal-respecting distance $\dportH$, where each entry of the DP table is encoded by a cluster $C$ of $\decom$, and a configuration indicating how $C$ interacts with other clusters via its portals $P_C$.

\paragraph{Integrating the dimension reduction with portal-respecting distance.}
In our setting,
if we would like to combine our approach with the algorithm of~\cite{Cohen-AddadFS21},
the first step is to combine our dimension reduction idea with the portal-respecting distance.
This is done in a rather straightforward way, by modifying \eqref{eqn:dim_reduction} to
\[
\widehat{\dist}_{\mathrm{port}}^\decom(x, y) = \min_{u\in N_x} \dist(x,u) + \dportH(u,y).
\]
Now, a main challenge would be bounding the error incurred by the portal-respecting distance, namely, $\dporthatH(x, y) - \dist(x, y)$.
In~\cite{Cohen-AddadFS21}, this can be done by bounding the highest cutting level for a single pair of points $\{x, y\}$.
However, since $x$ is now represented by $N_x$, we have to bound the cutting level of $\{u, y\}$ for every $u \in N_x$, which becomes more involved.
Nevertheless, we argue that the above error can still be effectively bounded, by only considering the cutting level for two sets $B_Y(\proj{Y}{x}, \dist(x, Y) /\epsilon)$ and $\{\proj{Y}{x}, y\}$.
Specifically, if $j$ is the highest level where $B_Y(\proj{Y}{x}, \dist(x, Y) /\epsilon)$ is cut, and $\ell$ is the highest level where $\{\proj{Y}{x}, y\}$ is cut, then we show that
\[
\dporthatH(x, y) - \dist(x, y) \leq O(\epsilon) 2^{\max\{j, \ell\}}.
\]

\paragraph{Dynamic program.}
We further integrate the dimension reduction with the dynamic program framework in~\cite{Cohen-AddadFS21}.
The algorithm is run on $\decom$.
Whenever we compute the connection cost from a client $x \in X$ to some set of facilities $F \subseteq Y$, we replace $\dist(x, F)$ with $\dporthatH(x, F)$.
Two challenges come with such a replacement.
First, recall that $\dporthatH(x, F)$ depends on $\dportH(u, F)$ for all proxies $u \in N_x$;
therefore, to compute it we have to enumerate all $u \in N_x$, which can be done only when $N_x$ is entirely contained in some cluster $C$.
This is different from~\cite{Cohen-AddadFS21}, where the connection cost of $x$ can be trivially computed at the leaf node $\{x\}$.
Second, even if we successfully define such a cluster $C$ which contains $N_x$, it is unclear how we access the distance $\dportH(u, F)$.

To resolve the first issue, we find for every $x \in X$ a suitable cluster $C(x)$ that entirely contains $N_x$ in the preprocessing stage, and ``defer'' the computation of $\dporthatH(x, F)$ to cluster $C(x)$.
More concretely, for a client $x \in X$ and a cluster $C$ on $\decom$, we say $x$ is \emph{revealed} in $C$ if $B_Y(\proj{Y}{x}, \dist(x, Y)/\epsilon) \subseteq C$, and say $x$ is \emph{newly revealed} in $C$, denoted $C(x) = C$, if $C$ is the lowest level cluster where $x$ is revealed.
Then each DP table entry $(C, \veca_C)$ is defined to be the minimum \emph{revealed facility location cost} of $C$, given the configuration $\veca_C$, i.e., 
\[
g(C, \veca_C) := \min_{\substack{
    F \colon \text{$F$ is consistent} \\
    \text{with } \veca_C}
    }
    \left\{\sum_{x \in X \colon x \text{ is revealed in } C}
    \dporthatH(x, F) + \ocost(F \cap C)\right\}.
\]
Since all clients are revealed in the largest cluster (root node) $Y$, the root node stores the optimal facility location cost.

For the second issue, 
let us consider how the DP table is updated.
The revealed cost of $C$ can be decomposed into two parts --- 
(1) the total revealed cost of $C$'s child clusters, and 
(2) the \emph{newly revealed cost} of $C$, i.e., $\sum_{x \colon C(x) = C} \dporthatH(x, F)$.
The first part can be obtained from the corresponding table entries of $C$'s children.
For the second part, the key observation is that if $x$ is newly revealed in $C$ ($C(x) = C$), then 
the proxy set of $x$ is entirely contained in the portal set of $C$, i.e., $N_x \subseteq P_{C}$. 
Therefore, for each $u \in N_x$, $\dportH(u, F)$ can be directly obtained from the configuration $\veca_C$ which has the information of the connection of $P_{C}$, and thus $\dporthatH(x, F)$ can be effectively computed.

Our final algorithm has a running time of
$2^{2^t} \tilde O(n+m)$ where $t=O(\mathrm{ddim} \log \frac{\mathrm{ddim}}{\epsilon})$.
We remark that this is faster than the algorithm of \cite{Cohen-AddadFS21}, which has $t^2$ instead of $t$ in the (double) exponent. This improvement comes from the fact that one can improve the analysis of their hierarchical composition \cite{Cohen-AddadFS21}. 
The fact, that such an improvement is possible had been observed in \cite{ChanHJ18} referring to the paper \cite{AbrahamBN06}. 
In the present paper, we provide a self-contained proof of this fact in Lemma \ref{lemma:hierarchical_decomposition}.

\subsubsection{Low-dimensional Clients}
Next we consider the case when $X$ has bounded doubling dimension $\ddim$ and $Y$ not. 
One may wonder if the same
approach for the case when $Y$ is low-dimensional applies here. 
However, due to the asymmetric nature of the problem it is unclear how we could benefit from the dimension reduction in this case. Intuitively, this can be seen by the fact that ``moving'' a facility/center can be much more costly than moving a client.
Indeed, if we move a point in $X$ by a distance $z$ we change the cost of any solution by at most $z$. Moving a point in $Y$ can change the cost of a solution by as much as $|X| \cdot z$, since we could potentially assign every point in $X$ to the same point in $Y$. 

Thus, it is unclear how to apply the previous approach. 

\paragraph{Our new decomposition.}
For the reason above, we extend Talwar's hierarchical decomposition to the case that only a subset of the metric space is doubling. 
We start from constructing Talwar's decomposition on top of $X$; denote the resulting decomposition by $\decom$.
It then remains to decide how to add the points in $Y$ to $\decom$.
Intuitively, we want that every point $y \in Y$ always lies in the same cluster as $\proj{X}{y}$, the nearest neighbor of $y$ in $X$. 
More concretely, our plan is to assign to every $y \in Y$ a suitable level $0 \leq h(y) \leq L$, and add $y$ to $\decom$ in such a way that 
\begin{enumerate}[label=(\alph*)]
    \item $\{y\}$ is a leaf node at level $h(y)$; and 
    \item $y$ and $\proj{X}{y}$ are in the same cluster at levels higher than (or equal to) $h(y) + 1$.
\end{enumerate}
We call such cluster $\{y\}$ an \emph{ornament} at level $h(y)$.
After adding all $y \in Y$ to $\decom$, we obtain a hierarchical decomposition for $X \cup Y$, denoted by $\modifydecom$.

It could be tricky to define for every ornament $\{y\}$ the level $h(y)$ to which it should be attached as a leaf.
On the one hand, we want $h(y)$ to be sufficiently large, so that $y$ can be covered by some level $h(y)$ portal (i.e., $\dist(y, P_C) \leq \epsilon 2^{h(y)}$ for the cluster $C \ni y$).
Therefore, we can extend the definition of portal-respecting path to $\modifydecom$, and define portal-respecting distance $\dportP(x, y)$ on $\modifydecom$ the same way as $\dportH(x, y)$.
On the other hand, we want $h(y)$ to be sufficiently small, so that the error incurred by $\dportP(x, y)$, namely $\dportP(x, y) - \dist(x, y) = O(\epsilon) 2^{h(y)}$, is negligible.

As a first attempt, consider choosing $h(y) = \log(\dist(y, X)/\epsilon)$.
Then we have $\dist(y, \proj{X}{y}) \leq \epsilon 2^{h(y)}$, 
and we can further argue that $y$ is within distance $\epsilon 2^{h(y) + 1}$ to some level $h(y) + 1$ portal.
Therefore, the portal-respecting distance $\dportP(x, y)$ can be defined the same way as $\dportH(x, y)$.
However, this choice of $h(y)$ becomes problematic in terms of the error $\dportP(x, y) - \dist(x, y)$.
Consider an arbitrary point $x \in X$, and assume $h(y)$ is the highest level where $\{x, y\}$ is cut.
Then the error is $\dportP(x, y) - \dist(x, y) = O(\epsilon) 2^{h(y)} = O(\dist(y, X))$, which is too large to afford.

To resolve the issue, we choose $h(y)$ as a slightly smaller value $\log(\dist(y, X)/\sqrt{\epsilon})$.
For this choice, the error becomes $\dportP(x, y) - \dist(x, y) = O(\epsilon) 2^{h(y)} = O(\sqrt{\epsilon}) \dist(y, X)$,
which is at most $O(\sqrt{\epsilon}) \dist(x, y)$ and thus can be charged to $\dist(x, y)$.
The tradeoff is that $y$ is no longer $\epsilon 2^{h(y) + 1}$-covered by the level $h(y) + 1$ portal set.
Instead, the covering radius becomes $\sqrt{\epsilon} 2^{h(y) + 1}$.
Nonetheless, we can still define $\dportP(x, y)$ similarly under this weaker covering property, and it does not change the $O(\sqrt{\epsilon}) \dist(x, y)$ error bound above when the cutting level is exactly $h(y)$.
We obtain a weaker error bound of $\dportP(x, y) - \dist(x, y) = O(\sqrt{\epsilon}) 2^{\ell}$ only if the 
cutting level $\ell$ is strictly greater than $h(y)$.
In this case, the question reduces to finding the highest cutting level of $\{x, \proj{X}{y}\}$, which can be answered fairly well using techniques in~\cite{Cohen-AddadFS21} and the previous section, since both points are in $X$.

\paragraph{Dynamic program.}
Our new hierarchical decomposition $\modifydecom$ and the portal-respecting distance $\dportP(x, y)$ can then be combined with the dynamic program of Cohen-Addad et al. \cite{Cohen-AddadFS21} to get our result in the case when the clients are low-dimensional.
This seems difficult at first, because each cluster on $\modifydecom$ now has an unbounded number of child clusters, mainly due to the newly added ornaments.
A naive enumeration of the configurations of these ornament children would blow up the time complexity of the DP.

Perhaps surprisingly, we show that we can avoid doing the enumeration for ornaments, and it suffices to only enumerate the configurations of non-ornament children, the number of which is bounded by $2^{O(\ddim)}$.
Our key observation is that ornaments must be candidate facilities, and thus only the opening cost needs to be computed.
Therefore, we can obtain from portals of non-ornament children the information which ornaments are potentially required to be opened as a facility.
Once we have this information, we simply select the smallest set of ornaments that serve all the unserved portals.
This can be done via solving a set cover problem with a bounded universe.

\subsubsection{Applications to the $(k,\ell)$-Median Problem under Discrete Fr\'echet Distance}

We next discuss how to apply our results to the $(k,\ell)$-median problem under discrete Fr\'echet distance.
It is a folklore result that the doubling dimension of the metric space $(\mathbb X_z^d,\dist_{dF})$, i.e. 
the space of polygonal curves of $z$ points in $d$-dimensional space equipped with the discrete Fr\'echet distance\footnote{Technically, we consider equivalence classes of curves of pairwise discrete Fr\'echet distance $0$ to obtain a proper metric space.}, is $\Theta(dz)$. This implies that for the $(k,\ell)$-median problem the space of center candidates has bounded doubling dimension (for constant $\ell$ and $d$). However, it also has an infinite number of points. Thus, we have to compute a discrete subset of that space that contains a $(1+\epsilon)$-approximation and is not too large.  In order to do so we make use of a result by Filtser et al. \cite{FFK23} from the context of nearest neighbor search. Applying \Cref{thm:bounded_centers}, we obtain the following result.

\thmklmedian *

\paragraph{Complexity reduction.}
For the case of the ambient space being one-dimensional we develop a new complexity reduction for the discrete Fr\'echet distance that can be summarized as follows. 
This provides an alternative approach that allows to apply the work of Cohen-Addad et al. \cite{Cohen-AddadFS21} in a more direct way for this special case.

We believe that this complexity reduction is of independent interest. For example, it allows us to get an improved coreset construction for clustering under the discrete Fr\'echet distance.

\begin{restatable}{theorem}{thmdimreduction}\label{theorem:dimreduction}
Let $\varepsilon \in (0,1)$ and $\ell \in \NN$ be constants. Given an input time series $x \in \RR^m$ Algorithm \ref{alg: ComplexityReduction} computes in time $O(\comp \ell \log^2 \comp) + O(\ell/\eps)^{2 \comp'}$ with $\comp' \in O(2^{O(\ell/\eps)^{(2\ell+2)}})$   
a time series $z \in \RR^{\comp'})$ s.t. for all $y \in \RR^\ell$,
\begin{equation*}
    (1-\varepsilon) \dist_{dF}(x,y) \leq \dist_{dF}(z,y) \leq (1 + \varepsilon)  \dist_{dF}(x,y).
\end{equation*}
\end{restatable}

Our idea for the complexity reduction can be described as follows.
As a first (simple) step, we show that one can reduce the number of distinct values appearing in a time series of complexity $\comp$ to $O(\ell/\eps)$ while maintaining the distance to any time series of complexity $\ell$ up to a factor of $(1+\eps)$. 

Then we observe that the discrete Fr\'echet distance between two time series $x$ and $y$, each of fixed complexity, can be written as a minimum over all traversals. Our goal is to describe the function minimizing over all traversals with a function minimizing over a much smaller set.
During each traversal, every value $y_i$ of $y$ is matched to a subsequence of $x$. In order to determine the Fr\'echet distance, it suffices to know the minimum and maximum value of $x$ matched to $y_i$. 
As it turns out, each traversal can equivalently (but not uniquely!) be described by remembering a sequence of constraints that consist of the minimum and maximum value matched to each $y_i$. Furthermore, the function minimizing over the set of all possible traversals can likewise be described by minimizing over all possible ordered constraint sets. 
Since we have reduced the number of different values of the time series to $O(\ell/\eps)$, the number of different constraint sets is small.
The set of constraints 
will be called an $\ell$-profile (see Figure \ref{fig: example profile} for an example).
It is important to note that the set of all $\ell$-profiles
of a time series $x$ with values from a fixed set $X$
completely determines the Fr\'echet distance
to any time series of complexity $\ell$.
Since there are only a constant number of different sets of $\ell$-profiles (where the constant depends on $\ell$ and $\eps$) for any time series $x$ with $O(\ell/\eps)$ distinct values, we can replace $x$ by the shortest time series $x'$ over the same set of values that has the same set of $\ell$-profiles. 
The length of the shortest such time series is a constant that depends on the set of profiles and the number of distinct values of the time series. 
Since the number of profiles is also a constant depending on $\eps$ and $\ell$, the maximum length of these shortest time series is constant as well. We can compute such a time series using a dynamic programming approach. The resulting time series have length $O(2^{O(\ell/\eps)^{(2\ell+2)}})$.

\subsection{Further Related Work}

The $k$-median problem in metric spaces is known to be hard to approximate with a factor better than $1+2/e$ unless set cover can be approximated within a factor $c \ln n$ for $c<1$ \cite{JMS02}.
 A number of different constant factor polynomial time approximation algorithms are known \cite{cowen2003constant,jain2001approximation,arya2004local,gupta2008simpler,cohen2022improved,mettu2003online,charikar2012dependent,thorup2005}, and the currently best approximation ratio is $(2+\eps)$ \cite{CLSS25}.
In the Euclidean plane, the problem is NP-hard \cite{megiddo1984complexity}.
The first polynomial time approximation scheme for $k$-median in the Euclidean plane has been developed by Arora et al. \cite{arora1998approximation} and later improved to near-linear time in $\mathbb R^d$ when $d$ is constant \cite{kolliopoulos2007nearly}. This result has been generalized to metric spaces of bounded doubling dimension \cite{thorup2005}
         and later to a near-linear approximation scheme  \cite{Cohen-AddadFS21}.

The metric (uncapacitated) facility location problem can be approximated within a constant factor \cite{STA97,CG99,jain2001approximation, mettu2003online,thorup2005,MYZ06,B10}, and the currently best polynomial time approximation algorithm achieves an approximation guarantee of 1.488 \cite{Li11}. At the same time, there is a conditional lower bound of $1.463$ on the best possible approximation \cite{GK99,JMS02}. 
In the constant-dimensional Euclidean setting, there are similar results as for the $k$-median problem \cite{kolliopoulos2007nearly,Talwar04,Cohen-AddadFS21}.
   
 Driemel et al. \cite{DKS16} defined the $(k,\ell)$-clustering problem for time series as follows: Given a set $P$ of $n$ time series of complexity $m$ and parameters $k,\ell \in \NN$ find $k$ center time series of complexity $\ell$, such that (a) the maximum distance of an element in $P$ to its closest center time series or (b) the sum of these distances is minimized. Variant (a) is referred to  as $(k,\ell)$-center and (b) as $(k,\ell)$-median. Under the continuous Fréchet distance, they developed near-linear time $(1+\eps)$-approximation algorithms for both clustering variants, assuming $\eps,k$, and $\ell$ are constants. They complement these algorithmic results with hardness results, showing that both $(k,\ell)$-median and $(k,\ell)$-center are NP-hard under continuous Fréchet distance. Approximating $(k,\ell)$-median for polygonal curves in arbitrary dimensions was recently studied in \cite{BDR23}. Cheng and Huang give the first $(1+\eps)$-approximation algorithm for $(k,\ell)$-median under continuous Fréchet distance in $d > 1$~\cite{CH23}. Both clustering problems are also NP-hard under the discrete Fréchet distance and even for the case $k = 1$ \cite{BDGHKLS19} \cite{buchin2019hardness}. 
 Buchin et al. developed the first $(1+\eps)$-approximation algorithm for $(k,\ell)$-median under discrete Fréchet distance, which runs in $\tilde{O}((1/\varepsilon k \ell)^{k\ell} \cdot k \ell nm^{(k\ell+1)})$ time  \cite{buchin2019hardness}. Nath and Taylor \cite{nath2021k} improved this to $\tilde{O}(nm2^{O(k/\varepsilon \log(k/\varepsilon))} \cdot (\ell/\varepsilon^2)^{O(kl)})$. 
Buchin and Rohde \cite{buchin2022coresets} designed the first coreset construction for $(k,\ell)$-median under both variants of the Fréchet distance, where the size of the coreset has logarithmic dependence on the number of input curves.
Recently, Cohen-Addad et al. introduced a coreset construction for $(k,\ell)$-median under discrete Fr\'echet distance that has size independent of the number of input curves~\cite{Cohen-Addad2025Coreset}.

Related to our dimension reduction are some data structures for approximate nearest neighbor search under discrete Fr\'echet distance \cite{driemel2019sublinear, filtser2023static,FFK23}. In the asymmetric setting where the query time series has complexity $\ell$, the data structures cited above replace each input time series by a set of lower dimensional time series. This is fundamentally different from our dimension reduction, which replaces each time series with exactly one lower dimensional time series.

\section{Preliminaries}
\label{sec:prelim}

Consider a metric space $(X \cup Y, \dist)$.
For a point $x \in X \cup Y$ and $r > 0$, define the \emph{ball} centered at $x$ with radius $r$ to be $B(x, r) := \{y \in X \cup Y \colon \dist(x, y) \leq r \}$.
The \emph{$r$-neighborhood} of a subset $C \subseteq X \cup Y$ is defined as $B(C, r) := \bigcup_{x \in C} B(x, r)$. 
For a subset $S \subseteq X \cup Y$, denote $B_S(x, r) := B(x, r) \cap S$.
Denote the \emph{diameter} of $S$ to be $\diam(S) := \max_{x, y \in S} \dist(x, y)$.
The \emph{aspect ratio} of $S$ is defined as the ratio between the largest and smallest inter-point distances of $S$, denoted as $\Delta(S) := \frac{\max_{x, y \in S}\dist(x, y)}{\min_{x, y \in S}\dist(x, y)}$.
For a set $S \subseteq X \cup Y$ and a point $x \in X \cup Y$ let $\proj{S}{x}$ be a point in $S$ that is closest to $x$, i.e. $\proj{S}{x} \in \argmin_{y \in S} \dist(x, y)$.
Denote $\dist(x, S) := \dist(x, \proj{S}{x})$ as the \emph{distance} from $x$ to point set $S$.

\begin{definition}[Doubling dimension~\cite{GuptaKL03}]
\label{def:ddim}
The \emph{doubling dimension} of a metric space $(X \cup Y, \dist)$ 
is the smallest $t\ge0$ such that every metric ball can be covered by at most $2^t$ balls of half the radius.
The doubling dimension of a point set $S\subseteq X \cup Y$ is 
the doubling dimension of the metric space $(S, \dist)$,
and is denoted $\ddim(S)$.
\end{definition}

\begin{definition}[Packing, covering and net]
    Consider a metric space $(X \cup Y, \dist)$ and a subset $S \subseteq X \cup Y$.
    For $\rho > 0$, $S$ is \emph{$\rho$-packing} if $\forall u, v \in S$, $\dist(u, v) \geq \rho$.
    $S$ is $\rho$-covering for $X \cup Y$ if for every $x \in X \cup Y$, there exists $u \in S$ such that $\dist(x, u) \leq \rho$.
    $S$ is called a $\rho$-net of $X \cup Y$ if it is both $\rho$-packing and $\rho$-covering for $X \cup Y$.
\end{definition}

\begin{lemma}[Packing Property~\cite{GuptaKL03}]
    \label{lemma:packing}
    If $S$ is $\rho$-packing, then $|S| \leq (\diam(S) / \rho)^{O(\ddim(S))}$.
\end{lemma}

\subsection{Hierarchical Decomposition of Doubling Metrics}

In this section, we review the hierarchical decomposition for doubling metrics, 
first introduced by~\cite{Talwar04}.
Let $(X, \dist)$ be a metric space with doubling dimension $\ddim$.
Without loss of generality, assume the minimum interpoint distance of $X$ is $1$ and that the diameter of $X$ is $\Delta $.
Let $L = \lceil \log \Delta \rceil + 4 = O(\log \Delta)$.
Construct a sequence of \emph{nested nets} on $X$:
\[
X = N_0 \supseteq N_1 \dots \supseteq N_L, 
\]
such that for $0 \leq \ell \leq L$, $N_\ell$ is a $2^{\ell - 2}$-net of $X$.
Specifically, $N_0 = X$ and $N_L$ contains only one point in $X$.
The hierarchical decomposition $\decom$ is constructed in \Cref{alg:talwar_decomposition}.

\begin{algorithm}[!ht]
\caption{Hierarchical decomposition $\decom$~\cite{Talwar04}}
\DontPrintSemicolon
\label{alg:talwar_decomposition}
    \KwIn{finite metric space $(X, \dist)$ with $\ddim(X) \leq \ddim$, parameter $\rho \in (0, \frac{1}{2})$}
    construct the nested net $X = N_0 \supseteq N_1 \dots \supseteq N_L$, such that for $0 \leq \ell \leq L$, 
    $N_\ell$ is a $2^{\ell - 2}$-net of $X$ \;
    sample $\alpha \sim U(\frac{1}{2}, 1)$ \;
    sample a random permutation $\sigma$ over $X$ \;
    let $\decom_L \gets \{X\}$ \;
    \For{$\ell = L - 1, L - 2, \dots, 0$}{
        let $\decom_\ell \gets \emptyset$ \;
        \For{$C \in \decom_{\ell + 1}$ with $|C| \geq 1$}{
            \For{$u \in N_\ell$}{
                let $C_u := (C \cap B_X(u, \alpha 2^\ell) ) \setminus \bigcup_{v \in N_\ell \colon \sigma(v) < \sigma(u)} B_X(v, \alpha 2^\ell)$ 
                be a child cluster of $C$ 
                \label{line:metric_ball} \;
                let $\decom_\ell \gets \decom_\ell \cup \{C_u\}$ \;
            }
        }
    }
    \For{$0 \leq \ell \leq L$}{ 
    \For{$C \in \decom_\ell$}{
        construct a $\rho 2^\ell$-net $P_C$ as the portal set of $C \cap X$ \;
    }
}
    \Return $\decom := \{\decom_0, \decom_1, \dots, \decom_L\}$
\end{algorithm}

\begin{definition}\label{def:cut}
    For a set $T \subseteq X$,
    say $T$ is \emph{cut} at level $\ell$ w.r.t. $\decom$, if there 
    exists a cluster $C \in \decom_\ell$, such that $T \cap C \neq \emptyset$ and $T \setminus C \neq \emptyset$.
\end{definition}

We summarize the properties of $\decom$ below.

\begin{lemma}[Hierarchical decomposition~\cite{Talwar04,Cohen-AddadFS21}]\label{lemma:hierarchical_decomposition}
    Consider a metric space $(X, \dist)$ with $\ddim(X) \leq \ddim$.
    Let $\decom = \{\decom_0, \decom_1, \dots, \decom_L\}$ be the hierarchical decomposition returned by \Cref{alg:talwar_decomposition} when given as input $X$.
    Then $\decom$ satisfies the following properties:
    \begin{enumerate}[label=(\arabic*)]
        \item \label{item:new_decomposition_prop_1} 
        For $0 \leq \ell \leq L$ and $C \in \decom_\ell$, $\diam(C) \leq 2^{\ell + 1}$.
        \item \label{item:new_decomposition_prop_2} 
        Every cluster $C$ has at most $2^{O(\ddim)}$ child clusters.
        \item \label{item:new_decomposition_prop_3} 
        There exists a universal constant $c > 0$, such that for every point set $T \subseteq X$,
        \begin{equation}\label{eqn:cutting_probability}
        \Pr[T \text{ is cut at level } \ell \text{ w.r.t. } \decom] \leq \frac{c \cdot \ddim \cdot \diam(T)}{2^\ell}.    
        \end{equation}
        \item \label{item:portals}
        Portals: For $0 \leq \ell \leq L$, every cluster $C \in \decom_\ell$ comes with a portal set $P_C \subseteq X$, which satisfies
        \begin{enumerate}[label=(\alph*)]
            \item Bounded size: $|P_C| \leq \rho^{-O(\ddim)}$.
            \item Covering: Every $x \in C$ has $\dist(x, P_C) \leq \rho 2^\ell$.
            \item Nested: If $p \in P_{C'} \cap C$ for some $C' \in \modifydecom_{\ell+1}$, then $p \in P_C$.
        \end{enumerate}
    \end{enumerate}

    Furthermore, $\decom$ can be computed in time $\rho^{-O(\ddim)} O(n) \log \Delta$.
\end{lemma}

\begin{remark*}
    In~\cite{Talwar04}, the cutting probability bound is proposed for a pair of points.
    Indeed, it can be generalized to arbitrary subsets of $X$ as \ref{eqn:cutting_probability}.
    Similar results can also be found in e.g.~\cite{AbrahamBN06}.

    \cite{Cohen-AddadFS21} also considered the cutting probability for arbitrary subsets, but
    their bound in~\cite[Lemma 9]{Cohen-AddadFS21} has an exponential dependence in $\ddim$.
    We note that this dependence can be improved to linear.
    This will slightly improve the running time of the final algorithm.

    In our final algorithm, we will set parameter $\rho$ to be $\rho = \epsilon^{10}/\ddim^2$.
\end{remark*}

\begin{proof}
    Properties \ref{item:new_decomposition_prop_1}, \ref{item:new_decomposition_prop_2} and \ref{item:portals} are stated in~\cite[Section 3]{Talwar04}, and the time complexity is given in~\cite[Lemma 9]{Cohen-AddadFS21},
    thus we omit the proof here.
   
    For property \ref{item:new_decomposition_prop_3}, let us first consider the probability that $T$ is cut at level $\ell$ but not at level $\ell + 1$.
    If $\diam(T) > 2^{\ell + 2}$, then by property \ref{item:new_decomposition_prop_1}, $T$ must be cut at level $\ell + 1$, and this probability is $0$.
    We assume $\diam(T) \leq 2^{\ell + 2}$ below.

    Say $T$ is cut by $v \in N_\ell$, if among all net points with $\min_{x \in T} \dist(u, x) \leq \alpha 2^\ell$ and $\max_{y \in T} \dist(u, y) > \alpha 2^\ell$, $v$ is the one with the minimum index $\sigma(v)$.
    Sort all points in $N_\ell \cap B_X(T, 2^\ell)$ as $u_1, u_2, \dots, u_s$, such that $\dist(u_1, T) \leq \dist(u_2, T) \leq \dots \leq \dist(u_s, T)$.
    Note that for $j \in [s]$, 
    $T$ is cut by $u_j$ only if $\min_{x \in T} \dist(u, x) \leq \alpha 2^\ell$ and $\max_{y \in T} \dist(u, y) > \alpha 2^\ell$, and $\forall i < j, \sigma(i) > \sigma(j)$.
    Since $\alpha$ and $\sigma$ are independent, we have
    \begin{align*}
        \Pr[T \text{ is cut by } u_j]
        &\leq \Pr\left[\min_{x \in T} \dist(u, x) \leq \alpha 2^\ell \text{ and } \max_{y \in T} \dist(u, y) > \alpha 2^\ell\right]
        \cdot \Pr[\sigma(i) > \sigma(j), \ \forall i < j] \\
        & \leq \Pr\left[
            \frac{\min_{x \in T} \dist(u, x)}{2^\ell}
            \leq \alpha < 
            \frac{\max_{y \in T} \dist(u, y)}{2^\ell}
        \right]
        \cdot \Pr[\sigma(i) > \sigma(j), \ \forall i < j] \\
        & \leq \frac{\max_{y \in T} \dist(u, y) - \min_{x \in T} \dist(u, x)}{2^{\ell - 1}} \cdot \frac{1}{j} \\
        & \leq \frac{\diam(T)}{2^{\ell - 1}} \cdot \frac{1}{j}.
    \end{align*}
    Therefore, \begin{align*}
        &\qquad \Pr[T \text{ is cut at level } \ell \text{ but not at level } \ell + 1] \\
        &= \sum_{j = 1}^s \Pr[T \text{ is cut by } u_j] 
        \leq \frac{\diam(T)}{2^{\ell - 1}} \cdot \sum_{j = 1}^s \frac{1}{j} 
        \leq \frac{O(\log s) \cdot \diam(T)}{2^\ell} 
        \leq \frac{O(\ddim) \cdot \diam(T)}{2^\ell},
    \end{align*}
    where the last inequality is because $\diam(T) \leq 2^{\ell + 2}$ and thus $s = |N_\ell \cap B_X(T, 2^\ell)| \leq 2^{O(\ddim)}$.
    Finally, we have 
    \begin{align*}
        \Pr[T \text{ is cut at level } \ell]
        &\leq \sum_{i = \ell}^{L - 1}
        \Pr[T \text{ is cut at level } \ell \text{ but not at level } \ell + 1] \\
        &\leq O(\ddim) \cdot \diam(T) \sum_{i = \ell}^{L - 1} 2^{-i}  \\
        &\leq \frac{O(\ddim) \cdot \diam(T)}{2^\ell}.
    \end{align*}
\end{proof}

\subsubsection{Portal-respecting Paths}
\label{sec:portal_respecting_paths_original}
Following~\cite{Cohen-AddadFS21}, we use the notion of \emph{portal-respecting paths}.
For a pair of points $x, y \in X$, the portal-respecting path between $x$ and $y$ is a collection of segments, each of which connects some level $i$ portal to its closest level $i + 1$ portal.
Specifically, let $\ell$ be the highest level where $\{x, y\}$ is cut w.r.t. $\decom$. 
Consider a sequence of clusters $C_0(x), \dots, C_\ell(x)$ and a sequence of portals $p_0(x), \dots, p_\ell(x)$, where $C_i(x) \in \decom_i$ is the level $i$ cluster containing $x$, and $p_i(x) \in P_{C_i(x)}$ is the portal closest to $p_{i-1}(x)$.
In particular, $p_0(x) = x$.
Define $C_i(y)$ and $p_i(y)$ analogously.
The portal-respecting path between $x$ and $y$ is then the sequence of portals $(p_0(x), p_1(x), \dots, p_\ell(x), p_\ell(y), \dots, p_0(y))$.
The portal-respecting distance between $x$ and $y$ is defined as the length of the path:
\begin{equation*}
    \dportH(x, y) := \dist(p_{\ell}(x), p_\ell(y))
    + \sum_{i = 0}^{\ell - 1} \dist(p_i(x), p_{i + 1}(x))
    + \sum_{i = 0}^{\ell - 1} \dist(p_i(y), p_{i + 1}(y)).
\end{equation*}

\begin{lemma}\label{lemma:detour}
    For $x, y \in X$, the portal-respecting distance between $x$ and $y$ satisfies 
    \begin{equation*}
        \dist(x, y) \leq 
        \dportH(x, y) \leq \dist(x, y) + O(\rho 2^\ell),
    \end{equation*}
    where $\ell$ is the highest level where $\{x, y\}$ is cut w.r.t. $\decom$.
\end{lemma}
\begin{proof}
    The lower bound for $\dportH(x, y)$ is straightforward.
    For the upper bound, by definition we have 
    \begin{align*}
        \dportH(x, y) 
        &= \dist(p_{\ell}(x), p_\ell(y))
    + \sum_{i = 0}^{\ell - 1} \dist(p_i(x), p_{i + 1}(x))
    + \sum_{i = 0}^{\ell - 1} \dist(p_i(y), p_{i + 1}(y)) \\
        &\leq \dist(x, y) + 
        2\sum_{i = 0}^{\ell - 1} \dist(p_i(x), p_{i + 1}(x))
    + 2\sum_{i = 0}^{\ell - 1} \dist(p_i(y), p_{i + 1}(y)) \\
    & \leq \dist(x, y) + 4 \sum_{i = 0}^{\ell - 1} O(\rho) 2^i \\
    & \leq \dist(x, y) + O(\rho 2^\ell).
    \end{align*}

\end{proof}

\subsection{Nearest Neighbor Search in Partially Doubling Metrics}
\label{sec:ANN}

In this section we give known results on approximate nearest neighbor search in partially doubling metrics.
Consider a metric space $(X \cup Y, \dist)$, where $X$ is a dataset of size $n$, and the query points come from $Y$.
In \Cref{sec:ANN_low_dim_data}, we consider the setting where only $X$ is doubling.
In \Cref{sec:ANN_low_dim_queries}, we consider the setting where only $Y$ is doubling.

\subsubsection{Low-dimensional Data, High-dimensional Queries}
\label{sec:ANN_low_dim_data}

The first lemma is due to Krauthgamer and Lee \cite{KrauthgamerL04} and gives a result for nearest neighbor search when $X$ is doubling. We state it slightly differently from the original paper so that it fits better to our setting.

\begin{lemma}
\label{lemma:ANN_doubling_data}
    Let $\epsilon \in (0, \frac{1}{2})$ and 
   $(X \cup Y, \dist)$ be a metric space with $|X| = n$ and $\ddim(X) \leq \ddim$ and let the aspect ratio of $X$ be $\Delta$. Then there is an algorithm that
   builds in $\epsilon^{-O(\ddim)} O(n \log \Delta)$ time a data structure that,
   given a query point $y \in Y$, returns a $(1 + \epsilon)$-ANN of $y$ in $X$.
   The query time is $\epsilon^{-O(\ddim)} O(\log \Delta)$.
\end{lemma}

\subsubsection{High-dimensional Data, Low-dimensional Queries}
\label{sec:ANN_low_dim_queries}

For the case where query points come from a subset $Y$ with bounded doubling dimension $\ddim$, we adapt the following ANN data structure proposed by~\cite{Har-PeledK13}, which assumes that one can compute an exact nearest neighbor that lies in sets of bounded doubling dimension in constant time.

\begin{lemma}[{\cite[Theorem 4.2]{Har-PeledK13}}]\label{lemma:ANN_doubling_queries_original}
    Given a metric space $(X \cup Y, \dist)$ with $|X| = n$ and $\ddim(Y) \leq \ddim$ and assume that for $x \in X$ it nearest neighbor in $Y$ can be computed in constant time. Then there exists an algorithm that 
    builds in $\epsilon^{-O(\ddim)} {O}(n \log n)$ time a data structure that,
    given a query point $y \in Y$, returns a $(1 + \epsilon)$-ANN of $y$ in $X$.
    The query time is $\epsilon^{-O(\ddim)} O(\log n)$.
\end{lemma}

We slightly modify Lemma \ref{lemma:ANN_doubling_queries_original}
 by using Lemma \ref{lemma:ANN_doubling_data} to compute an $(1+\eps)$-approximate nearest neighbor in a doubling space instead of an exact one. This only requires minor adjustments in the proof and introduces an additional factor $O(\log \Delta)$ to the construction and query time, which is stated in the following Lemma.  
 \begin{lemma}\label{lemma:ANN_doubling_queries}
    There exists an algorithm that, 
    given a metric space $(X \cup Y, \dist)$ with $|X| = n$ and $\ddim(Y) \leq \ddim$, 
    builds in $\epsilon^{-O(\ddim)} {O}(n \log n \log \Delta)$ time a data structure that,
    given a query point $y \in Y$, returns a $(1 + \epsilon)$-ANN of $y$ in $X$.
    The query time is $\epsilon^{-O(\ddim)} O(\log n \log \Delta)$, where $\Delta$ is the aspect ratio of $Y$.
\end{lemma}

\section{Dimension Reduction Using Proxy Sets}
\label{sec:dim_reduction}

In this section, we develop our dimension reduction techniques for the setting where the candidate center set $Y$ has bounded doubling dimension.
We first explain 
a simple form of dimension reduction that replaces a point in $X$ by a constant-size set $N_x \subseteq Y$ called \emph{proxies} of $x$, such that the distance $\dist(x,y)$ from $x$ to any $y\in Y$ is approximated by $\min_{u\in N_x} \dist(x,u) + \dist(u,y)$. We later argue how this idea can be modified so that it can be incorporated in the algorithm of Cohen-Addad et al. \cite{Cohen-AddadFS21}.
The set of proxies $N_x$ will be defined as 
an $\epsilon\cdot  \dist(x,Y)$-net of a ball $B_Y(\pi_Y(x),  \dist(x,Y) /\epsilon)$, where $\pi_Y(x)$ is the nearest neighbor of $x$ in $Y$. For simplicity of exposition we assume here that 
$\pi_Y(x) \in N_x$. We later present a more general statement that also does not require this condition.

\LemmaDimReduction*

\begin{proof}
    By the triangle inequality we have $\dist(x,y) \le \dist(x,u) + \dist(u,y)$ for every $u\in Y$. Thus, it follows that $\dist(x,y) \le  \min_{u\in N_x} \dist(x,u) + \dist(u,y) = \widehat{\dist}(x,y)$.
    Now consider the case that $y\in B_Y(\pi_Y(x), \dist(x,Y)/\epsilon)$. 
    Define $\hat u \in N_x$ to be the closest point in $N_x$ to $y$.
    In this case, $\dist (\hat u, y) \le \epsilon \cdot \dist(x,Y)
    $ and so
    \[\widehat{\dist}(x,y)= \min_{u\in N_x} d(x,u) + d(u,y) \le \dist(x,\hat u) +\dist(\hat u, y) 
    \le \dist(x,y) + \dist(\hat u, y) + \dist(\hat u,y) \le (1+2\epsilon)\cdot \dist(x,y) . 
    \]
    If $y\notin  B_Y(\pi_Y(x), \dist(x,Y)/\epsilon)$ then we have
\[
    \dist(x,y) \ge \dist(\pi_Y(x), y) - \dist(\pi_Y(x), x) \ge \frac{1-\epsilon}{\epsilon} \cdot \dist( x,Y) 
\]
    and so $\dist(x,Y) \le \epsilon/(1-\epsilon) \cdot \dist(x,y)
    $.
    It follows that
    \[
    \widehat{\dist}(x,y) \le \dist(x,\pi_Y(x)) + \dist(\pi_Y(x),y)
    \le 2\dist(x,\pi_Y(x)) + \dist(x,y) 
    = 2 \cdot \dist(x,Y) + \dist(x,y)
    \le \frac{1+2\epsilon}{1-\epsilon} \cdot \dist(x,y).
    \]
The lemma follows with $\epsilon \in (0,\frac12 )$.
\end{proof}

We would like to incorporate the idea of Lemma \ref{lemma:dim-reduction-simple} into the algorithm of Cohen-Addad et al \cite{Cohen-AddadFS21}. For this purpose, we need a more flexible version of Lemma \ref{lemma:dim-reduction-simple} that allows us to replace $x\in X$ by a net around the nearest points from an arbitrary set $S\subseteq Y$. While this will potentially result in larger errors, it allows us to charge the error to a constant approximation of the facility location problem. Note that 
we do not require $\pi_S(x)$ to be in $N_x$ in the following definition.

\begin{definition}\label{def:dim_reduction_advanced}
 Let $\epsilon \in (0,\frac12 )$ and $S\subseteq Y$. Let
  $N_x$ be an $\epsilon \cdot \dist(x, S)$-net of $B_Y(\pi_S(x),  \dist(x,S)/\epsilon)$, for $x\in X$. 
  For all $x\in X$ and $y\in Y$ we define
  \[
    \widehat{\dist}_S(x,y) = \min_{u\in N_x} \dist(x,u) + \dist(u,y).
  \]
  For a set $F\subseteq Y$ we define
  $
  \widehat{\dist}_S(x,F) = \min_{y\in F} \widehat{\dist}_S(x,y).
  $
  For sets $C\subseteq X$ and $F\subseteq Y$ we define
  $ \widehat{\dist}_S(C,F) = \sum_{x\in C} \widehat{\dist}_S(x,F)$.
\end{definition}

We remark that, strictly speaking, $\widehat{\dist}_S$ depends on the choice of $S$ and the $N_x$. However, the following arguments -- unless stated otherwise -- work for arbitrary choice of $N_x$. Therefore and for readability we ignore this dependence. 
We will generalize our previous lemma to the setting of the above definition.

\begin{lemma}
\label{lemma:dim-reduction-advanced}
Let $\epsilon \in (0,\frac{1}{2})$, $S \subseteq Y$.
  Then we have
  \[
 \dist(x,y) \le \widehat{\dist}_S(x,y) \le (1+4\epsilon) \cdot \dist(x,y) + 2 \epsilon \cdot \dist(x,S). 
  \]
\end{lemma}

\begin{proof}
 Our proof follows the ideas and structure of the proof of Lemma \ref{lemma:dim-reduction-simple}.
 By the triangle inequality we have $\dist(x,y) \le \dist(x,u) + \dist(u,y)$ for every $u\in Y$. Thus, it follows that $\dist(x,y) \le  \min_{u\in N_x} \dist(x,u) + \dist(u,y) = \widehat{\dist}_S(x,y)$.
     Now consider the case that $y\in B_Y(\pi_S(x), \dist(x,S)/\epsilon)$. 
     Let us define $\hat u \in N_x$ to be the
     closest point of $N_x$ to $y$.
     In this case, $\dist(\hat u, y) \le \epsilon \cdot \dist(x,S)
    $ and so
    \[\widehat{\dist}_S(x,y) \le \dist(x,\hat u) + \dist(\hat u, y) 
    \le \dist(x,y) + \dist(\hat u, y) + \dist(\hat u,y) \le  \dist(x,y) + 2\epsilon \dist(x,S). 
    \]
    If $y\notin  B_Y(\pi_S(x), \dist(x,S)/\epsilon)$ then 
    we have
\[
    \dist(x,y) \ge 
    \dist(\pi_S(x), y) - \dist(\pi_S(x), x) \ge \frac{1-\epsilon}{\epsilon} \cdot \dist( x,S) 
\]
    and so $\dist(x,S) \le \epsilon/(1-\epsilon) \cdot \dist(x,y)
    $.
    Now let $\bar \pi_S(x)$ be the closest point in $N_x$ to $\pi_S(x)$.
    It follows that
    \begin{eqnarray*}
    \widehat{\dist}_S(x,y) & \le & \dist(x, \bar \pi_S(x)) + 
    \dist(\bar \pi_S(x),y) \\
    & \le & 
    \dist(x,\pi_S(x)) +
    2\cdot \dist(\bar \pi_S(x), \pi_S(x))+  
    \dist(\pi_S(x),y)\\
    & \le & 2\dist(x,\pi_S(x)) + \dist(x,y)  + 2 \epsilon \dist(x,S)\\\ 
    & = & 2 \cdot \dist(x,S) + \dist(x,y) + 2 \epsilon \dist(x,S)\\
    &\le & (1+ \frac{2\epsilon }{1-\epsilon}) \cdot \dist(x,y) + 2 \epsilon \dist(x,S).
    \end{eqnarray*}
The lemma follows with $\epsilon \in (0,\frac12 )$.
\end{proof}

To see how we can use Lemma \ref{lemma:dim-reduction-advanced} for the clustering problems we consider, let $S\subseteq Y$ be a constant approximation for the facility location problem. 
Then we get the following corollary.

\begin{corollary}
    Let $\epsilon \in (0,\frac 12 )$ and let $S\subseteq Y$ be a $c$-approximation to the facility location problem on instance $(X\cup Y,d)$. Then 
    for any set $F\subseteq Y$ we have
    \[
    \widehat{\dist}_S(X,F) + \sum_{f\in F} \ocost(f) \le (1+6c \epsilon) \cdot \floc(X,F).
    \]
\end{corollary}

\begin{proof}
    We have 
    \begin{eqnarray*}
    \widehat{\dist}_S(X,F)
    & = & \sum_{x\in X} \min_{f\in F}
    \widehat{\dist}_S(x,f) \\
    &\le & \sum_{x\in X} \min_{f\in F}
    (1+4\epsilon) \cdot \dist(x,f) + 2 \epsilon \dist(x,S) \\
    &=& 2\epsilon \sum_{x\in X} \dist(x,S) + (1+4\epsilon) \cdot \sum_{x\in X} \min_{f\in F} \dist(x,f) \\
    &\le & 2c\epsilon \floc(X,F) + (1+4\epsilon) \cdot \sum_{x\in X}
     \dist(x,F) \\
     & \le & (1+6c\epsilon) \cdot \sum_{x\in X} \dist(x,F) + 2c\epsilon \floc(X,F).
     \end{eqnarray*}
Now the corollary follows from the fact that $\floc(X,F) = \sum_{x\in X} \dist(x,F) +\sum_{x\in F} \ocost(f)$.
\end{proof}

\subsection{Integrating Proxy Sets with Portal-respecting Paths}

Next, we show how to integrate our dimension reduction techniques with portal-respecting distance introduced in \Cref{sec:portal_respecting_paths_original}.
Let $\decom$ be the hierarchical decomposition of $Y$ defined in \Cref{lemma:hierarchical_decomposition}.
Recall that $\dportH$ is the portal-respecting distance w.r.t. $\decom$.
We have the following definition that combines \Cref{def:dim_reduction_advanced} with $\dportH$.

\begin{definition}\label{def:connection_cost_portal}
Let $\epsilon \in (0,\frac 12 )$ and $S\subseteq Y$.
For all $x\in X$
Let $N_x$ be an $\epsilon \dist(x,S)$-net of 
$B_Y(\pi_S(x),d(x,S)/\epsilon)$. For $y\in Y$ we define
\[
    \dporthatH(x, y) := 
    \min_{u\in N_x} \dist(x, u) + \dportH(u, y).
\]
For a set $F\subseteq Y$ we define
  $
  \dporthatH(x,F) = \min_{y\in F} \dporthatH(x,y).
  $
  For sets $C\subseteq X$ and $F\subseteq Y$ we define
  $ \dporthatH(C,F) = \sum_{x\in C} \dporthatH(x,F)$.
\end{definition}

Strictly speaking, $\dporthatH$ also depends on a subset $S \subseteq Y$.
Since $S$ is already clear in the context, we omit it and directly write $\dporthatH$ instead of $\widehat{\dist}_{\mathrm{port}}^{\decom, S}$ for ease of notation.

The following lemma bounds $\dporthatH$.

\begin{lemma}\label{lemma:connection_cost_portal}
    Consider a metric space $(X \cup Y, \dist)$ with $\ddim(Y) \leq \ddim$.
    For $x \in X$ and  $y \in Y$, assume that $\ell$ is the highest level where $\{\proj{S}{x}, y\}$ is cut, 
    and that $j$ is the highest level where 
    $B_Y(\proj{S}{x}, \dist(x, S)/\epsilon)$ is cut, then
    \begin{align*}
        \dist(x, y)
        \leq \dporthatH(x, y)
        \leq (1 + 10\epsilon) \dist(x, y) + 2 \epsilon \dist(x, S) + O(\rho) (2^{\ell} + 2^{j}). 
    \end{align*}
\end{lemma}

\begin{proof}
    Consider the following two cases.

    \paragraph{Case 1: $y \in B_Y(\proj{S}{x}, \dist(x, S) /\epsilon)$.}

    Then there exists $u \in N_x$, such that $\dist(u, y) \leq \epsilon \dist(x, S)$.
    Since $u, y \in B_Y(\proj{S}{x}, \dist(x, S)/\epsilon)$,
    $\{u, y\}$ is cut at level at most $j$.
    By \Cref{lemma:detour}, 
    \begin{align*}
        \dportH(u, y) 
        \leq \dist(u, y) + O(\rho) \cdot 2^{j}
        \leq \epsilon \dist(x, S) + O(\rho) \cdot 2^{j}.
    \end{align*}
    Therefore, \begin{align*}
        \dporthatH(x, y) 
        &\leq \dist(x, u) + \dportH(u, y) \\
        &\leq \dist(x, y) + \dist(y, u) + \dportH(u, y)\\
        &\leq \dist(x, y) + 2 \epsilon \dist(x, S) + O(\rho) \cdot 2^{j}.
    \end{align*}

    \paragraph{Case 2: $y \notin B_Y(\proj{S}{x}, \dist(x, S) /\epsilon)$.}

    There exists $v \in N_x$ such that $\dist(v, \proj{S}{x}) \leq \epsilon \dist(x, S)$.
    Hence,
    \begin{align*}
        \dportH(x, y) 
        \leq \dist(x, v) + \dportH(v, y)
        \leq (1 + \epsilon) \dist(x, S) + \dportH(v, y).
    \end{align*}
    
    To upper bound $\dportH(v, y)$, we consider the level where $\{v, y\}$ is cut. 
    $\{v, \proj{S}{x}\}$ is cut at level at most $j$, since $v, \proj{S}{x} \in B_Y(\proj{S}{x}, \dist(x, S)/\epsilon)$.
    On the other hand, $\{\proj{S}{x},y\}$ is cut at level at most $\ell$.
    Therefore, $\{v, y\}$ is cut at level at most $\max\{j, \ell\}$.
    By \Cref{lemma:detour}, 
    \begin{align*}
        \dportH(v, y) 
        &\leq \dist(v, y) + O(\rho) 2^{\max\{j, \ell\}} \\
        &\leq \dist(v, y) + O(\rho) (2^{\ell} + 2^{j})   \\
        & \leq \dist(x, y) + \dist(x, \proj{S}{x}) + \dist(\proj{S}{x}, v) + O(\rho) (2^{\ell} + 2^{j}) \\
        & \leq \dist(x, y) + (1 + \epsilon) \dist(x, S) + O(\rho) (2^{\ell} + 2^{j})
    \end{align*}
    Therefore, 
    \begin{align*}
        \dporthatH(x, y) \leq \dist(x, y) + (2 + 2 \epsilon) \dist(x, S) + O(\rho) (2^{\ell} + 2^{j}).
    \end{align*}

    Finally, recall that 
    \begin{align*}
        \dist(x, S) \leq \epsilon \dist(y, \proj{S}{x}) 
        \leq \epsilon \dist(x, S) + \epsilon \dist(x, y).
    \end{align*}
    Thus, $\dist(x, S) \leq \frac{\epsilon}{1 - \epsilon} \dist(x, y)$.
    We have 
    \begin{align*}
        \dporthatH(x, y) \leq \left(
            1 + \epsilon \frac{2 + 2 \epsilon}{1 - \epsilon}
        \right)\dist(x, y) + O(\rho) (2^{\ell} + 2^{j}).
    \end{align*}

    \medskip Combining the two cases completes the proof.
\end{proof}

\section{A New Metric Decomposition}
\label{sec:metric_decomposition}

In this section, we turn to the other setting where the client set $X$ has bounded doubling dimension while the facility set $Y$ does not.
At first, it may be tempting to apply the dimension reduction techniques in \Cref{sec:dim_reduction} to $Y$, and define a proxy set of every facility $y \in Y$.
Unfortunately, this idea does not work in this setting, mainly due to the asymmetric nature between $X$ and $Y$.
Algorithmically, the solution (facility set) $F \subseteq Y$ is unknown to us in advance.
Thus it is difficult to decide the ``active'' proxy sets which should be used to compute the connection cost.
Moreover, even if we are able to maintain such active proxy sets, it can be tricky to integrate proxy sets for facilities with portal-respecting paths.
Therefore, some new ideas are needed.

We use a different approach in this setting.
Instead of first reducing $Y$ to $X$ and then applying Talwar's metric decomposition on $X$,
we directly construct a hierarchical decomposition for the entire metric space $X \cup Y$.
Compared to Talwar's decomposition for doubling metrics~\cite{Talwar04}, our decomposition has the new feature of handling points in the high-dimensional ambient space $Y$.
We first summarize its key properties in \Cref{lemma:new_decomposition_for_general_metric}, and show how to efficiently construct it based on Talwar's decomposition.
In \Cref{sec:portal_respecting_new}, 
we further discuss how to define portal-respecting paths and portal-respecting distances w.r.t. our new decomposition.

\begin{restatable}{lemma}{LemmaNewDecomposition}
\label{lemma:new_decomposition_for_general_metric}
    Given a metric $(X \cup Y, \dist)$ with $\ddim(X) \leq \ddim$ and parameter $\rho \in (0, 1/2)$, one can compute in $\rho^{-O(\ddim)} O(n + m) \log \Delta$ time a random hierarchical decomposition $\modifydecom = \{\modifydecom_0, \modifydecom_1, \dots, \modifydecom_{L = O(\log \Delta)}\}$ for $X \cup Y$, satisfying
    \begin{enumerate}[label=(\arabic*)]
        \item \label{it:bounded_diam}
        Bounded diameter: For $0 \leq \ell \leq L$, every level $\ell$ cluster $C \in \modifydecom_\ell$ has diameter $O(2^\ell)$.
        \item \label{it:ornaments}
        Ornaments: $\{y\}$ is a leaf node at level $\min\{L - 1, \lceil \log(\dist(y, X)/\sqrt{\rho}) \rceil\}$, for every point $y \in Y \setminus X$.
        Moreover, $\proj{X}{y}$ is contained in one of the siblings of $\{y\}$.
        Cluster $\{y\}$ is called an \emph{ornament}.
        \item \label{it:nested}
        Nested: For every cluster $C \in \modifydecom_\ell$, 
        let $\Schild(C)$ and $\Nchild(C)$ be the set of $C$'s ornament child clusters and non-ornament child clusters, respectively.
        Then \[
        C = \bigcup_{D \in \Nchild(C)} D \cup \bigcup_{\{y\} \in \Schild(C)}\{y\},
        \]
        i.e., $C$ is the union of its child clusters.
        Furthermore, $|\Nchild(C)| \leq 2^{O(\ddim)}$.
        \item \label{it:cutting}
        Cutting probability: There exists a universal constant $c > 0$ s.t. for every subset $T \subseteq X$, 
        \[
            \Pr[T \text{ is cut at level } \ell] \leq \frac{c\cdot \ddim \cdot \diam(T)}{2^\ell}.
        \]

        \item \label{it:portals}
        Portals: For $0 \leq \ell \leq L$, every non-ornament cluster $C \in \modifydecom_\ell$ comes with a portal set $P_C \subseteq X$, which satisfies
        \begin{enumerate}[label=(\alph*)]
            \item Bounded size: $|P_C| \leq \rho^{-O(\ddim)}$.
            \item Covering: When $\ell \leq L - 1$, every $x \in C \cap X$ has $\dist(x, P_C) \leq \rho 2^\ell$; every $y \in C \setminus X$ has $\dist(y, P_C) \leq 2 \sqrt{\rho} 2^\ell$.
            \item Nested: If $p \in P_{C'} \cap C$ for some $C' \in \modifydecom_{\ell+1}$, then $p \in P_C$.
        \end{enumerate}
    \end{enumerate}
\end{restatable}       

Every cluster $C$ on our new decomposition $\modifydecom$ is a subset of $X \cup Y$.
Properties \ref{it:bounded_diam}, \ref{it:nested} and \ref{it:cutting} are inherited from Talwar's original decomposition (\Cref{lemma:hierarchical_decomposition}).
Property \ref{it:ornaments} is related to points in the high-dimensional ambient space $Y$, and becomes the key feature of our decomposition.
Specifically, every point $y \in Y \setminus X$ will be attached to the tree as a leaf node at a level that depends on $\dist(y, X)$.
This is different from Talwar's decomposition, where every point in $X$ becomes a leaf node at level $0$.
Denote $h(y)$ as the level where $y$ is a leaf node on $\modifydecom$.
Then \begin{align}
    \forall y \in X \cup Y, \qquad 
    h(y) = 
        \min\{L - 1, \lceil \log(\dist(y, X)/\sqrt{\rho}) \rceil\}.
        \label{eqn:leaf_level}
\end{align}
Note that $h(y) = 0$ when $y \in X$, which is consistent with Talwar's decomposition.

Property \ref{it:portals} defines portals for each $C \in \modifydecom_\ell$, which are useful for us to further define portal-respecting paths on $\modifydecom$.
Roughly speaking, the portal set $P_C$ mimics the $\rho 2^\ell$-net on $C$.
It is $\rho 2^\ell$-packing for $C$, $\rho 2^\ell$-covering for $C \cap X$, but only $2 \sqrt{\rho} 2^\ell$-covering for $C \setminus X$.
This weaker covering property turns out sufficient for bounding the error of portal-respecting distances in \Cref{sec:portal_respecting_new}.
See \Cref{fig:new_decomposition} for an illustration of our decomposition.

\begin{figure}[!ht]
    \centering
    \includegraphics[width=\textwidth]{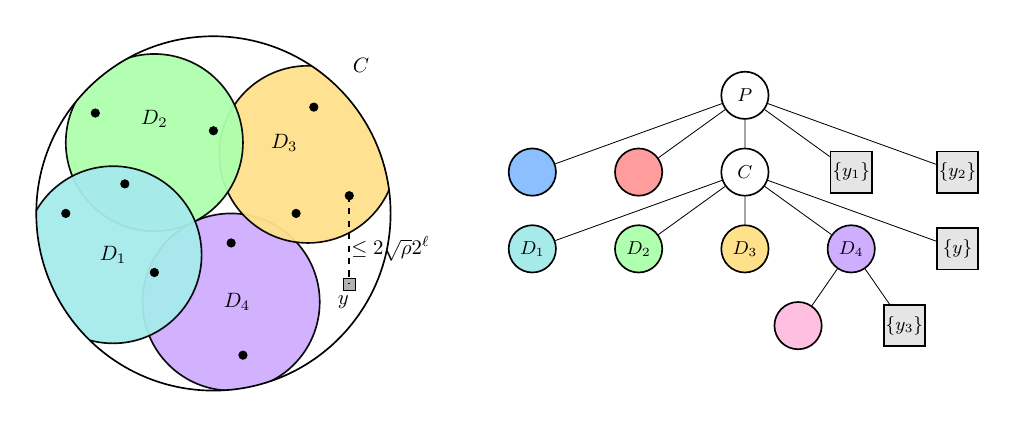}
    \caption{An illustration of the new decomposition $\modifydecom$.
    $C$ is a level $\ell$ cluster on $\modifydecom$, which has non-ornament child clusters $D_1, D_2, D_3, D_4$ and an ornament child $\{y\}$, where $y \in Y \setminus X$.
    The left figure shows how $C$ is partitioned.
    Black points are portals of $C$, and $\dist(y, P_C) \leq 2 \sqrt{\rho} 2^\ell$.
    The right figure shows part of the decomposition $\modifydecom$, where circles are non-ornament nodes, and squares are ornaments.
    }
    \label{fig:new_decomposition}
\end{figure}

\paragraph{Construction of $\modifydecom$.}
Assume wlog the minimum interpoint distance of $X \cup Y$ is $1$ and that $\diam(X \cup Y) = \Delta$.
The construction of our new decomposition $\modifydecom$ is formally given in \Cref{alg:modified_decomposition}.

\begin{algorithm}
\caption{New decomposition $\modifydecom$}
\label{alg:modified_decomposition}
\DontPrintSemicolon
\KwIn{finite metric space $(X \cup Y, \dist)$, 
with $|X| = n, |Y| = m$ and $\ddim(X) \leq \ddim$, 
and 
a parameter $\rho \in (0, 1/2)$}
run \Cref{alg:talwar_decomposition} on $(X, \dist)$ to obtain the hierarchical decomposition $\decom = \{\decom_0, \decom_1, \dots, \decom_L\}$ \label{line:compute_H}\;
let $\modifydecom_\ell \gets \decom_\ell$ for $0 \leq \ell \leq L$ \label{line:init_P} \;
\For{$y \in Y \setminus X$}{ \label{line:handle_ambient_points}
    compute 
    $\proj{X}{y} \in X$ \label{line:compute_proj_for_y}\;
    let $h(y) \gets \min\{L - 1, \lceil \log(\dist(y, \proj{X}{y})/\sqrt{\rho}) \rceil\}$ \label{line:h_y}\;
    $\modifydecom_{h(y)} \gets \modifydecom_{h(y)} \cup \{\{y\}\}$ \label{line:hang}\;
    find the level $h(y)$ cluster $C$ that contains $\proj{X}{y}$, let $\{y\}$ be a sibling of $C$, 
    and add $y$ to the ancestors of $C$ (excluding $C$) \label{line:add_y_to_ancestor} \;
}
\For{$0 \leq \ell \leq L$}{ \label{line:construct_portal}
    \For{$C \in \modifydecom_\ell$}{
        $C$ inherits the portal set $P_C$ from $\decom$
        \label{line:construct_net} \;
    }
}
\Return $\modifydecom := \{\modifydecom_0, \modifydecom_1, \dots, \modifydecom_L\}$
\end{algorithm}

\Cref{alg:modified_decomposition} has three stages.
In stage 1 (Lines \ref{line:compute_H}-\ref{line:init_P}), it runs \Cref{alg:talwar_decomposition} to compute the hierarchical decomposition $\decom$ for $X$, and initializes $\modifydecom$ as a copy of $\decom$.
In stage 2 (Lines \ref{line:handle_ambient_points}-\ref{line:add_y_to_ancestor}), it adds points in the ambient space $Y \setminus X$ to the decomposition.
Specifically, for every $y \in Y \setminus X$, it first finds the point $\proj{X}{y} \in X$ closest to $y$ (Line \ref{line:compute_proj_for_y}), 
and decides the level $h(y)$ where the ornament $\{y\}$ should be a leaf node, based on the distance $\dist(y, \proj{X}{y})$ (Line \ref{line:h_y}).
Ornament $\{y\}$ is then attached to level $h(y)$ as a sibling of the cluster that contains $\proj{X}{y}$ (Lines \ref{line:hang}-\ref{line:add_y_to_ancestor}).
Finally, it adds $y$ to all ancestors (Line \ref{line:add_y_to_ancestor}), which guarantees that the new decomposition $\modifydecom$ is still nested.
In stage 3 (Lines \ref{line:construct_portal}-\ref{line:construct_net}), our algorithm constructs the portal set for each cluster $C$ by directly inheriting $P_C$ from $\decom$.
By construction, $\modifydecom$ is a random decomposition, whose randomness solely comes from $\decom$.

\begin{remark*}
    Line \ref{line:compute_proj_for_y} computes the nearest neighbor of $y$ in $X$, which requires $O(n)$ time for each $y$.
    In our implementation, the time complexity can be reduced via $(1 + \epsilon)$-ANN search.
    Moreover, replacing $\proj{X}{y}$ with a $(1 + \epsilon)$-ANN in \Cref{lemma:new_decomposition_for_general_metric,alg:modified_decomposition} will not affect the correctness of our subsequent analysis, 
    and only introduces an extra $(1 + \epsilon)$ factor to our final approximation ratio.
\end{remark*}

The notion of being cut (\Cref{def:cut}) can be adapted to $\modifydecom$ as well.
Formally, say a set $T \subseteq X \cup Y$ is \emph{cut} at level $\ell$ w.r.t. $\modifydecom$, if there 
exists a cluster $C \in \modifydecom_\ell$, such that $T \cap C \neq \emptyset$ and $T \setminus C \neq \emptyset$.
In fact, we will mainly use this notion of cut w.r.t. $\modifydecom$ for $T \subseteq X$ in our analysis.

We next show that the decomposition $\modifydecom$ computed by \Cref{alg:modified_decomposition} satisfies the properties in \Cref{lemma:new_decomposition_for_general_metric}.

\begin{proof}[Proof of \Cref{lemma:new_decomposition_for_general_metric}]

Properties \ref{it:ornaments}, \ref{it:nested} and \ref{it:cutting} follow  from the construction and \Cref{lemma:hierarchical_decomposition}.
We focus on properties \ref{it:bounded_diam} and \ref{it:portals} below.

\paragraph{Property \ref{it:bounded_diam}: bounded diameter.}
The bound is trivial for ornaments.
Consider a non-ornament cluster $C \in \modifydecom_\ell$.
By construction, there exists a corresponding cluster $\widehat{C} \in \decom_\ell$, such that $C \supseteq \widehat{C}$ and $C \setminus \widehat{C} \subseteq Y \setminus X$.
For every $y \in C \setminus \widehat{C}$, by Line \ref{line:add_y_to_ancestor} of \Cref{alg:modified_decomposition} we have $\ell > h(y) = \min\{L - 1, \lceil \log(\dist(y, \proj{X}{y})/\sqrt{\rho}) \rceil\}$.

If $\ell > L - 1$ then $\ell = L$.
We have $\diam(C) = \diam(X \cup Y) \leq 2^L$.

If $\ell > \lceil \log(\dist(y, \proj{X}{y})/\sqrt{\rho}) \rceil$, then $\dist(y, \proj{X}{y}) \leq \sqrt{\rho} 2^\ell$.
By Line \ref{line:add_y_to_ancestor}, we have $\proj{X}{y} \in C$ as well.
Hence, for every $y, y' \in C$, 
\begin{align*}
    \dist(y, y') 
    &\leq \dist(y, \proj{X}{y}) + \dist(\proj{X}{y}, \proj{X}{y'}) + \dist(y', \proj{X}{y'}) \\
    &\leq \sqrt{\rho} 2^\ell + \diam(\widehat{C}) + \sqrt{\rho} 2^\ell &&\text{Since $\proj{X}{y}, \proj{X}{y'} \in \widehat{C}$}\\
    &\leq \sqrt{\rho} 2^{\ell+1} + O(2^\ell) &&\text{By \Cref{lemma:hierarchical_decomposition}} \\
    &\leq O(2^\ell).
\end{align*}
Therefore, $\diam(C) \leq O(2^\ell)$.

\paragraph{Property \ref{it:portals}: portals.}
By Line \ref{line:construct_net} of \Cref{alg:modified_decomposition}, $P_C$ is inherited from $\decom$.
Therefore, the size bound and nested property go through.
We focus on the covering property below.

For $x \in C \cap X$, the bound follows from the standard property of nets.
For $y \in C \setminus X$, by construction we have 
\[
L - 1 \geq \ell > h(y) = \min\{L - 1, \lceil \log(\dist(y, \proj{X}{y})/\sqrt{\rho}) \rceil\}.
\]
Thus, $\dist(y, \proj{X}{y}) \leq \sqrt{\rho} 2^\ell$.
Furthermore, $\proj{X}{y} \in C \cap X$.
Therefore, $\dist(\proj{X}{y}, P_C) \leq \rho 2^\ell$.
We thus have 
\[
    \dist(y, P_C) \leq \dist(y, \proj{X}{y}) + \dist(\proj{X}{y}, P_C) \leq \sqrt{\rho} 2^\ell + \rho 2^\ell
    \leq 2 \sqrt{\rho} 2^\ell.
\]

\paragraph{Time complexity.}
For stage 1 of \Cref{alg:modified_decomposition}, by \Cref{lemma:hierarchical_decomposition}, the decomposition $\decom$ can be computed in $\rho^{-O(\ddim)} O(n \log \Delta)$ time on top of $X$.

For stage 2, as discussed above, $\proj{X}{y}$ will be replaced by a $(1 + \epsilon)$-ANN of $y$ in $X$ (Line \ref{line:compute_proj_for_y}), which can be found in $\epsilon^{-O(\ddim)} O(\log \Delta)$ time by \Cref{lemma:ANN_doubling_data}.
The algorithm then calculates $h(y)$ and attaches $\{y\}$ to level $h(y)$, both of which can be done in constant time.
To add $y$ to its ancestor clusters on $\modifydecom$,
note that $\{y\}$ has at most $L = O(\log \Delta)$ ancestors, so this step can be done in $O(\log \Delta)$ time.
Therefore, stage 2 has complexity $\epsilon^{-O(\ddim)} O(\log \Delta)$ for a single $y \in Y \setminus X$, and a total of $\epsilon^{-O(\ddim)} O(m \log \Delta)$.

For stage 3, the portal set $P_C$ for each cluster $C$ is inherited from $H$.
Since $|P_C| \leq \rho^{-O(\ddim)}$, 
the total complexity of stage 3 is $\sum_{\ell=0}^L \sum_{C \in \modifydecom_\ell} \rho^{-O(\ddim)} = \rho^{-O(\ddim)} O(n \log \Delta)$.

In conclusion, the time complexity of \Cref{alg:modified_decomposition} is $\rho^{-O(\ddim)} O(n + m) \log \Delta$.

\end{proof}

\subsection{Portal-respecting Paths on $\modifydecom$}
\label{sec:portal_respecting_new}
Consider metric space $(X \cup Y, \dist)$ with $\ddim(X) \leq \ddim$.
Let $\modifydecom$ be the decomposition computed by \Cref{alg:modified_decomposition}, with parameter $\rho = \epsilon^{10}/\ddim^2$.

Consider a pair of points $x \in X$ and $y \in Y$.
It will be helpful to think of $x$ as a client and $y$ as the facility to which $x$ is assigned in a certain solution.
Let $0 \leq \ell \leq L - 1$ be the highest level where the set $\{x, y\}$ is cut w.r.t. $\modifydecom$.
Consider a sequence of clusters $C_0(x), \dots, C_\ell(x)$ and a sequence of portals $p_0(x), p_1(x), \dots, p_\ell(x)$, such that $C_i(x) \in \modifydecom_i$ is the level $i$ cluster containing $x$, and $p_i(x) \in P_{C_i(x)}$ is the portal closest to $p_{i - 1}(x)$.
In particular, $p_0(x) = x$.
The clusters and portals for $y$ are defined in an analogous but slightly different way.
Let $j = h(y) \leq \ell$ be the level where $y$ is a leaf node.
(If $y \in X$, then $j = 0$.)
Let $D_{j + 1}(y), \dots, D_\ell(y)$ be a sequence of clusters and $p_{j + 1}(y), \dots, p_\ell(y)$ be a sequence of portals, such that 
$D_i(y) \in \modifydecom_i$ is the level $i$ cluster containing $y$, and $p_i(y) \in P_{D_i(y)}$ is the portal closest to $p_{i - 1}(y)$.
Define $D_j(y) = \{y\}$ and $p_j(y) = y$.
The portal-respecting path between $x$ and $y$ is thus defined as \[
\Big(x = p_0(x), p_1(x), \dots, p_\ell(x), p_\ell(y), \dots, p_j(y) = y \Big), 
\]
and the portal-respecting distance between $x$ and $y$ is defined as the length of the path:
\begin{align*}
    \dportP(x, y) := \dist(p_{\ell}(x), p_\ell(y))
    + \sum_{i = 0}^{\ell - 1} \dist(p_i(x), p_{i + 1}(x))
    + \sum_{i = j}^{\ell - 1} \dist(p_i(y), p_{i + 1}(y)).
\end{align*}

We note that above definition of portal-respecting path is analogous to that in \Cref{sec:portal_respecting_paths_original} and~\cite{Cohen-AddadFS21}, with the only difference that the path does not necessarily start from the $0$-th level.
This misalignment mainly comes from the construction of our modified decomposition $\modifydecom$, where a point $y$ in the ambient space can be a leaf node at a high level $j > 0$.
Also, all points except one endpoint of the path are net points (portals) in $X$.
We have the following bound on $\dportP(x, y)$, which depends on both the highest level where $\{x, y\}$ 
is cut and the level where $y$ is a leaf node.

\begin{lemma}\label{lemma:detour_complicated}
    For $x \in X$ and $y \in Y$,
    the portal-respecting distance between $x$ and $y$ satisfies 
    \begin{equation*}
        \dist(x, y) \leq 
        \dportP(x, y) \leq 
        \begin{cases}
            \dist(x, y) + O(\sqrt{\rho}) 2^\ell, &\text{if } h(y) < \ell; \\
            \dist(x, y) + O(\rho) 2^\ell, & \text{if } h(y) = \ell,
        \end{cases}
    \end{equation*}
    where $0 \leq \ell \leq L - 1$ is the highest level where $\{x, y\}$ is cut w.r.t. $\modifydecom$,
    and $h(y)$ is defined in \eqref{eqn:leaf_level}.
    
\end{lemma}

We note that \Cref{lemma:detour_complicated} is similar but weaker than \Cref{lemma:detour}.
If $h(y) = \ell$, then the bound is the same as \Cref{lemma:detour}.
If $h(y) < \ell$, we lose a $\sqrt{\rho}$ factor in the additive error term.
This is mainly due to the weaker covering property in \ref{it:portals} of \Cref{lemma:new_decomposition_for_general_metric}.
More concretely, it is only guaranteed that $y$ can be connected to a nearby portal within distance $O(\sqrt{\rho}) 2^{h(y) + 1}$ rather than $O(\rho) 2^{h(y) + 1}$.
We give the proof below.

\begin{proof}[Proof of \Cref{lemma:detour_complicated}]
    The lower bound for $\dportP(x, y)$ is straightforward.
    For the upper bound, by definition we have 
    \begin{align*}
        \dportP(x, y) 
        &= \dist(p_{\ell}(x), p_\ell(y))
    + \sum_{i = 0}^{\ell - 1} \dist(p_i(x), p_{i + 1}(x))
    + \sum_{i = h(y)}^{\ell - 1} \dist(p_i(y), p_{i + 1}(y)) \\
        &\leq \dist(x, y) + 
        2\sum_{i = 0}^{\ell - 1} \dist(p_i(x), p_{i + 1}(x))
    + 2\sum_{i = h(y)}^{\ell - 1} \dist(p_i(y), p_{i + 1}(y)).
    \end{align*}

    It is easy to show that the first summation is at most $O(\rho) 2^\ell$.
    One should be more careful about the second summation.
    If $h(y) = \ell$, then the second summation does not exist, and we have $\dportP(x, y) \leq \dist(x, y) + O(\rho) 2^\ell$.

    If $h(y) < \ell$, then $h(y) < L - 1$.
    By property \ref{it:portals} of \Cref{lemma:new_decomposition_for_general_metric}, we have  
    $\dist(y, p_{h(y) + 1}(y)) \leq 2\sqrt{\rho} 2^{h(y) + 1}$.
    Therefore, 
    \begin{align*}
        \sum_{i = h(y)}^{\ell - 1} \dist(p_i(y), p_{i + 1}(y))
        &= \dist(y, p_{h(y) + 1}(y)) + \sum_{i = h(y) + 1}^{\ell - 1} \dist(p_i(y), p_{i + 1}(y)) \\
        &\leq 2\sqrt{\rho} 2^{h(y) + 1} + \sum_{i = h(y) + 1}^{\ell - 1} O(\rho) 2^i \\
        &\leq O(\sqrt{\rho}) 2^{h(y)} + O(\rho) 2^\ell \\
        &\leq O(\sqrt{\rho}) 2^\ell.
    \end{align*}
\end{proof}

The following lemma is a corollary of \Cref{lemma:detour_complicated}.
It shows that the error incurred by portal-respecting distance (i.e., $\dportP(x, y) - \dist(x, y)$) is essentially determined by the highest level where $\{x, \proj{X}{y}\}$ is cut w.r.t. $\modifydecom$, 
where $\proj{X}{y}$ is the closest point to $y$ in $X$.
This allows us to utilize the cutting probability (Property \ref{it:cutting} of \Cref{lemma:new_decomposition_for_general_metric}) on $\modifydecom$ to bound the error. 

\begin{lemma}\label{lemma:detour_by_projection}
    Let $x \in X$, $y \in Y$ and $0 \leq \ell \leq L - 1$ be the highest level where $\{x, \proj{X}{y}\}$ is cut w.r.t. $\modifydecom$.
    Then 
    \[
    \dist(x, y) \leq \dportP(x, y) \leq (1 + O(\sqrt{\rho}))\dist(x, y) + O(\sqrt{\rho}) 2^\ell.
    \]
\end{lemma}

\begin{proof}
    If $y \in X$, then the proof is the same as \Cref{lemma:detour}.
    We thus assume $y \in Y \setminus X$.
    Recall that $h(y)$ is the level where $y$ is a leaf node.
    By Property \ref{it:ornaments} of \Cref{lemma:new_decomposition_for_general_metric}, $\{y, \proj{X}{y}\}$ is cut at level at most $h(y)$.
    Consider the following two cases.

    If $\ell > h(y)$, then at level $\ell + 1$ of $\modifydecom$, the three points $x, \proj{X}{y}, y$ are in the same cluster.
    Hence, $\ell$ is the highest level where $\{x, y\}$ is cut.
    By \Cref{lemma:detour_complicated},
    \begin{equation}\label{eqn:l_greater_than_h}
        \dportP(x, y) \leq \dist(x, y) + O(\sqrt{\rho}) 2^\ell.
    \end{equation}

    If $\ell \leq h(y)$, then at level $h(y) + 1$ of $\modifydecom$, the three points $x, \proj{X}{y}, y$ are in the same cluster.
    Hence, $h(y)$ is the highest level where $\{x, y\}$ is cut.
    By \Cref{lemma:detour_complicated},
    \begin{align}
        \dportP(x, y) 
        &\leq \dist(x, y) + O(\rho) 2^{h(y)} \notag \\
        &\leq \dist(x, y) + O(\rho) \frac{\dist(y, X)}{\sqrt{\rho}} && \text{By \Cref{lemma:new_decomposition_for_general_metric}} \notag \\
        &\leq (1 + O(\sqrt{\rho})) \dist(x, y). \label{eqn:l_smaller_than_h}
    \end{align}

    Combining \eqref{eqn:l_greater_than_h} with \eqref{eqn:l_smaller_than_h} completes the proof.
\end{proof}

\section{Facility Location with Low-dimensional Centers}
\label{sec:fl_bounded_centers}

The first problem we consider is facility location when the set $Y$ has doubling dimension at most $\ddim$.
Our main result in this section is the following.

\begin{theorem}\label{thm:FL_bounded_centers}
    There is a randomized algorithm that, given as input $\epsilon \in (0, \tfrac{1}{2})$, $n, m \in \NN$ and  $(X \cup Y, \dist)$ with $|X| = n, |Y| = m, \ddim(Y) \leq \ddim$, computes a $(1 + \epsilon)$-approximation of facility location in time 
    $2^{2^t} \cdot \tilde{O}(n + m)$ 
    with constant success probability, where
    \begin{equation*}
        t \in O\left(\ddim \log\frac{\ddim}{\epsilon}\right).
    \end{equation*}
\end{theorem}

\subsection{Structural Lemmas}

Our algorithm is based on our dimension reduction techniques in \Cref{sec:dim_reduction}.
Specifically, every client $x \in X$ will be represented by a proxy set $N_x$, which is an $\epsilon \dist(x, S)$-net on  $B_Y(\proj{S}{x}, \dist(x, S)/\epsilon)$, 
where $S \subseteq Y$ will be chosen as a constant approximate solution for facility location. 
Moreover, we integrate our dimension reduction with portal-respecting distances.
For a facility set $F \subseteq Y$, we use $\dporthatH(x, F)$ defined in \Cref{def:connection_cost_portal} to replace the metric distance $\dist(x, F)$, where the scaling parameter $\rho$ of portals is set to $\rho = \epsilon^{10}/\ddim^2$.

Recall in \Cref{lemma:connection_cost_portal}, the error incurred by replacing $\dist(x, F)$ with $\dporthatH(x, F)$ depends on the highest level where the two sets $B_Y(\proj{S}{x}, \dist(x, S)/\epsilon)$ and $\{\proj{S}{x}, \proj{F}{x}\}$ are cut.
For the first set, we show in \Cref{lemma:fl_bounded_centers_new_instance} that we can upper bound the cutting level by preprocessing the dataset.
For the second set, we show in \Cref{lemma:fl_bounded_centers_good_solution} the existence of a good solution $F \subseteq Y$, such that $\{\proj{S}{x}, \proj{F}{x}\}$ is cut at a bounded level.
Both of these structural results rely on the notion of badly-cut sets and bad points, which are first introduced in~\cite{Cohen-AddadFS21}.
We further adapt these notions to our setting.

\begin{definition}[Badly cut]\label{def:badly_cut_old}
Let $\decom$ be a random hierarchical decomposition as in \Cref{lemma:hierarchical_decomposition}.
Say a set $T \subseteq Y$ is badly cut w.r.t. $\decom$ if $T$ is cut at level $\log \diam(T) + \epsddim$ on $\decom$.
\end{definition}

The next lemma follows immediately from Equation \eqref{eqn:cutting_probability} of Lemma \ref{lemma:hierarchical_decomposition}, which bounds the probability of a set being badly cut.

\begin{lemma}\label{lemma:badly_cut}
    For every $T \subseteq Y$, 
    \begin{align*}
        \Pr[T \text{ is badly cut}]  \le c \cdot \epsilon, 
    \end{align*}
    where $c$ is the same constant as in \eqref{eqn:cutting_probability}.
\end{lemma}

\begin{definition}[Bad points]\label{def:bad_points_bounded_centers}
Consider a metric space $(X \cup Y, \dist)$ and let $\decom$ be the decomposition for $Y$ in \Cref{lemma:hierarchical_decomposition}.
Let $F^*, S \subseteq Y$.
\begin{enumerate}[label=(\arabic*)]
    \item A point $x \in X$ is called a \emph{bad client} (w.r.t. $S$ and $\decom$) if $B_Y(\proj{S}{x}, \dist(x,S)/\epsilon)$ is badly cut w.r.t. $\decom$.
    Denote the set of bad clients as $\Bad_X \subseteq X$.
    \item A point $y \in S$ is called a \emph{bad facility} (w.r.t. $F^*$ and $\decom$) if $B_Y(y, 10 \dist(y, F^*))$ is badly cut. 
    Denote the set of bad facilities as $\Bad_S \subseteq S$.
\end{enumerate}
\end{definition}

In our subsequent analysis, we instantiate $F^* \subseteq Y$ as an optimal solution for facility location and $S \subseteq Y$ as a constant approximate solution,
and define bad clients $\Bad_X \subseteq X$ and bad facilities $\Bad_S \subseteq S$.

\Cref{def:bad_points_bounded_centers} is similar to~\cite[Definition 10]{Cohen-AddadFS21} expect for one major difference.
In~\cite{Cohen-AddadFS21}, bad clients are defined w.r.t. a ball around $x$ (i.e., $B(x, \dist(x, S)/\epsilon)$).
However, since in our setting $x$ does not necessarily lie in the doubling subset, we change the center of the ball to $\proj{S}{x}$ which is in $Y$.
This new definition also integrates well with the proxy set of $x$ defined in \Cref{sec:dim_reduction}.

As mentioned earlier, to bound the error $\dporthatH(x, F) - \dist(x, F)$ we need to bound the cutting level of $B_Y(\proj{S}{x}, \dist(x, S)/\epsilon)$ for every $x \in X$.
By \Cref{def:bad_points_bounded_centers}, this level is at most $\log(\dist(x, S)/\epsilon) + \epsddim$ if $x$ is not a bad client.
It then remains to handle the bad clients, and the idea is to ``eliminate'' them by moving every $x \in \Bad_X$ to $\proj{S}{x}$.
This movement creates a new dataset $X'$.
The next lemma shows that such movement changes any clustering cost by at most $\epsilon \floc(X, S)$.
This is essentially the same as~\cite[Lemma 12]{Cohen-AddadFS21}.
We provide the proof below for completeness.

\begin{lemma}[New instance]\label{lemma:fl_bounded_centers_new_instance}
    Given $(X \cup Y, \dist)$, with $\ddim(Y) \leq \ddim$, and a solution $S \subseteq Y$ for facility location, 
    construct a new 
    (multi-)set of clients $X' \subseteq X \cup Y$ as $X' := \phi(X)$ for 
    \begin{align*}
        \phi(x) := \begin{cases}
            x, &\text{if } x \notin \Bad_X; \\
            \proj{S}{x}, & \text{if } x \in \Bad_X,
        \end{cases}
    \end{align*} 
    namely, $X'$ is constructed from $X$ by moving every bad client $x$ to $\proj{S}{x}$.
    Then with probability $0.99$,
    \begin{align*}
        \forall F \subseteq Y, \qquad \Big|\floc(X, F) - \floc(X', F)\Big|
        \leq \epsilon \floc(X, S).
    \end{align*}
\end{lemma}

\begin{proof}
    By construction, 
    for every $x \in X$, either $x$ is not a bad client or $x = \proj{S}{x}$.
    To bound the change in clustering cost, 
    note that 
    \begin{align*}
        \E \left[\sum_{x \in X} \dist(x, \phi(x)) \right] 
        &= \sum_{x \in X} \E[\dist(x, \phi(x))] \\
        & = \sum_{x \in X} \dist(x, \proj{S}{x}) \cdot \Pr[x \text{ is a bad client}] \\
        & \leq O(\epsilon) \sum_{x \in X} \dist(x, S) && \text{By \Cref{lemma:badly_cut}}\\
        & \leq O(\epsilon) \floc(X, S)
    \end{align*}
    By Markov's inequality, with probability $0.99$, 
    \[
    \sum_{x \in X} \dist(x, \phi(x)) \leq O(\epsilon) \floc(X, S).
    \]

    Conditioning on this, consider an arbitrary 
    facility set $F \subseteq Y$.
    We have
    \begin{align*}
        \Big|\floc(X, F) - \floc(X', F)\Big|
        &= \left|\sum_{x \in X} (\dist(x, F) - \dist(\phi(x), F))\right| \\
        &\leq \sum_{x \in X} |\dist(x, F) - \dist(\phi(x), F)| \\
        &\leq \sum_{x \in X} \dist(x, \phi(x)) \\
        &\leq \epsilon \floc(X, S).
    \end{align*}
\end{proof}

We remark that, if $S$ is a constant approximation, then we have \[\floc(X', F) \in (1 \pm O(\epsilon)) \floc(X, F)\] for arbitrary solution $F \subseteq Y$ with probability $0.99$.

In the following lemma, we show the existence of a facility set $F \subseteq Y$, such that the portal-respecting cost of $X$ w.r.t. $F$ is bounded by $(1 + \epsilon) \optFL(X, Y) + \epsilon \floc(X, S)$.
This allows us to compute the solution with minimum portal-respecting cost in our final algorithm.
The lemma can be viewed as an analogy to~\cite[Lemmas 13 and 14]{Cohen-AddadFS21},
with an essential difference that we use the portal-respecting distance defined in \Cref{def:connection_cost_portal}, which utilizes our dimension reduction and proxy sets for clients.
The validity of such replacement is mainly based on \Cref{lemma:connection_cost_portal}. 
It also makes the proof more challenging.

\begin{lemma}[Good portal-respecting solution]\label{lemma:fl_bounded_centers_good_solution}
    Let $(X \cup Y, \dist)$ with $\ddim(Y) \leq \ddim$ be a facility location instance, $S \subseteq Y$ be a solution and 
    $\decom$ be the hierarchical decomposition on $Y$ in \Cref{lemma:hierarchical_decomposition} with portal parameter $\rho = \epsilon^{10}/\ddim^2$.
    Let $X' \subseteq X$ be the new client set constructed by \Cref{lemma:fl_bounded_centers_new_instance}.
    For $x \in X'$ and $F \subseteq Y$, define the portal-respecting connection cost of $x$ as
    \begin{align*}
        \ccostport(x, F) := \begin{cases}
            \dportH(x, F), & \text{if } x = \proj{S}{x}; \\
            \dporthatH(x, F), &\text{if } x \neq \proj{S}{x}.
        \end{cases}
    \end{align*}
    Then with probability $0.9$, there exists a solution $F \subseteq Y$, such that 
    \begin{equation}\label{eqn:fl_bounded_centers_good_solution}
        \sum_{x \in X'} \ccostport(x, F) + \ocost(F) \leq (1 + \epsilon) \optFL(X, Y)
        + \epsilon \floc(X, S).
    \end{equation}
\end{lemma}

Let us briefly explain the definition of $\ccostport(x, F)$.
The case $x \neq \proj{S}{x}$ implies that $x$ is \emph{originally} not a bad client in $X$.
Then the set $B_Y(\proj{S}{x}, \dist(x, S)/\epsilon)$ is cut at a bounded level of $\decom$.
We can therefore use $\dporthatH(x, F)$ for approximation, and by \Cref{lemma:connection_cost_portal}, it suffices to only bound the highest level where $\{\proj{S}{x}, \proj{F}{x}\}$ is cut.
However, if $x$ is originally a bad client (which corresponds to $x = \proj{S}{x}$ in \Cref{lemma:fl_bounded_centers_good_solution}), then we can no longer use $\dporthatH(x, F)$ for approximation, because we now have no control on the highest level where $B_Y(\proj{S}{x}, \dist(x, S)/\epsilon)$ is cut.
Fortunately, in this case $x$ and $\proj{S}{x}$ are the same in $X'$, so we can directly use the original portal-respecting distance $\dportH(x, F)$ for approximation.
By \Cref{lemma:detour}, we still only need to bound the highest level where $\{\proj{S}{x}, \proj{F}{x}\}$ is cut.

\begin{proof}[Proof \Cref{lemma:fl_bounded_centers_good_solution}]
    Recall that a facility $f \in S$ is called a bad facility if $B_Y(f, 10 \dist(f, F^*))$ is badly cut.
    Let $\Bad_S \subseteq S$ be the set of bad facilities in $S$.
    Define 
    \[F := F^* \cup \Bad_S.\]
    We show that the facility set $F$ satisfies \eqref{eqn:fl_bounded_centers_good_solution}.
    Recall that $X' = \phi(X)$ by \Cref{lemma:fl_bounded_centers_new_instance}, thus we can rewrite the LHS of \eqref{eqn:fl_bounded_centers_good_solution} as 
    \[
    \sum_{x \in X} \ccostport(\phi(x), F) + \ocost(F).
    \]

    We first bound for each $x \in X$ the highest level where $\{\proj{S}{\phi(x)}, \proj{F}{\phi(x)}\}$ is cut.
    Specifically, 
    let $\ell(x)$ be the highest level where $\{\proj{S}{\phi(x)}, \proj{F}{\phi(x)}\}$ is cut;
    we show that 
    \[\ell(x) \leq \log(\dist(\phi(x), S)/\epsilon + 10 \dist(\phi(x), F^*)) + \epsddim.
    \]
    Consider the following cases.

    \paragraph{Case 1: $x$ is a bad client, and $\proj{S}{x}$ is a bad facility.}
    By the construction of $X'$, $\phi(x) = \proj{S}{x}$. 
    Since $\proj{S}{x}$ is a bad facility, it is in $\Bad_S$, and thus in $F$.
    Therefore, $\proj{S}{\phi(x)} = \phi(x) = \proj{F}{\phi(x)}$, and $\{\proj{S}{\phi(x)}, \proj{F}{\phi(x)}\}$ is never cut.

    \paragraph{Case 2: $x$ is a bad client, and $\proj{S}{x}$ is not a bad facility.}
    By the construction of $X'$, $\phi(x) = \proj{S}{x}$. 
    Since $\proj{S}{x}$ is not a bad facility, the ball 
    \[B_Y(\proj{S}{x}, 10 \dist(\proj{S}{x}, F^*)) = B_Y(\phi(x), 10 \dist(\phi(x), F^*))\] 
    is cut at level at most $\log(10 \dist(\phi(x), F^*)) + \epsddim$.
    Note that $\proj{F}{\phi(x)} \in B_Y(\phi(x), 10 \dist(\phi(x), F^*))$, and $\proj{S}{\phi(x)} = \phi(x)$.
    Hence, $\{\proj{S}{\phi(x)}, \proj{F}{\phi(x)}\}$ is cut at level at most $\ell(x) \leq \log(10 \dist(\phi(x), F^*)) + \epsddim$.

    \paragraph{Case 3: $x$ is not a bad client, and $\proj{S}{x}$ is a bad facility.}
    In this case $\phi(x) = x$, and the ball $B_Y(\proj{S}{x}, \dist(x,S) / \epsilon)$ is cut at level at most $\log(\dist(x,S)/\epsilon) + \epsddim$.
    On the other hand, since $\proj{S}{x}$ is a bad facility, we have $\proj{S}{x} \in F$.
    Therefore, 
    \begin{align*}
        \dist(\proj{S}{x}, \proj{F}{x}) \leq \dist(x, \proj{S}{x}) + \dist(x, \proj{F}{x}) \leq 2 \dist(x,S),
    \end{align*}
    which implies $\proj{F}{x} \in B_Y(\proj{S}{x}, \dist(x,S) / \epsilon)$.
    Therefore, $\{\proj{S}{x}, \proj{F}{x}\}$ is cut at level at most $\ell(x) \leq \log(\dist(x,S)/\epsilon) + \epsddim = \log(\dist(\phi(x), S)/\epsilon) + \epsddim$.

    \paragraph{Case 4: $x$ is not a bad client, and $\proj{S}{x}$ is not a bad facility.}
    In this case $\phi(x) = x$, and the ball $B_Y(\proj{S}{x}, \dist(x,S) / \epsilon)$ is cut at level at most $\log(\dist(x,S)/\epsilon) + \epsddim$, 
    and the ball $B_Y(\proj{S}{x}, 10 \dist(\proj{S}{x}, F^*))$ is cut at level at most $\log(10 \dist(\proj{S}{x}, F^*)) + \epsddim$.
    Consider the following sub-cases.

    If $\dist(x, F^*) \leq \dist(x,S)/(2\epsilon)$, then 
    \begin{align*}
        \dist(\proj{S}{x}, \proj{F}{x}) \leq \dist(x, \proj{S}{x}) + \dist(x, \proj{F}{x}) 
        \leq \dist(x,S) + \dist(x, F^*) 
        \leq \dist(x,S)/\epsilon.
    \end{align*}
    Hence, $\proj{F}{x} \in B_Y(\proj{S}{x}, \dist(x,S) / \epsilon)$. Thus $\{\proj{S}{x}, \proj{F}{x}\}$ is cut at level at most $\log(\dist(x,S)/\epsilon) + \epsddim = \log(\dist(\phi(x),S)/\epsilon) + \epsddim$.

    If $\dist(x, F^*) > \dist(x,S)/(2\epsilon)$, first note that
    \begin{align*}
        \dist(x, F^*) \leq \dist(\proj{S}{x}, F^*) + \dist(x, \proj{S}{x})
        = \dist(\proj{S}{x}, F^*) + \dist(x,S)
        \leq \dist(\proj{S}{x}, F^*) + 2\epsilon \dist(x, F^*),
    \end{align*}
    which implies $\dist(x, F^*) \leq \tfrac{1}{1 - 2 \epsilon} \dist(\proj{S}{x}, F^*) \leq 2 \dist(\proj{S}{x}, F^*)$.
    Then 
    \begin{align*}
        \dist(\proj{S}{x}, \proj{F}{x}) 
        &\leq \dist(x, \proj{S}{x}) + \dist(x, \proj{F}{x}) \\
        &\leq \dist(x,S) + \dist(x, F^*) \\
        &\leq (1 + 2 \epsilon) \dist(x, F^*) && \text{Since $\dist(x,S) < 2\epsilon \dist(x, F^*)$} \\
        &\leq 2 (1 + 2 \epsilon) \dist(\proj{S}{x}, F^*).
    \end{align*}
    This implies $\proj{F}{x} \in B_Y(\proj{S}{x}, 10 \dist(\proj{S}{x}, F^*))$.
    Therefore, $\{\proj{S}{x}, \proj{F}{x}\}$ is cut at level at most $\log(10 \dist(\proj{S}{x}, F^*)) + \epsddim \leq \log(10 \dist(\phi(x), F^*) + 10 \dist(\phi(x), S)) + \epsddim$.

    \medskip

    In conclusion, $\{\proj{S}{\phi(x)}, \proj{F}{\phi(x)}\}$ is cut at level at most $\ell(x) \leq \log(\dist(\phi(x),S)/\epsilon + 10 \dist(\phi(x), F^*)) + \epsddim$.

    To further prove \eqref{eqn:fl_bounded_centers_good_solution}, we first show that $\floc(X, F) \leq \optFL(X, Y) + O(\epsilon) \floc(X, S)$ with probability $0.99$.
    First note that 
    \begin{align*}
        \floc(X, F) 
        &= \sum_{x \in X} \dist(x, F) + \ocost(F) \\
        & \leq \sum_{x \in X} \dist(x, F^*) + \ocost(F^*) + \sum_{f \in \Bad_S} \ocost(f) \\
        &= \optFL(X, Y) + \sum_{f \in \Bad_S} \ocost(f).
    \end{align*}
    Then 
    \begin{align*}
        0 \leq \floc(X, F) - \optFL(X, Y) \leq \sum_{f \in \Bad_S} \ocost(f).
    \end{align*}
    Taking expectation, we have
    \begin{align*}
        \E[\floc(X, F) - \optFL(X, Y)]
        &\leq \E\left[\sum_{f \in \Bad_S} \ocost(f)\right] \\
        &= \sum_{f \in S} \ocost(f) \cdot \Pr[f \text{ is a bad facility}] \\
        &\leq O(\epsilon) \sum_{f \in S} \ocost(f) && \text{By \Cref{lemma:badly_cut}} \\
        & \leq O(\epsilon) \floc(X, S).
    \end{align*}
    Applying Markov's inequality to the (non-negative) random variable $\floc(X, F) - \optFL(X, Y)$, with probability $0.99$, $\floc(X, F) \leq \optFL(X, Y) + O(\epsilon) \floc(X, S)$.

    Next we bound $\ccostport(\phi(x), F)$ for every $x \in X$.
    If $\phi(x) = \proj{S}{\phi(x)}$, since $\ell(x)$ is the highest level where $\{\proj{S}{\phi(x)}, \proj{F}{\phi(x)}\}$ is cut, 
    by \Cref{lemma:detour} we have 
    \[\ccostport(\phi(x), F) = \dportH(\phi(x), F) \leq \dist(\phi(x), F) + O(\rho) 2^{\ell(x)}.\]
    If $\phi(x) \neq \proj{S}{\phi(x)}$, then $x$ is not a bad client.
    By \Cref{lemma:connection_cost_portal}, we have  
    \begin{align*}
        \ccostport(\phi(x), F) 
        &= \dporthatH(\phi(x), F) 
    \leq \dporthatH(\phi(x), \proj{F}{\phi(x)}) \\
    &\leq (1 + O(\epsilon))\dist(\phi(x), F) + O(\rho) 2^{\ell(x)} + O(\rho) \expepsddim \cdot \frac{\dist(\phi(x), S)}{\epsilon}.
    \end{align*}
   For $\rho = \epsilon^{10}/\ddim^2$ we can bound the LHS of \eqref{eqn:fl_bounded_centers_good_solution} as 
    \begin{align*}
        &\qquad \sum_{x \in X} \ccostport(\phi(x), F) + \ocost(F) \\
        &\leq \sum_{\substack{x \in X \\ \phi(x) = \proj{S}{\phi(x)}}} (
            \dist(\phi(x), F) + O(\rho) 2^{\ell(x)}
        ) \\
        & \qquad + \sum_{\substack{x \in X \\ \phi(x) \neq \proj{S}{\phi(x)}}} 
        \left(
            (1 + O(\epsilon)) \dist(\phi(x), F) + O(\rho) 2^{\ell(x)} + O(\rho) \frac{\ddim \dist(\phi(x), S)}{\epsilon^2}
        \right)
        + \ocost(F)  \\
        & \leq (1 + O(\epsilon)) \floc(\phi(X), F) + O(\rho) \expepsddim \sum_{x \in X} \left(
            \frac{2 \dist(\phi(x), S)}{\epsilon} + 10 \dist(\phi(x), F^*) 
        \right)\\
        & \leq (1 + O(\epsilon)) \floc(X', F) + 
        O(\epsilon) \floc(X', S) 
        + O(\epsilon) \floc(X', F^*)  \\
        & \leq (1 + O(\epsilon)) \floc(X, F) + 
        O(\epsilon) \floc(X, S) 
        + O(\epsilon) \floc(X, F^*) + O(\epsilon) \floc(X, S) &&\text{\Cref{lemma:fl_bounded_centers_new_instance}} \\
        &\leq (1 + O(\epsilon)) \floc(X, F) + O(\epsilon) \floc(X, S) \\
        & \leq (1 + O(\epsilon)) \optFL(X, Y) + O(\epsilon) \floc(X, S) &&\text{w.p. } 0.99.
    \end{align*}

    We conclude that with probability $0.9$, 
    \[
    \sum_{x \in X'} \ccostport(x, F) + \ocost(F) \leq (1 + O(\epsilon)) \optFL(X, Y) + O(\epsilon) \floc(X, S).
    \]
    Rescaling $\epsilon$ completes the proof.
\end{proof}

\subsection{The Algorithm}
\label{sec:fl_alg}
Our algorithm is based on the dynamic programming framework proposed by~\cite{Cohen-AddadFS21}.
Given an instance $(X, Y)$ with $\ddim(Y)$ bounded, the algorithm constructs the hierarchical decomposition $\decom$ on top of $Y$.
It then modifies the instance according to \Cref{lemma:fl_bounded_centers_new_instance}.
The dynamic program is run on top of the modified instance, with respect to the portal-respecting connection cost $\ccostport(x, F)$ defined in \Cref{lemma:fl_bounded_centers_good_solution}.
The correctness of the algorithm is proven in \Cref{sec:fl_alg_correctness}, which is based on \Cref{lemma:fl_bounded_centers_new_instance,lemma:fl_bounded_centers_good_solution}.
We analyze the time complexity in \Cref{sec:fl_alg_time}.
We describe the algorithm as follows.

\paragraph{Preprocessing stage.}
Given as input $(X, Y)$, the algorithm first applies the techniques in \Cref{sec: constant-approx-algos} to computes a $2^{O(\ddim)}$-approximation solution $S \subseteq Y$ for facility location.
At first, this approximation might be too large, as it requires us to rescale the precision parameter $\epsilon$ by a factor of $1/2^{O(\ddim)}$, which may blow up the time complexity of our algorithm.
We ignore this potential issue when describing our algorithm at this point, and discuss how we can fix it in \Cref{sec:fl_alg_correctness}.

Our algorithm then constructs the hierarchical decomposition $\decom$ together with portals on top of $Y$.
The scaling parameter of portals is set to be $\rho = \epsilon^{10}/\ddim^2$.
For every $x \in X$, the algorithm checks if $B_Y(\proj{S}{x}, \dist(x,S) / \epsilon)$ is badly cut (i.e., if $x$ is a bad client).
It moves every bad client $x \in X$ to $\proj{S}{x}$, creating a new instance $(X', Y)$.\footnote{In fact, this step can be efficiently done by moving every $x$ to its $(1 + \epsilon)$-ANN in $S$.
This replacement does not affect the correctness of our previous analysis, and only enlarge the approximation ratio by a $(1 + \epsilon)$ factor.
}

For every $f \in S$, the algorithm decides how many copies of $f$ exist in $X'$; the value is denoted by $w(f) \in \NN$.
For $x \in X'$, it computes 
\begin{align*}
    j(x) := \begin{cases}
        0, &\text{if } x = \proj{S}{x}; \\
        \log(\dist(x,S)/\epsilon) + \epsddim + 1, & \text{otherwise}.
    \end{cases}
\end{align*}
By \Cref{lemma:fl_bounded_centers_new_instance}, $B_Y(\proj{S}{x}, \dist(x,S) / \epsilon)$ is entirely contained in some cluster $C \in \decom_{j(x)}$.
Such $C$, denoted by $C(x) \in \decom_{j(x)}$, is called the cluster where $x$ is \emph{revealed}.
(In particular, if $x = \proj{S}{x}$, then $C(x) = \{x\} \in \decom_0$.)
Roughly speaking, if $j(x) > 0$, then the connection cost of client $x$ will not be computed at level $0$.
Instead, we will ``defer'' the computation of $\ccostport(x, F)$ to a higher level cluster $C(x) \in \decom_{j(x)}$.
Note that every $x$ is revealed in exactly one cluster on $\decom$.
Compute $C(x)$ for every $x \in X'$.

The algorithm also computes the proxies $N_x$ for every $x \in X'$.
This can be done by computing $N_x = B_Y(\proj{S}{x}, \dist(x,S) / \epsilon) \cap P_{C(x)}$.
For  $\rho = \epsilon^{10}/\ddim^2$, $N_x$ is indeed an $\rho 2^{j(x)} \leq \epsilon \dist(x,S)$-net of $B_Y(\proj{S}{x}, \dist(x,S) / \epsilon)$.

\paragraph{Dynamic Program.}
Each table entry of the dynamic program is represented by a cluster $C$ on the hierarchical decomposition, together with a configuration
\[
\veca_C = \braket{a_C^1, a_C^2, \dots, a_C^{|P_C|}}, \quad \vecb_C = \braket{b_C^1, b_C^2, \dots, b_C^{|P_C|}}.
\]
Roughly speaking, the configuration encodes the positional information of the current facility set $F$.
Specifically, for every portal $p$ of cluster $C$,
$a_C^p$ encodes the portal-respecting distance from $p$ to the closest facility inside $C$;
$b_C^p$ encodes the portal-respecting distance from  $p$ to the closest facility outside $C$.
The value stored in entry $(C, \veca_C, \vecb_C)$
is the minimum \emph{revealed} cost within cluster $C$.
Formally,
\begin{align}
    g(C, \veca_C, \vecb_C)
    := \min_{\substack{
    F \colon \text{$F$ is consistent} \\
    \text{with the configuration}}
    }
    \left\{\sum_{x \in X' \colon C(x) \subseteq C}
    \ccostport(x, F)
    + \ocost(F \cap C)\right\}.
    \label{eqn:entry_value}
\end{align}

\paragraph{Base case.}
The base case of the dynamic program corresponds to the leaf nodes of $\decom$.
Consider a leaf node $C = \{y\}$ and the corresponding configuration $a, b$.
(We can wlog assume $\{x\}$ has itself as its only portal.)
It is easy to check consistency for the configuration, since there is either a facility on $y$ (then $a = 0$) or no facility on $y$ (then $a = \infty$).
Furthermore, the opening cost inside $\{y\}$ is either $\ocost(y)$ or $0$, which can also be easily computed.

The connection cost of $\{y\}$ is non-zero iff $y \in X'$.
This can be checked easily, since in the preprocessing stage, we have already marked the number of copies of $y$ in $X'$, denoted by $w(y)$.
For each of these copies, its connection cost can be read from the configuration, specifically, $\ccostport(y, F) = \min\{a, b\}$.
Therefore, the revealed cost in $\{y\}$ is 
\[
g\left(\{y\}, a, b\right)
= w(y) \cdot \ccostport(y, F) + \ocost(F \cap \{y\}).
\]

\paragraph{Updating the DP table.}
Now consider a higher level cluster $C$ and a configuration \sloppy $\veca_C = \braket{a_C^1, \dots, a_C^{|P_C|}}, \vecb_C = \braket{b_C^1, \dots, b_C^{|P_C|}}$.
To compute the value of the entry  $(C, \veca_C, \vecb_C)$, the algorithm first computes the minimum cost inherited from all its child clusters.
Specifically, it enumerates all combinations of configurations for the child clusters.
For each of these combinations
\begin{equation}\label{eqn:config_combination}
    \Big\{\Big(D, \veca_D = \braket{a^1_D, \dots, a^{|P_D|}_D}, \vecb_D = \braket{b^1_D, \dots, b^{|P_D|}_D}\Big)\Big\}_{D \in \child(C)},
\end{equation}
our algorithm checks if it is consistent with $(C, \veca_C, \vecb_C)$. 
Following~\cite{Cohen-AddadFS21}, we use the following criteria:
\begin{enumerate}[label=(\alph*)]
    \item \label{it:criteria1}
    Every portal of $C$ connects to a facility inside through a portal of its child cluster.
    Formally, for every $p \in P_C$, there exists $D \in \child(C)$ and $q \in P_D$, such that $a_C^p = a^q_D + \dist(p, q)$.
    \item \label{it:criteria2}
    Every portal of $D$ connects to a facility outside through a portal of either of its parent cluster or its sibling cluster.
    Formally, for every $D \in \child(C)$ and $q \in P_D$, either there exists a portal $p \in P_C$ such that $b^q_D = b_C^p + \dist(p, q)$, or there exists a cluster $D' \in \child(C)$ and a portal $q' \in P_{D'}$ such that $b^q_D = a^{q'}_{D'} + \dist(q, q')$.
\end{enumerate}
For every combination of configurations that is consistent, compute the summation of clustering cost over $D \in \child(C)$.
The inherited cost of $C$ is the minimum summation, i.e., 
\begin{align}
    \min_{\substack{
    \{(D, \veca_D, \vecb_D)\}_{D \in \child(C)} \\
    \text{ is consistent}
    }}
    \sum_{D \in \child(C)} g(D, \veca_D, \vecb_D).
    \label{eqn:inherited_cost}
\end{align}

Besides the inherited cost, the algorithm also computes for $C$ the connection cost of $x$ for all $x$ that is revealed in $C$, which is called the \emph{newly revealed cost} in $C$, i.e., 
\[\sum_{x \in X' \colon C(x) = C} \ccostport(x, F).\]
Consider an arbitrary $x \in X'$ with $C(x) = C$.
Wlog, assume $x \neq \proj{S}{x}$, since otherwise $x$ is already decided in level $0$ (i.e., the base case).
Recall the definition of $\ccostport(x, F)$ in \Cref{lemma:fl_bounded_centers_good_solution,def:connection_cost_portal}:
\begin{align*}
        \ccostport(x, F) := 
        \dporthatH(x, F)
        = \min_{f \in F} \min_{u \in N_x} \dist(x, u) + \dportH(u, f)
        = \min_{u \in N_x} \dist(x, u) + \dportH(u, F).
    \end{align*}
To compute $\ccostport(x, F)$, 
it suffices to compute $\dist(x, u)$ and $\dportH(u, F)$ for every $u \in N_x$.
The former can be easily computed.
For the latter, 
by construction $N_x = B_Y(\proj{S}{x}, \dist(x,S) / \epsilon) \cap P_C$ is a subset of $P_C$.
Hence $u$ is a portal of $C$; therefore $\dportH(u, F)$ can be directly obtained from the configuration of $C$, which is $\min\{a_u, b_u\}$.

We conclude that each entry in the DP table can be updated by
\begin{align*}
    g(C, \veca_C, \vecb_C) 
    = \min_{\substack{
    \{(D, \veca_D, \vecb_D)\}_{D \in \child(C)} \\
    \text{ is consistent}
    }}
    \sum_{D \in \child(C)} g(D, \veca_D, \vecb_D) 
    + \sum_{x \in X' \colon C(x) = C} \ccostport(x, F).
\end{align*}

\paragraph{Reducing the number of configurations.}
For every $\ell \in [0, \log \Delta]$ and a level $\ell$ cluster $C \in \decom_\ell$, the number of configurations of $C$ can be infinite. 
To resolve this issue, we use two tricks following~\cite{Cohen-AddadFS21}.
First, we restrict $\{a_C^p\}, \{b_C^p\}$ to be multiplications of $\rho 2^{\ell}$.
This discretization incurs a $\pm \rho 2^\ell$ additive error in the portal-respecting distance $\dportH(p, F)$ for every portal $p$, but the error can be charged to the detour between $p$ and $F$.

The second trick is to restrict each $a_C^p, b_C^p$ within a reasonable range $a_C^p, b_C^p \in [0, 2^\ell/\epsilon]$.
The restriction for $a_C^p$ is straightforward, since the distance from a portal $p \in P_C$ to a facility inside $C$ is at most $O(2^\ell) \leq 2^\ell/\epsilon$.
Regarding the restriction for $b_C^p$, note that if there exists $b_C^p > 2^\ell/\epsilon$, then it falls in either of the following two cases:
\begin{enumerate}[label=(\alph*)]
    \item $C \cap F \neq \emptyset$.
    In this case, the closest facility of each portal $p \in P_C$ is inside $C$, which means the actual value of $b_C^p$ will never be used for computing connection cost.
    Therefore, it is safe to reduce the configuration to $\veca_C, \vecb_C = \braket{\infty, \dots, \infty}$.
    \item $C \cap F = \emptyset$.
    In this case, we claim that there is no facility in $B_Y(C, 2^\ell/\epsilon)$.
    Moreover, up to losing a $(1 + \epsilon)$ multiplicative factor, we can wlog assume every point $\{x \in X' \colon C(x) \subseteq C\}$ is assigned to the same facility in $F$.
    Therefore, we can treat the whole cluster $C$ as a single point with weight $|\{x \in X' \colon C(x) \subseteq C\}|$, and decide its connection cost at a higher level.
    An extra boolean flag can be used in the configuration to indicate that $C$ is such a ``compressed'' cluster.
\end{enumerate}

In conclusion, up to small error, we can restrict $a_C^p, b_C^p$ to be multiplications of $\rho 2^\ell$ in range $[0, 2^\ell/\epsilon]$.
The number of configurations of $C$ thus can be bounded by 
\begin{align*}
    (\epsilon \rho)^{-2 |P_C|} = 
    \left(\frac{\ddim}{\epsilon}\right)^{\left(\frac{\ddim}{\epsilon}\right)^{O(\ddim)}}.
\end{align*}

\subsubsection{Proof of Correctness}
\label{sec:fl_alg_correctness}

Let $\widehat{F} \subseteq Y$ be the set of facilities returned by our algorithm in \Cref{sec:fl_alg}.
Our plan is to show that $\floc(X, \widehat{F}) \leq (1 + \epsilon) \optFL(X, Y)$.
We first show the following weaker bound for $\floc(X, \widehat{F})$.

\begin{lemma}\label{lemma:fl_bounded_centers_correctness_weaker}
    Let $\widehat{F} \subseteq Y$ be the set of facilities returned by the algorithm in \Cref{sec:fl_alg} w.r.t. a solution $S \subseteq Y$.
    Then with constant probability, 
    \begin{align*}
        \floc(X, \widehat{F}) \leq (1 + \epsilon) \optFL(X, Y) + 2 \epsilon \floc(X, S).
    \end{align*}
\end{lemma}

\begin{proof}
    By the algorithm in \Cref{sec:fl_alg}, 
    $\widehat{F}$ minimizes the facility location cost w.r.t. $\ccostport(\cdot, \cdot)$, i.e., 
    \begin{align*}
        \widehat{F} = \argmin_{F \subseteq Y}
        \sum_{x \in X'} \ccostport(x, F) + \ocost(F).
    \end{align*}
    We note that the definition of $X'$ and $\ccostport$ also depend on $S$.

    By \Cref{lemma:fl_bounded_centers_good_solution}, 
    with probability $0.9$, there exists a solution $F \subseteq Y$, such that 
    \begin{equation}\label{eqn:apply_fl_bounded_centers_good_solution}
        \sum_{x \in X'} \ccostport(x, F) + \ocost(F) \leq (1 + \epsilon) \optFL(X, Y) + \epsilon \floc(X, S).
    \end{equation}
    Therefore, 
    \begin{align*}
        \floc(X, \widehat{F}) 
        &\leq \floc(X', \widehat{F}) + \epsilon \floc(X, S) && \text{By \Cref{lemma:fl_bounded_centers_new_instance}, w.p. $0.99$} \notag \\
        & \leq \sum_{x \in X'} \ccostport(x, \widehat{F}) + \ocost(F)
        + \epsilon \floc(X, S) \\
        & \leq \sum_{x \in X'} \ccostport(x, F) + \ocost(F)
        + \epsilon \floc(X, S) &&\text{By definition of $\widehat{F}$}\\
        &\leq (1 + \epsilon) \optFL(X, Y) 
        + 2 \epsilon \floc(X, S) && \text{By \eqref{eqn:apply_fl_bounded_centers_good_solution}}.
    \end{align*}
\end{proof}

Unfortunately, \Cref{lemma:fl_bounded_centers_correctness_weaker} does not imply that $\floc(X, \widehat{F}) \leq (1 + \epsilon) \optFL(X, Y)$.
Recall that in \Cref{sec:fl_alg}, we are only able to obtain a $2^{O(\ddim)}$-approximate solution $S$.
Therefore, if we want to obtain the desired bound, we have to rescale $\epsilon$ by a factor of $1/2^{O(\ddim)}$.
Since the time complexity of our algorithm also depends on $\epsilon$, it could be blow up due to such naive rescaling.

To resolve this issue, we apply a similar bootstrap as in~\cite[Section 3.3]{Cohen-AddadFS21}.
By \Cref{lemma:fl_bounded_centers_correctness_weaker}, if $S$ is an $\alpha$-approximation of $\optFL(X, Y)$, then $\widehat{F}$ is a $(1 + \epsilon + 2 \epsilon\alpha)$-approximation, which improves over $S$ by a $2 \epsilon$ factor.
Therefore, we can start from an arbitrary $2^{O(\ddim)}$-approximation $\widehat{F}_0$; at each stage $i$ run the algorithm in \Cref{sec:fl_alg} with $S = \widehat{F}_i$
to obtain a better approximate solution $\widehat{F}_{i + 1}$.
The procedure halts when $i$ reaches $c \cdot \ddim$ for some sufficiently large constant $c$.
We return $\widehat{F}_{c \cdot \ddim}$ as our final solution.
Since the success probability in \Cref{lemma:fl_bounded_centers_correctness_weaker} can be boosted to $1 - 1/\ddim^2$ by standard amplification,
we can guarantee that with constant probability, all steps of our bootstrap succeed.
Moreover, the bootstrap only introduces a $\poly(\ddim)$ overhead in the running time.

Finally, we prove the returned solution $\widehat{F}_{c \cdot \ddim}$ satisfies that $\floc(X, \widehat{F}_{c \cdot \ddim}) \leq (1 + \epsilon) \optFL(X, Y)$.

\begin{proof}[Proof of \Cref{thm:FL_bounded_centers} (correctness)]
    For $0 \leq i \leq c \cdot \ddim$, assume $\widehat{F}_i$ is an $\alpha_i$-approximate solution, i.e., $\floc(X, \widehat{F}_i) \leq \alpha_i \optFL(X, Y)$.
    By \Cref{lemma:fl_bounded_centers_correctness_weaker}, we have
    \[
    \alpha_{i + 1} \leq 1 + \epsilon + 2 \epsilon \alpha_i, \qquad 
    0 \leq i < c \cdot \ddim.
    \]
    Therefore, for every $s \leq c \cdot \ddim$, 
    \begin{align*}
        \alpha_{s} 
        &\leq 1 + \epsilon + 2 \epsilon \alpha_{s - 1} \\
        &\leq 1 + \epsilon + 2 \epsilon (1 + \epsilon) + (2 \epsilon)^2 \alpha_{s - 2} \\
        & \leq \dots \\
        & \leq (1 + \epsilon) \sum_{j = 0}^{\infty} (2 \epsilon)^j + (2 \epsilon)^{s} \alpha_0 \\
        & \leq \frac{1 + \epsilon}{1 - 2 \epsilon} + (2 \epsilon)^s 2^{O(\ddim)}. 
    \end{align*}
    Choosing $s = c \cdot \ddim$, we have $\alpha_s \leq 1 + O(\epsilon)$. 
    Rescaling $\epsilon$ completes the proof.
\end{proof}

\subsubsection{Time Complexity}
\label{sec:fl_alg_time}

We prove the time complexity of our algorithm is $2^{2^t} \tilde{O}(n + m)$ for $t = O(\ddim \log(\ddim/\epsilon))$.
By \Cref{lem: transform_to_bounded_aspect_ratio,lem:merge_instances}, we can wlog assume that the aspect ratio of $X \cup Y$ is $\Delta = \poly(n, m)$.

\begin{proof}[Proof of \Cref{thm:FL_bounded_centers} (time complexity)]
    We analyze time complexity for the preprocessing stage and the dynamic program separately.

    \paragraph{Preprocessing stage.}
    In the preprocessing stage, the $2^{O(\ddim)}$-approximation solution $S \subseteq Y$ can be computed in time $2^{O(\ddim)} O(n + m) \log \Delta$ by \Cref{sec: constant-approx-algos}.
    By \Cref{lemma:hierarchical_decomposition}, the hierarchical decomposition $\decom$ together with portals for each cluster can be computed in time $\rho^{-O(\ddim)} m \log \Delta = (\ddim /\epsilon)^{O(\ddim)} m \log \Delta$.
    By~\cite{Cohen-AddadFS21}, it can be checked in $\epsilon^{-O(\ddim)} O(n \log \Delta)$ time if $B_Y(\proj{S}{x}, \dist(x,S) /\epsilon)$ is badly cut for all $x$.
    Therefore, by \Cref{lemma:ANN_doubling_data}, the new instance $(X', Y)$ can be constructed in $\epsilon^{-O(\ddim)} O(n) \log \Delta$ time by moving all bad points $x \in X$ to its $(1 + \epsilon)$-ANN in $S$.
    The number of copies of each $f \in S$ can be computed in  a total of $O(n)$ time.

    For $x \in X'$, to determine the cluster $C(x)$ where $x$ is revealed, it suffices to find the cluster $C \in \decom_{j(x)}$ that contains $\proj{S}{x}$.
    This can be done in $O(\log \Delta)$ time, thus a total of $O(n \log \Delta)$ time.
    A list is maintained for every $C$, which contains all points $x$ that are revealed in $C$, i.e., $\{x \in X' \colon C(x) = C\}$.

    The proxies $N_x$ of $x$ is computed by $N_x := B_Y(\proj{S}{x}, \dist(x,S) / \epsilon) \cap P_{C(x)}$.
    This can be done in $O(|P_{C(x)}|) = \rho^{-O(\ddim)} = (\ddim/\epsilon)^{O(\ddim)}$ time by checking for each $u \in P_{C(x)}$ if $\dist(u, \proj{S}{x}) \leq \dist(x,S) /\epsilon$.
    Hence, the time complexity of computing $N_x$ for all $x \in X'$ is $(\ddim/\epsilon)^{O(\ddim)} O(n)$.

    In conclusion, the time complexity of the preprocessing stage is $(\ddim/\epsilon)^{O(\ddim)} O(n + m) \log \Delta$.

    \paragraph{Dynamic program.}
    Fix a cluster $C$ and a configuration $\veca_C = \braket{a_C^1, \dots, a_C^{|P_C|}}, \vecb_C = \braket{b_C^1, \dots, b_C^{|P_C|}}$.
    The computation of $g(C, \veca_C, \vecb_C)$ consists of two parts --- the inherited cost and the newly revealed connection cost.

    To compute the inherited cost, the algorithm enumerates all combinations of configurations of the child clusters of $C$, in the form of \eqref{eqn:config_combination}.
    Recall that $a_C^p, b_C^p \in \{0, \rho 2^\ell, 2\rho 2^\ell, \dots, 2^\ell/\epsilon\}$;
    therefore there are $(\epsilon \rho)^{- \sum_{D \in \child(C)} 2|P_D|} = (\epsilon \rho)^{\rho^{-O(\ddim)}}$ such combinations.
    For each one of these combinations, checking its consistency takes time 
    \begin{align*}
        |P_C| \cdot \sum_{D \in \child(C)} |P_D|
        + \sum_{D \in \child(C)} \sum_{p \in P_D} \left(|P_C| + \sum_{D' \in \child(C)} |P_{D'}|\right)
        = \rho^{-O(\ddim)}.
    \end{align*}
    Computing the summation of $g(D, \veca_D, \vecb_D)$ takes time $|\child(C)| = 2^{O(\ddim)}$.
    Therefore, the time complexity of computing the inherited cost of $C$ (i.e., \eqref{eqn:inherited_cost}) is 
    \begin{align*}
        (\epsilon \rho)^{\rho^{-O(\ddim)}} \cdot (\rho^{-O(\ddim)} + 2^{O(\ddim)})
        = \left(\frac{\ddim}{\epsilon}\right)^{\left(\frac{\ddim}{\epsilon}\right)^{O(\ddim)}}.
    \end{align*}

    To compute the connection cost for all $x$ with $C(x) = C$, it suffices to compute $\ccostport(x, F)$ for every such $x$.
    By the previous analysis, this can be done by computing $\dist(x, u)$ and $\dportH(u, F)$ for every $u \in N_x$, which can be done in $O(|N_x|) = \epsilon^{-O(\ddim)}$.
    Hence, the time complexity of computing the newly revealed connection cost is 
    \begin{align*}
        \epsilon^{-O(\ddim)} \cdot |\{x \colon C(x) = C\}|.
    \end{align*}

    We conclude that the time complexity of computing 
    $g(C, \veca_C, \vecb_C)$ is 
    \begin{align*}
        \left(\frac{\ddim}{\epsilon}\right)^{\left(\frac{\ddim}{\epsilon}\right)^{O(\ddim)}}
        + \epsilon^{-O(\ddim)} \cdot |\{x \colon C(x) = C\}|.
    \end{align*}
    The total time complexity of filling the DP table is 
    \begin{align*}
        &\qquad \sum_{\ell}\sum_{C \in \decom_\ell} (\epsilon \rho)^{-2 |P_C|} \left[
        \left(\frac{\ddim}{\epsilon}\right)^{\left(\frac{\ddim}{\epsilon}\right)^{O(\ddim)}}
        + \epsilon^{-O(\ddim)} \cdot |\{x \colon C(x) = C\}|
        \right] \\
        & = \left(\frac{\ddim}{\epsilon}\right)^{\left(\frac{\ddim}{\epsilon}\right)^{O(\ddim)}} \sum_{\ell} |\decom_\ell|
        + \left(\frac{\ddim}{\epsilon}\right)^{\left(\frac{\ddim}{\epsilon}\right)^{O(\ddim)}} 
        \sum_\ell\sum_{C \in \decom_\ell} |\{x \colon C(x) = C\}| \\
        & = \left(\frac{\ddim}{\epsilon}\right)^{\left(\frac{\ddim}{\epsilon}\right)^{O(\ddim)}} O(m + n) \log \Delta.
    \end{align*}

    \medskip 

    Combining the analysis above, we conclude that the time complexity of our algorithm is $2^{2^t} \cdot O(n + m) 
    \log \Delta$, for 
    \begin{align*}
        t = O\left(\ddim \log\frac{\ddim}{\epsilon}\right).
    \end{align*}
    As mentioned in \Cref{sec:fl_alg_correctness}, the bootstrap only introduces a $\poly(\ddim)$ overhead in the running time, which can be charged to the $2^{2^t}$ factor.
    \Cref{thm:FL_bounded_centers} follows with $\Delta = \poly(n, m)$.
\end{proof}

\section{Facility Location with Low-dimensional Clients}
\label{sec:fl_bounded_clients}

In this section, we focus on the setting where clients $X$ have bounded doubling dimension $\ddim$, while $Y$ is high-dimensional.
Our main result is the following.

\begin{theorem}\label{thm:FL_bounded_clients}
    There is a randomized algorithm that, given as input $\epsilon \in (0, \tfrac{1}{2})$, $n, m \in \NN$ and  $(X \cup Y, \dist)$ with $|X| = n, |Y| = m$,
    $\ddim(X) \leq \ddim$, 
    computes a $(1 + \epsilon)$-approximation of facility location in time 
    $2^{2^t} \cdot \tilde{O}(n + m)$
    with constant success probability, where
    \begin{equation*}
        t \in O\left(\ddim \log\frac{\ddim}{\epsilon}\right).
    \end{equation*}
\end{theorem}

\subsection{Structural Lemmas}

Our algorithm is based on the new hierarchical decomposition in \Cref{sec:metric_decomposition}.
Specifically, we construct the new decomposition $\modifydecom$ for the entire metric $X \cup Y$, with portal scaling parameter $\rho := \epsilon^{10}/\ddim^2$.
When calculating the connection cost, we will replace metric distances with portal-respecting distances $\dportP(x, F)$ introduced in \Cref{sec:portal_respecting_new}. 
\Cref{lemma:detour_complicated,lemma:detour_by_projection} are used to bound the error $\dportP(x, y) - \dist(x, y)$.
In \Cref{lemma:fl_bounded_clients_good_solution} we further prove the existence of $(1 + \epsilon)$-approximate portal-respecting solution.

In order to control the error $\dportP(x, y) - \dist(x, y)$,
we need to bound the highest level where $\{x, \proj{X}{y}\}$ is cut on $\modifydecom$.
We need the notion of badly cut and bad points w.r.t. $\modifydecom$.
The following definition extends the badly cut notion in~\cite{Cohen-AddadFS21} to our new decomposition $\modifydecom$.

\begin{definition}[Badly cut]\label{def:badly_cut}
Let $\modifydecom$ be some random hierarchical decomposition as defined in Lemma \ref{lemma:new_decomposition_for_general_metric}.
    Say a set $T \subseteq X$ is badly cut w.r.t. $\modifydecom$, if $T$ is cut at level  $\log \diam(T) + \epsddim$ in $\modifydecom$.
\end{definition}
The following lemma is a direct corollary of property \ref{it:cutting} in \Cref{lemma:new_decomposition_for_general_metric}.

\begin{lemma}\label{lemma:badly_cut_ambient}
    For every $T \subseteq X$, 
    \begin{align*}
        \Pr[T \text{ is badly cut w.r.t. } \modifydecom] \leq c\cdot \epsilon,
    \end{align*}
    where $c$ is the same constant as in \ref{it:cutting}.
\end{lemma}
Next, 
we extend the definition of bad points (\Cref{def:bad_points_bounded_centers}) to our current setting.

\begin{definition}[Bad points]\label{def:two_types_of_bad_points}
    Consider a metric space $(X \cup Y, \dist)$ and let $\modifydecom$ be the decomposition of $X \cup Y$ in \Cref{lemma:new_decomposition_for_general_metric}.
    Let $F^*, S \subseteq Y$.
    \begin{enumerate}[label=(\arabic*)]
        \item A point $x \in X$ is called a \emph{bad client} (w.r.t. $S$ and $\modifydecom$), if $B_X(x, \dist(x, S)/\epsilon)$ is badly cut w.r.t. $\modifydecom$.
        Denote the set of bad clients as $\Bad_X$.
        \item A point $y \in S$ is called a \emph{bad facility} (w.r.t. $S, F^*$ and $\modifydecom$) if $B_X(\proj{X}{y}, 100 \dist(\proj{X}{y}, F^*) + 100 \dist(\proj{X}{y}, S))$ is badly cut w.r.t. $\modifydecom$.
        Denote the set of bad facilities as $\Bad_S$.
    \end{enumerate}
\end{definition}

In our subsequent analysis, we instantiate $F^* \subseteq Y$ as an optimal solution for facility location and $S \subseteq Y$ as a constant approximate solution, 
and define bad clients $\Bad_X \subseteq X$ and bad facilities $\Bad_S \subseteq S$.

Let us compare \Cref{def:two_types_of_bad_points} with \Cref{def:bad_points_bounded_centers}, which defines bad points in a similar but different way.
In \Cref{def:two_types_of_bad_points}, both bad clients and bad facilities are defined w.r.t. balls in $X$ instead of $Y$.
This is because our current notion of badly cut (\Cref{def:badly_cut}) is only defined for subsets of the low-dimensional space $X$.
Moreover, the centers of the balls are $x$ and $\proj{X}{y}$ respectively, which is symmetric to \Cref{def:bad_points_bounded_centers} where the centers are $\proj{S}{x}$ and $y$.
Another difference is that for bad facilities, the radius of the ball depends on both $\dist(\proj{X}{y}, F^*)$ and $\dist(\proj{X}{y}, S)$.
A similar extra $100 \dist(\proj{X}{y}, S)$ term does not appear in \Cref{def:bad_points_bounded_centers}, since $\dist(y, S) = 0$ for $y \in S$.

The following lemma shows how we can eliminate bad clients in $X$.
Similar to \Cref{lemma:fl_bounded_centers_new_instance}, we will move every bad client $x$ to a nearby point and create a new dataset $X'$.
However, since $\proj{S}{x}$ is now in the high-dimensional ambient space $Y$, moving $x$ to $\proj{S}{x}$ as \Cref{lemma:fl_bounded_centers_new_instance} would increase the dimension of our dataset, 
which could be problematic for our algorithm design.
Therefore, we slightly alter our strategy by moving $x$ to $\proj{X}{\proj{S}{x}}$, i.e., the nearest neighbor of $\proj{S}{x}$ in $X$.
After this modification, our new dataset $X'$ still lies in $X$ and thus has bounded doubling dimension.
We can further prove that the cost of such movement is at most an $\epsilon$ fraction of $\floc(X, S)$.

\begin{lemma}[New instance]\label{lemma:fl_bounded_clients_new_instance}
    Given $(X \cup Y, \dist)$, with $\ddim(X) \leq \ddim$, and a solution $S \subseteq Y$ for facility location, construct a new 
    (multi-)set of clients $X' \subseteq X$ as $X' := \phi(X)$ for 
    \begin{align*}
        \phi(x) := \begin{cases}
            x, &\text{if } x \notin \Bad_X; \\
            \proj{X}{\proj{S}{x}}, & \text{if } x \in \Bad_X,
        \end{cases}
    \end{align*} 
    namely, $X'$ is constructed from $X$ by moving every bad client $x$ to $\proj{X}{\proj{S}{x}}$.
    Then with probability $0.99$,
    \begin{align*}
        \forall F \subseteq Y, \qquad \Big|\floc(X, F) - \floc(X', F)\Big|
        \leq \epsilon \floc(X, S).
    \end{align*}
\end{lemma}

We remark that, if $S$ is a constant approximation, then we have $\floc(X', F) \in (1 \pm O(\epsilon)) \floc(X, F)$ for arbitrary solution $F \subseteq Y$ with probability $0.99$.
The proof of \Cref{lemma:fl_bounded_clients_new_instance} is analogous to \Cref{lemma:fl_bounded_centers_new_instance}, and thus is omitted.

The following lemma claims the existence of $(1 + \epsilon)$-approximate portal-respecting solution.
It is similar to \Cref{lemma:fl_bounded_centers_good_solution}, with the only difference that we now use the portal-respecting distance $\dportP$ introduced in \Cref{sec:portal_respecting_new}.
It allows us to compute portal-respecting solution in our facility location algorithm.

\begin{lemma}[Good portal-respecting solution]\label{lemma:fl_bounded_clients_good_solution}
    Let $(X \cup Y, \dist)$ with $\ddim(X) \leq \ddim$ be a facility location instance, $S \subseteq Y$ be a solution and $\modifydecom$ be the hierarchical decomposition on $X \cup Y$ in \Cref{lemma:new_decomposition_for_general_metric} with portal parameter $\rho = \epsilon^{10}/\ddim^2$.
    Let $X' \subseteq X$ be the new client set constructed by \Cref{lemma:fl_bounded_clients_new_instance}.
    Then with probability $0.9$, there exists a solution $F \subseteq Y$, such that 
    \begin{equation}\label{eqn:fl_bounded_clients_good_solution}
        \sum_{x \in X'} \dportP(x, F) + \ocost(F) \leq (1 + \epsilon) \optFL(X, Y) + \epsilon \floc(X, S).
    \end{equation}
\end{lemma}

Technically, proving \Cref{lemma:fl_bounded_clients_good_solution} will be more challenging than \Cref{lemma:fl_bounded_centers_good_solution}.
Let us first recall the key ideas behind the proof of \Cref{lemma:fl_bounded_centers_good_solution}.
Our main goal is to upper bound for every $x \in X'$ the highest level $\ell$ where $\{\proj{S}{x}, \proj{F}{x}\}$ is cut.
For this purpose, consider client $x$ and facility $\proj{S}{x}$; if either $x$ is not a bad client or $\proj{S}{x}$ is not a bad facility, then we can bound $\ell$ by utilizing the fact that either $B_Y(\proj{S}{x}, \dist(x, S)/\epsilon)$ or $B_Y(\proj{S}{x}, 10\dist(\proj{S}{x}, F^*))$ is not badly cut.
Otherwise, we directly add $\proj{S}{x}$ to $F$, guaranteeing that $\proj{S}{x} = \proj{F}{x} = x$; and thus $\{\proj{S}{x}, \proj{F}{x}\}$ is never cut.

However, this idea does not work in our current setting.
Due to the definition of $\dportP$, we now need to bound the highest level where $\{x, \proj{F}{x}\}$ is cut w.r.t. $\modifydecom$.
This introduces two main issues.
First, $\proj{F}{x}$ is in the high-dimensional space $Y$, but existing tools such as badly cut (\Cref{lemma:badly_cut_ambient}) and bad points (\Cref{def:two_types_of_bad_points}) are only defined w.r.t. subsets of $X$.
Second, in the case where both $x$ is a bad client and $\proj{S}{x}$ is a bad facility, 
we no longer have the property that $x$ and $\proj{F}{x}$ are collocated/never cut in $X'$, which makes it difficult to bound the cutting level.

To resolve these issues, we need some new ideas.
For the first issue, instead of directly considering the level where $\{x, \proj{F}{x}\}$ is cut, we will consider the level $\{x, \proj{X}{\proj{F}{x}}\}$ is cut.
Since $\proj{X}{\proj{F}{x}}$ is in the low-dimensional space, we can still use the tools developed before.
The correctness of replacing $\proj{F}{x}$ with it nearest neighbor in $X$ is guaranteed by \Cref{lemma:detour_by_projection}.
The second issue is much more challenging, 
because we have completely no information for any meaningful cutting level of $\{x, \proj{F}{x}\}$ in this bad case.
Nevertheless, we can still utilize the level $h(\proj{S}{x})$, the level where ornament $\proj{S}{x}$ is attached to $\modifydecom$ (Property \ref{it:ornaments} of \Cref{lemma:new_decomposition_for_general_metric}), 
and show that $\dportP(x, F) \leq O(\dist(x, S))$.
This bound looks useless at first, because $\dist(x, S)$ is not directly comparable with $\dist(x, F^*)$.
However, we can show that bad cases happen with probability $\epsilon$.
Hence the total error can be bounded by $O(\epsilon) \floc(X, S)$.

\begin{proof}[Proof of \Cref{lemma:fl_bounded_clients_good_solution}]
    Recall that $F^* \subseteq Y$ is the optimal solution and that $\Bad_S \subseteq S$ is the set of bad facilities in $S$.
    We show that with probability $0.97$, solution 
    \[
    F := F^* \cup \Bad_S
    \]
    satisfies \eqref{eqn:fl_bounded_clients_good_solution}.
    By \Cref{lemma:fl_bounded_clients_new_instance}, $X' = \phi(x)$.
    We thus rewrite the LHS of Equation \eqref{eqn:fl_bounded_clients_good_solution} as 
    \[
    \sum_{x \in X} \dportP(\phi(x), F) + \ocost(F).
    \]

    To bound the connection cost $\dportP(\phi(x), F)$ for each $x \in X$, we consider the following cases.

    \paragraph{Case 1: $x$ is a bad client, and $\proj{S}{x}$ is a bad facility.}
    By the construction of $X'$, $\phi(x) = \proj{X}{\proj{S}{x}}$. 
    Since $\proj{S}{x}$ is a bad facility, it is in $\Bad_S$, and thus in $F$.
    If $\proj{S}{x} \in X$, then $\proj{X}{\proj{S}{x}} = \proj{S}{x}$, and thus $\dportP(\phi(x), F) = 0$.
    If $\proj{S}{x} \in Y \setminus X$, then by the construction of $\modifydecom$, the set $\{\proj{X}{\proj{S}{x}}, \proj{S}{x}\}$ is cut exactly at level $h(\proj{S}{x})$, the level where $\proj{S}{x}$ is a leaf node on $\modifydecom$.
    By \Cref{lemma:detour_complicated}, 
    \begin{align}
        \dportP(\phi(x), F)
        &\leq \dportP(\proj{X}{\proj{S}{x}}, \proj{S}{x}) \notag\\
        &\leq \dist(\proj{X}{\proj{S}{x}}, \proj{S}{x}) + O(\rho) 2^{h(\proj{S}{x})} \notag\\
        &\leq \dist(x, \proj{S}{x}) + O(\rho) \frac{\dist(\proj{S}{x}, X)}{\sqrt{\rho}} 
        \notag \\
        &\leq (1 + O(\sqrt{\rho})) \dist(x, \proj{S}{x}) \notag\\
        &\leq 2 \dist(x, S). \label{eqn:case_1}
    \end{align}

    \paragraph{Case 2: $x$ is a bad client, and $\proj{S}{x}$ is not a bad facility.}
    By the construction of $X'$, $\phi(x) = \proj{X}{\proj{S}{x}}$.
    Our plan is to upper bound the highest level where $\{\phi(x), \proj{X}{\proj{F}{\phi(x)}}\}$ is cut w.r.t. $\modifydecom$.

    Since $\proj{S}{x}$ is not a bad facility, by \Cref{def:badly_cut,def:two_types_of_bad_points}, the ball 
    \begin{align*}
    &B_X\Big(\proj{X}{\proj{S}{x}}, 100 \dist(\proj{X}{\proj{S}{x}}, F^*) + 100 \dist(\proj{X}{\proj{S}{x}}, S)\Big) \\
    = & B_X\Big(\phi(x), 100 \dist(\phi(x), F^*) + 100 \dist(\phi(x), S)\Big)
    \end{align*}
    is cut at level at most 
    $\ell \leq \log(100 \dist(\phi(x), F^*) + 100 \dist(\phi(x), S)) + \epsddim$.
    Note that 
    \begin{align*}
        \dist(\phi(x), \proj{X}{\proj{F}{\phi(x)}})
        &\leq \dist(\phi(x), \proj{F}{\phi(x)}) + \dist(\proj{F}{\phi(x)}, \proj{X}{\proj{F}{\phi(x)}}) \\
        &\leq 2 \dist(\phi(x), \proj{F}{\phi(x)}) \\
        &\leq 2 \dist(\phi(x), F^*) &&\text{Since $F^* \subseteq F$.}
    \end{align*}
    Hence, $\proj{X}{\proj{F}{\phi(x)}} \in B_X(\phi(x), 100 \dist(\phi(x), F^*) + 100 \dist(\phi(x), S))$.
    Therefore, the highest level where $\{\phi(x), \proj{X}{\proj{F}{\phi(x)}}\}$ is cut can be bounded by 
    $\ell \leq \log(100 \dist(\phi(x), F^*) + 100 \dist(\phi(x), S)) + \epsddim$.

    By \Cref{lemma:detour_by_projection}, 
    \begin{align}
        \dportP(\phi(x), F)
        &\leq (1 + \sqrt{\rho}) \dist(\phi(x), F) + O(\sqrt{\rho}) 2^\ell \notag \\
        &\leq (1 + \sqrt{\rho}) \dist(\phi(x), F) + O(\sqrt{\rho}) \expepsddim (\dist(\phi(x), F^*) + \dist(\phi(x), S)) \notag \\
        &\leq (1 + \sqrt{\rho}) \dist(\phi(x), F) + O(\epsilon) (\dist(\phi(x), F^*) + \dist(\phi(x), S)). 
        \label{eqn:case_2}
    \end{align}

    \paragraph{Case 3: $x$ is not a bad client, and $\proj{S}{x}$ is a bad facility.}
    In this case $\phi(x) = x$, and the ball $B_X(x, \dist(x, S) / \epsilon)$ is cut at level at most  $\ell \leq \log(\dist(x, S)/\epsilon) + \epsddim$.
    On the other hand, since $\proj{S}{x} \in \Bad_S$, we have $\proj{S}{x} \in F$.
    Therefore, 
    \[
        \dist(x, \proj{X}{\proj{F}{x}})
        \leq 2 \dist(x, F) 
        \leq 2 \dist(x, S).
    \]
    We have $\proj{X}{\proj{F}{x}} \in B_X(x, \dist(x, S) / \epsilon)$.
    Hence, the highest level where $\{x, \proj{X}{\proj{F}{x}}\}$ is cut is at most 
    $\ell \leq \log(\dist(x, S)/\epsilon) + \epsddim$.

    By \Cref{lemma:detour_by_projection},
    \begin{align}
        \dportP(x, F) 
        &\leq (1 + O(\sqrt{\rho})) \dist(x, F) + O(\sqrt{\rho}) 2^\ell \notag \\
        &\leq (1 + O(\sqrt{\rho})) \dist(x, F) + O(\sqrt{\rho}) \expepsddim \frac{\dist(x, S)}{\epsilon} \notag\\
        & \leq (1 + O(\sqrt{\rho})) \dist(x, F) + O(\epsilon) \dist(x, S). \label{eqn:case_3}
    \end{align}

    \paragraph{Case 4: $x$ is not a bad client, and $\proj{S}{x}$ is not a bad facility.}
    In this case $\phi(x) = x$.
    Following the proof in \Cref{sec:fl_bounded_centers}, consider the following two sub-cases:

    If $\dist(x, F^*) \leq \dist(x, S) / (2 \epsilon)$, 
    first observe that 
    the ball $B_X(x, \dist(x, S) / \epsilon)$ is cut at level at most $\ell \leq \log(\dist(x, S)/\epsilon) + \epsddim$.
    On the other hand,
    \[
        \dist(x, \proj{X}{\proj{F}{x}}) 
        \leq 2 \dist(x, F) 
        \leq 2 \dist(x, F^*) 
        \leq \frac{\dist(x, S)}{\epsilon}.
    \]
    Hence $\proj{X}{\proj{F}{x}} \in B_X(x, \dist(x, S) / \epsilon)$.
    Therefore, the highest level where $\{x, \proj{X}{\proj{F}{x}}\}$ is cut can be bounded by
    $\ell \leq \log(\dist(x, S)/\epsilon) + \epsddim$.
    Similar to Case 3, we can show that 
    \begin{equation}\label{eqn:case_4.1}
        \dportP(\phi(x), F)
        \leq (1 + O(\sqrt{\rho})) \dist(x, F) + O(\epsilon) \dist(x, S).
    \end{equation}

    If $\dist(x, F^*) > \dist(x, S) / (2 \epsilon)$, we  utilize the fact that the ball
    \[
    B_X\Big(\proj{X}{\proj{S}{x}}, 100 \dist(\proj{X}{\proj{S}{x}}, F^*) + 100 \dist(\proj{X}{\proj{S}{x}}, S)\Big)
    \]
    is cut at level at most $\ell \leq \log(100 \dist(\proj{X}{\proj{S}{x}}, F^*) + 100 \dist(\proj{X}{\proj{S}{x}}, S)) + \epsddim$, by showing 
    both $x$ and $\proj{X}{\proj{F}{x}}$ are contained in the ball.

    First observe that 
    \begin{align*}
        \dist(x, F^*) 
        &\leq \dist(x, S) + \dist(\proj{S}{x}, F^*) \\
        &\leq \dist(x, S) + \dist(\proj{S}{x}, \proj{X}{\proj{S}{x}}) + \dist(\proj{X}{\proj{S}{x}}, F^*) \\
        &\leq 2 \dist(x, S) + \dist(\proj{X}{\proj{S}{x}}, F^*) \\
        & \leq 4 \epsilon \dist(x, F^*) + \dist(\proj{X}{\proj{S}{x}}, F^*).
    \end{align*}
    Hence, $\dist(x, F^*) \leq \frac{1}{1 - 4 \epsilon} \dist(\proj{X}{\proj{S}{x}}, F^*) \leq 3 \dist(\proj{X}{\proj{S}{x}}, F^*)$.
    
    Furthermore, note that 
    \begin{align*}
        \dist(x, \proj{X}{\proj{S}{x}})
        \leq 2 \dist(x, S) 
        < 4\epsilon \dist(x, F^*)
        \leq 12 \epsilon \dist(\proj{X}{\proj{S}{x}}, F^*)
        \leq 10 \dist(\proj{X}{\proj{S}{x}}, F^*), 
    \end{align*}
    and
    \begin{align*}
        \dist(\proj{X}{\proj{F}{x}}, \proj{X}{\proj{S}{x}})
        &\leq \dist(x, \proj{X}{\proj{F}{x}}) + \dist(x, \proj{X}{\proj{S}{x}}) \\
        &\leq 2 \dist(x, F) + 2 \dist(x, S) \\
        &\leq 2 \dist(x, F^*) + 4 \epsilon \dist(x, F^*)\\
        &= (2 + 4 \epsilon) \dist(x, F^*) \\
        & \leq 3(2 + 4 \epsilon) \dist(\proj{X}{\proj{S}{x}}, F^*) \\
        & \leq 10 \dist(\proj{X}{\proj{S}{x}}, F^*).
    \end{align*}
    We conclude that $x, \proj{X}{\proj{F}{x}} \in B_X(\proj{X}{\proj{S}{x}}, 100 \dist(\proj{X}{\proj{S}{x}}, F^*) + 100 \dist(\proj{X}{\proj{S}{x}}, S))$.
    Thus the highest level where $\{x, \proj{X}{\proj{F}{x}}\}$ is cut can be bounded by 
    $\ell \leq \log(100 \dist(\proj{X}{\proj{S}{x}}, F^*) + 100 \dist(\proj{X}{\proj{S}{x}}, S)) + \epsddim$.
    Similar to Case 2, we can show that 
    \begin{align}
        \dportP(\phi(x), F)
        &\leq (1 + O(\sqrt{\rho})) \dist(x, F) + O(\epsilon) \Big(\dist(\proj{X}{\proj{S}{x}}, F^*) + \dist(\proj{X}{\proj{S}{x}}, S) \Big) \notag \\
        &\leq (1 + O(\sqrt{\rho})) \dist(x, F) + O(\epsilon) \dist(x, F^*) + O(\epsilon) \dist(x, S). \label{eqn:case_4.2}
    \end{align}

    Combining \eqref{eqn:case_4.1} with \eqref{eqn:case_4.2}, we conclude that in this case, 
    \begin{equation}\label{eqn:case_4}
        \dportP(\phi(x), F)
        \leq (1 + O(\sqrt{\rho})) \dist(x, F) + O(\epsilon) \dist(x, F^*) + O(\epsilon) \dist(x, S).
    \end{equation}

    \medskip

    Combining above four cases (i.e., \eqref{eqn:case_1}, \eqref{eqn:case_2}, \eqref{eqn:case_3} and \eqref{eqn:case_4}), we have 
    \begin{align}
        &\qquad \sum_{x \in X} \dportP(\phi(x), F) + \ocost(F) \notag \\
        &\leq \sum_{\substack{x \in \Bad_X, \\ \proj{S}{x} \in \Bad_S}} 2 \dist(x, S) + \ocost(F)\notag \\
        & \hspace{5em} + \sum_{\substack{x \notin \Bad_X \text{ or} \\ \proj{S}{x} \notin \Bad_S}} \Big(
            (1 + O(\sqrt{\rho})) \dist(\phi(x), F)
            + O(\epsilon) \dist(\phi(x), F^*)
            + O(\epsilon) \dist(\phi(x), S)
        \Big)\\
        &\leq \sum_{\substack{x \in \Bad_X, \\ \proj{S}{x} \in \Bad_S}} 2 \dist(x, S) 
        + (1 + O(\sqrt{\rho}))\floc(X', F) 
        + O(\epsilon)\floc(X', F^*) 
        + O(\epsilon) \floc(X', S) \label{eqn:combine_4_cases}
    \end{align}

    To bound the RHS of \eqref{eqn:combine_4_cases}, first note that by \Cref{lemma:fl_bounded_clients_new_instance},
    \begin{align}
        \text{with probability } 0.99, \quad 
        &\floc(X', F) \leq \floc(X, F) + O(\epsilon) \floc(X, S), \notag \\ 
        &\floc(X', F^*) \leq \floc(X, F^*) + O(\epsilon) \floc(X, S) \notag\\
        &\floc(X', S) \leq \floc(X, S) + O(\epsilon) \floc(X, S)
        \leq O(1) \floc(X, S) \label{eqn:first_0.99}
    \end{align}
    hold simultaneously.

    Next, we bound $\floc(X, F)$ by $\optFL(X, Y) + O(\epsilon) \floc(X, S)$.
    Write 
    \begin{align*}
        \floc(X, F) 
        &= \sum_{x \in X} \dist(x, F) + \ocost(F) \\
        & \leq \sum_{x \in X} \dist(x, F^*) + \ocost(F^*) + \sum_{f \in \Bad_S} \ocost(f) \\
        &= \optFL(X, Y) + \sum_{f \in \Bad_S} \ocost(f).
    \end{align*}
    Then 
    \begin{align*}
        0 \leq \floc(X, F) - \optFL(X, Y) \leq \sum_{f \in \Bad_S} \ocost(f).
    \end{align*}
    Taking expectation, we have
    \begin{align*}
        \E[\floc(X, F) - \optFL(X, Y)]
        &\leq \E\left[\sum_{f \in \Bad_S} \ocost(f)\right] \\
        &= \sum_{f \in S} \ocost(f) \cdot \Pr[f \text{ is a bad facility}] \\
        &\leq O(\epsilon) \sum_{f \in S} \ocost(f) && \text{By \Cref{lemma:badly_cut_ambient}} \\
        & \leq O(\epsilon) \floc(X, S).
    \end{align*}
    Applying Markov's inequality to the (non-negative) random variable $\floc(X, F) - \optFL(X, Y)$, 
    \begin{equation}\label{eqn:second_0.99}
        \text{with probability } 0.99, \qquad  
    \floc(X, F) \leq \optFL(X, Y) + O(\epsilon) \floc(X, S).
    \end{equation}

    Finally, we bound $
    \sum_{\substack{x \in \Bad_X, \\ \proj{S}{x} \in \Bad_S}} 2 \dist(x, S)
    $. Observe that 
    \begin{align*}
        \E \sum_{\substack{x \in \Bad_X, \\ \proj{S}{x} \in \Bad_S}} 2 \dist(x, S)
        &=  \sum_{x \in X} 2 \dist(x, S) \cdot \Pr[x \in \Bad_X, \proj{S}{x} \in \Bad_S] \\
        & \leq O(\epsilon) \cdot \sum_{x \in X} 2 \dist(x, S) &&\text{By \Cref{lemma:badly_cut_ambient}}\\
        & \leq O(\epsilon) \cdot \floc(X, S).
    \end{align*}
    Applying Markov's inequality, we have 
    \begin{equation}\label{eqn:third_0.99}
        \text{with probability } 0.99, \qquad  
        \sum_{\substack{x \in \Bad_X, \\ \proj{S}{x} \in \Bad_S}} 2 \dist(x, S) 
        \leq O(\epsilon) \floc(X, S).
    \end{equation}

    Combining \eqref{eqn:first_0.99}, \eqref{eqn:second_0.99} and \eqref{eqn:third_0.99} with \eqref{eqn:combine_4_cases}, we conclude that with probability $0.97$ 
    \begin{align*}
        \sum_{x \in X} \dportP(\phi(x), F) + \ocost(F)
        &\leq 
        + (1 + O(\sqrt{\rho})) \floc(X, F) + O(\epsilon) \optFL(X, Y) + O(\epsilon) \floc(X, S)\\
        &\leq (1 + O(\epsilon)) \optFL(X, Y) + O(\epsilon) \floc(X, S).
    \end{align*}
    Rescaling $\epsilon$ concludes the proof.
    
\end{proof}

\subsection{The Algorithm}
\label{sec:alg_fl_bounded_clients}

Our algorithm is based on the dynamic programming framework proposed by~\cite{Cohen-AddadFS21}.
Given as input $(X \cup Y, \dist)$ with $\ddim(X) \leq \ddim$, the algorithm constructs the new hierarchical decomposition $\modifydecom$ on top of $X \cup Y$ (\Cref{alg:modified_decomposition}).
The dataset $X$ is subsequently transformed to $X'$ according to \Cref{lemma:fl_bounded_clients_new_instance}.
We run a dynamic program for $X'$ on top of $\modifydecom$, with respect to the portal-respecting distance measure $\dportP$.
Although the DP framework is similar to that of~\cite{Cohen-AddadFS21}, our algorithm handles facilities in the ambient space, which is essentially different from the setting of~\cite{Cohen-AddadFS21}.
This new feature makes our algorithm more technical challenging.
We describe the algorithm as follows.

\paragraph{Preprocessing stage.}
The algorithm first applies the techniques in \Cref{sec: constant-approx-algos} to compute a $2^{O(\ddim)}$-approximation solution $S \subseteq Y$ for facility location.
At first, this approximation might be too large, as it requires us to rescale the precision parameter $\epsilon$ by a factor of $1/2^{O(\ddim)}$, which may blow up the time complexity of our algorithm.
We ignore this potential issue when describing our algorithm at this point, and discuss how we can fix it in \Cref{sec:fl_bounded_clients_alg_correctness}.

The algorithm then uses \Cref{alg:modified_decomposition} to construct the new hierarchical decomposition $\modifydecom$,
where the scaling parameter of portals is set to $\rho = \epsilon^{10}/\ddim^2$.
For every $x \in X$, the algorithm checks if $B_X(x, \dist(x, S)/\epsilon)$ is badly cut w.r.t. $\modifydecom$ (i.e., if $x$ is a bad client).
It moves every bad client $x$ to $\proj{X}{\proj{S}{x}}$,\footnote{In fact, this step can be efficiently done by computing $(1 + \epsilon)$-ANN instead of the exact $\proj{X}{\proj{S}{x}}$.
This replacement does not affect the correctness of our previous analysis, and only enlarges the final approximation ratio by a $(1 + O(\epsilon))$ factor.}
creating the new instance $X'$.

\paragraph{Dynamic program.}
Each table entry of the dynamic program is represented by a cluster $C$ on the modified decomposition $\modifydecom$, together with a configuration
\[
\veca_C = \braket{a_C^1, a_C^2, \dots, a_C^{|P_C|}}, \qquad 
\vecb_C = \braket{b_C^1, b_C^2, \dots, b_C^{|P_C|}}.
\]
Roughly speaking, the configuration encodes the positional information of the current facility set $F$.
Specifically, each $a_C^p$ encodes the distance from portal $p$ to the closest facility inside $C$, 
and $b_C^p$ encodes the distance from portal $p$ to the closest facility outside $C$.
The value stored in entry $(C, \veca_C, \vecb_C)$ is the minimum cost of $C$ among all potential facility sets which are consistent with the configuration.
Formally, 
\begin{align*}
    g(C, \veca_C, \vecb_C)    
    := \min_{\substack{
    F \colon \text{$F$ is consistent} \\
    \text{with the configuration}}
    }
    \left\{\sum_{x \in X' \cap C}
    \dportP(x, F)
    + \ocost(F \cap C) \right\}.
\end{align*}

\paragraph{Base case.}
The base case of the dynamic program corresponds to the leaf nodes of $\modifydecom$ which are in $X'$.
Consider $x \in X'$ and its corresponding node $\{x\} \in \modifydecom_0$, together with a configuration $(a, b)$. 
(Wlog, assume that $x$ itself is the only portal of the node $\{x\}$.)
It is easy to check if the configuration is valid.
Recall that $a$ encodes the distance from $x$ to the closest facility in $\{x\}$.
Therefore, if $x \in Y \cap X'$, then $a$ is either $0$ or $\infty$;
if $x \in X' \setminus Y$, then $a$ must be $\infty$.
We set $g(\{x\}, a, b) = \infty$ if the configuration is invalid.

If the configuration is valid, we then compute the cost inside $\{x\}$.
If $a = 0$, then $x$ is a facility in the solution, we thus have $\ocost(F \cap \{x\}) = \ocost(x)$.
Otherwise, $\ocost(F \cap \{x\}) = 0$.
Furthermore, the connection cost of $x$ is $\dportP(x, F) = \min\{a, b\}$ multiplied by $w(x)$, the number of copies of $x$ in $X'$.
Hence, \[
g(\{x\}, a, b) = w(x) \cdot \min\{a, b\} + \ocost(F \cap \{x\}).
\]

\paragraph{Updating the DP table.}
Consider a higher level cluster $C \in \modifydecom_\ell$ with $\ell \geq 1$, together with its configuration $\veca_C = \braket{a_C^1, \dots, a_C^{|P_C|}}, \vecb_C = \braket{b_C^1, \dots, b_C^{|P_C|}}$.
Recall that $\Schild(C)$ is the set of $C$'s children which are ornaments, 
and $\Nchild(C)$ is the set of non-ornament children of $C$.

To compute $g(C, \veca_C, \vecb_C)$, the algorithm enumerates all combinations of configurations for clusters $D \in \Nchild(C)$.
For each of these combinations 
\begin{equation}\label{eqn:combination_of_configurations}
    \Big\{\Big(D, \veca_D = \braket{a^1_D, \dots, a^{|P_D|}_D}, \vecb_D  = \braket{b^1_D, \dots, b^{|P_D|}_D}\Big)\Big\}_{D \in \Nchild(C)},
\end{equation}
we discuss in the following how to compute its cost $\Gamma = \Gamma(\{(D, \veca_D, \vecb_D)\}_{D \in \Nchild(C)})$.
We will first need the following definition of weighted set cover problem. 

\begin{definition}[Weighted set cover problem]\label{prob:set_cover}
Given a universe $U$ with $m$ elements and $n$ subsets $R_1, R_2, \dots, R_n \subseteq U$ with $\bigcup_{i = 1}^n R_i = U$, together with weights $w_1, w_2, \dots, w_n$, 
compute a set of indices $I \subseteq [n]$ that minimizes $\sum_{i \in I} w_i$, subject to $\bigcup_{i \in I} R_i = U$.
\end{definition}

We compute the cost $\Gamma = \Gamma(\{(D, \veca_D, \vecb_D)\}_{D \in \Nchild(C)})$ by the following steps.

\begin{itemize}
    \item \textbf{Step 0: Initialize.}
    Maintain a universe $U$, which is set to $\emptyset$ initially.
    For every ornament $\{y\} \in \Schild(C)$, maintain a set $R_y$, which is set to $\emptyset$ initially.

    Roughly speaking, $U$ is the set of portals that need to be served by an ornament in $\Schild(C)$, 
    and $R_y$ contains the portals which can be served by $y$.

    \item \textbf{Step 1: Check consistency for $p \in P_C$.}
    For every $p \in P_C$, the algorithm finds $D \in \Nchild(C)$ and $q \in P_D$, such that $a_C^p = a_D^q + \dist(p, q)$.
    If such $D$ and $q$ exist, then check for the next $p \in P_C$.

    Otherwise we add the pair $(C, p)$ to $U$, indicating portal $p$ needs to be served by a facility in $\Schild(C)$.
    Our algorithm then finds all ornaments $\{y\} \in \Schild(C)$ that satisfy $a_C^p = \dist(p, y)$, and adds the pair $(C, p)$ to $R_y$, indicating $p$ can be served by $y$.
    If no such $y$ exists, we claim that the combination \eqref{eqn:combination_of_configurations} is inconsistent, and return $\Gamma = \infty$.
    
    \item \textbf{Step 2: Check consistency for $D \in \Nchild(C)$ and $q \in P_D$.}
    For every $D \in \Nchild(C)$ and $q \in P_D$, our algorithm first tries to find $p \in P_C$, such that $b^q_D = b_C^p + \dist(p, q)$.
    If such $p$ exists, it means $q$ connects to a facility outside $D$ first via $p$, then to a facility outside $C$.
    We can then check for the next $D$ and $q$.

    If such $p$ does not exist, our algorithm tries to find a non-ornament cluster $D' \in \Nchild(C), D' \neq D$ and a portal $q' \in P_{D'}$, such that $b^q_D = a^{q'}_{D'} + \dist(q, q')$.
    If such $D', q'$ exist, it means $q$ connects to a facility outside $D$ first via $q'$, then to a facility inside $D'$.
    We can then check for the next $D$ and $q$.

    Assume there exist neither $p \in P_C$, such that $b^q_D = b_C^p + \dist(p, q)$, nor $D' \in \Nchild(C), D' \neq D$ and $q' \in P_{D'}$, such that $b^q_D = a^{q'}_{D'} + \dist(q, q')$.
    We add the pair $(D, q)$ to the universe $U$, indicating that portal $q \in P_D$ needs to be served by a facility in $\Schild(C)$.
    Our algorithm then finds all ornaments $\{y\} \in \Schild(C)$ that satisfy $b^q_D = \dist(q, y)$, and adds $(D, q)$ to $R_y$, indicating $q$ can be served by $y$.
    If no such $y$ exist, we claim that the combination \eqref{eqn:combination_of_configurations} is inconsistent, and return $\Gamma = \infty$.

    \item \textbf{Step 3: Compute the opening cost for $\Schild(C)$.}
    At this stage, 
    the universe $U$ is a collection of (cluster, portal) pairs:
    \begin{equation}\label{eqn:universe}
        U \subseteq \{C\} \times P_C \cup \bigcup_{D \in \Nchild(C)} \{D\} \times P_D.
    \end{equation}
    For every $\{y\} \in \Schild(C)$, $R_y$ is a subset of $U$, and $\bigcup_{y \in \Schild(C)} R_y = U$.
    We solve the weighted set cover problem (\Cref{prob:set_cover}) for universe $U$ and set system $\{R_y \colon \{y\}\in \Schild(C)\}$, where the weight of each set $R_y$ is $\ocost(y)$.
    The opening cost for $\Schild(C)$ is the optimal cost for this set cover instance, denoted by $\Gamma_{\text{open}}$.

    \item \textbf{Step 4: compute the total cost $\Gamma$.}
    If the combination \eqref{eqn:combination_of_configurations} passes all the consistency checks, its cost is then computed by \[
    \Gamma = \sum_{D \in \Nchild(C)} g(D, \veca_D, \vecb_D) + \Gamma_{\text{open}}.
    \]
\end{itemize}

\medskip
\noindent The value of $g(C, \veca_C, \vecb_C)$ is the minimum $\Gamma$ over all possible combinations of child configurations:
\begin{equation*}
    g(C, \veca_C, \vecb_C)
    = \min_{\{(D, \veca_D, \vecb_D)\}_{D \in \Nchild(C)}} \Gamma\Big(
        \{(D, \veca_D, \vecb_D)\}_{D \in \Nchild(C)}
    \Big).
\end{equation*}

\paragraph{Reducing the number of configurations.}
Similar to \Cref{sec:fl_alg}, for every cluster $C \in \modifydecom_\ell$, we can restrict $a_p, b_p$ to multiplications of $\rho 2^\ell$ in the range $[0, 2^\ell/\epsilon]$.
This reduces the number of configurations to $(\epsilon \rho)^{-2 |P_C|} = 2^{2^t}$ for $ t = O(\ddim \log(\ddim/\epsilon))$.
The reason is exactly the same as \Cref{sec:fl_alg}, and thus is omitted here.

\subsubsection{Proof of Correctness}
\label{sec:fl_bounded_clients_alg_correctness}

Let $\widehat{F} \subseteq Y$ be the set of facilities returned by our algorithm in \Cref{sec:fl_alg}.
Similar to \Cref{sec:fl_alg_correctness}, 
we first show the following weaker bound for $\floc(X, \widehat{F})$.

\begin{lemma}\label{lemma:fl_bounded_clients_correctness_weaker}
    Let $\widehat{F} \subseteq Y$ be the set of facilities returned by the algorithm in \Cref{sec:alg_fl_bounded_clients} w.r.t. a solution $S \subseteq Y$.
    Then with constant probability, 
    \begin{align*}
        \floc(X, \widehat{F}) \leq (1 + \epsilon) \optFL(X, Y) + 2 \epsilon \floc(X, S).
    \end{align*}
\end{lemma}

\begin{proof}
    By \Cref{sec:alg_fl_bounded_clients}, $\widehat{F} \subseteq Y$ minimizes the facility location cost w.r.t. $\dportP(\cdot, \cdot)$, i.e., 
    \[
    \widehat{F} = \argmin_{F \subseteq Y} \sum_{x \in X'} \dportP(x, F) + \ocost(F).
    \]
    We note that the definition of $X'$ depends on $S$.

    By \Cref{lemma:fl_bounded_clients_good_solution}, 
    with probability $0.9$, there exists a solution $F \subseteq Y$, such that 
    \begin{equation}\label{eqn:apply_fl_bounded_clients_good_solution}
        \sum_{x \in X'} \dportP(x, F) + \ocost(F) \leq (1 + \epsilon) \optFL(X, Y) + \epsilon \floc(X, S).
    \end{equation}
    Therefore, 
    \begin{align*}
        \floc(X, \widehat{F}) 
        &\leq \floc(X', \widehat{F}) + \epsilon \floc(X, S) && \text{By \Cref{lemma:fl_bounded_clients_new_instance}, w.p. $0.99$} \notag \\
        & \leq \sum_{x \in X'} \dportP(x, \widehat{F}) + \ocost(F)
        + \epsilon \floc(X, S) \\
        & \leq \sum_{x \in X'} \dportP(x, F) + \ocost(F)
        + \epsilon \floc(X, S) &&\text{By definition of $\widehat{F}$}\\
        &\leq (1 + \epsilon) \optFL(X, Y) + 2 \epsilon \floc(X, S) && \text{By \eqref{eqn:apply_fl_bounded_clients_good_solution}}.
    \end{align*}
\end{proof}

We use the same bootstrap as in~\cite[Section 3.3]{Cohen-AddadFS21} and \Cref{sec:fl_alg_correctness}.
Concretely, we can start from an arbitrary $2^{O(\ddim)}$-approximation $\widehat{F}_0$; at each stage $i$ run the algorithm in \Cref{sec:alg_fl_bounded_clients} with $S = \widehat{F}_i$
to obtain a better approximate solution $\widehat{F}_{i + 1}$.
The procedure halts when $i$ reaches $c \cdot \ddim$ for some sufficiently large constant $c$.
We return $\widehat{F}_{c \cdot \ddim}$ as our final solution.
Since the success probability in \Cref{lemma:fl_bounded_clients_correctness_weaker} can be boosted to $1 - 1/\ddim^2$ by standard amplification,
we can guarantee that with constant probability, all steps of our bootstrap succeed.
Moreover, the bootstrap only introduces a $\poly(\ddim)$ overhead in the running time.

The proof of $\floc(X, \widehat{F}_{c \cdot \ddim}) \leq (1 + \epsilon) \optFL(X, Y)$ is the same as \Cref{sec:fl_alg_correctness}, and we omit it here.

\subsubsection{Time Complexity.}
\label{sec:time_complexity_fl_bounded_clients}

Recall in \Cref{sec:alg_fl_bounded_clients}, we need to solve a weighted set cover problem to determine the opening cost for every $\Schild(C)$.
It is well known that weighted set cover is NP-hard in general.
Nonetheless, we always have $m = |U| = O(1)$ in our setting.
This allows us to compute weighted set cover exactly in $O(n)$ time.

\begin{lemma}\label{lemma:set_cover_alg}
    There exists an algorithm that computes the weighted set cover problem (\Cref{prob:set_cover}) exactly in $O(2^m n)$ time.
\end{lemma}

\begin{proof}
    Consider the following dynamic program.
    For $i \in [n]$ and $T \subseteq U$, let $h(i, T)$ be the minimum weight of covering $T$ by $R_1, R_2, \dots, R_i$.
    Specifically, the optimal value for the original problem is $h(n, U)$.
    Define $h(0, \emptyset) = 0$ and $h(0, T) = \infty$ for $T \neq \emptyset$, and observe that
    \[
    h(i, T) = \min\{h(i - 1, T), h(i - 1, T \setminus R_i) + w_i\}.
    \]
    The value of all $h(i, T)$ can be computed in time $O(2^m n)$.
\end{proof}

We are ready to prove \Cref{thm:FL_bounded_clients} for the time complexity.
By \Cref{lem: transform_to_bounded_aspect_ratio,lem:merge_instances}, we can wlog assume that the aspect ratio of $X \cup Y$ is $\Delta = \poly(n, m)$.

\begin{proof}[Proof of \Cref{thm:FL_bounded_clients} (time complexity)]
    We analyze time complexity for the preprocessing stage and the dynamic program separately.

    \paragraph{Preprocessing stage.}
    The $2^{O(\ddim)}$-approximate solution $S \subseteq Y$ can be computed in $2^{O(\ddim)} O(n + m) \log \Delta$ time by \Cref{sec: constant-approx-algos}.
    By \Cref{lemma:new_decomposition_for_general_metric}, the new decomposition $\modifydecom$ can be computed in time $\rho^{-O(\ddim)} O(n + m) 
    \log \Delta = (\ddim/\epsilon)^{O(\ddim)} O(n + m) \log \Delta$.
    By~\cite{Cohen-AddadFS21},
    the algorithm can find all bad client $x \in \Bad_X$ (i.e., $x \in X$ s.t. $B_X(x, \dist(x, S)/\epsilon)$ is badly cut) in time $\epsilon^{-O(\ddim)} O(n \log \Delta)$.
    Finally, to construct the new dataset $X'$, the algorithm first uses the $(1 + \epsilon)$-ANN data structure in \Cref{lemma:ANN_doubling_queries} to compute $\proj{S}{x}$, then uses the $(1 + \epsilon)$-ANN data structure in \Cref{lemma:ANN_doubling_data} to further compute $\proj{X}{\proj{S}{x}}$, and moves $x$ to $\proj{X}{\proj{S}{x}}$.
    The total time complexity is $\epsilon^{-O(\ddim)} \tilde{O}(n + m) \log \Delta$.

    In conclusion, the time complexity of the preprocessing stage is $(\ddim/\epsilon)^{O(\ddim)} \tilde{O}(n + m) \log \Delta$.

    \paragraph{Dynamic program.}
    Fix a non-ornament cluster $C$ and a configuration $\veca_C = \braket{a_C^1, \dots, a_C^{|P_C|}}$, $\vecb_C = \braket{b_C^1, \dots, b_C^{|P_C|}}$.
    When computing $g(C, \veca_C, \vecb_C)$, the algorithm enumerates all combinations of configurations for $D \in \Nchild(C)$ in the form of \eqref{eqn:combination_of_configurations}, and computes the cost $\Gamma$ for each combination.
    We start by analyzing the time complexity of computing $\Gamma$.

    In step 0, the algorithm initializes $R_y \gets \emptyset$ for $\{y\} \in \Schild(C)$.
    This takes $O(|\Schild(C)|)$ time.

    In step 1, checking consistency for $p \in P_C$ requires going over all $D \in \Nchild(C), q \in P_D$ and $\{y\} \in \Schild(C)$. 
    It thus takes time
    \[
    |P_C| \cdot \left(\sum_{D \in \Nchild(C)} |P_D| + |\Schild(C)|\right)
    = \rho^{-O(\ddim)} + \rho^{-O(\ddim)} |\Schild(C)|.
    \]

    In step 2, checking consistency for $D \in \Nchild(C), q \in P_D$ requires going over all $p \in P_C$, $D' \in \Nchild(C), q' \in P_{D'}$ and $\{y\} \in \Schild(C)$. 
    It thus takes time
    \[
    \sum_{D \in \Nchild(C)} \sum_{q \in P_D} 
    \left(
        |P_C| + \sum_{D' \in \Nchild(C)} |P_{D'}| 
        + |\Schild(C)|
    \right)
    = \rho^{-O(\ddim)} \cdot(1+ |\Schild(C)|).
    \]

    In step 3, the algorithm solves a weighted set cover problem for universe $U$ given in \eqref{eqn:universe}, and set system $\{R_y \colon y \in \Schild(C)\}$.
    By \Cref{lemma:set_cover_alg}, this can be done in $O(2^{|U|} |\Schild(C)|)$ time.
    Since $|U| \leq \rho^{-O(\ddim)}$, the time complexity of step 3 is 
    \[
    2^{\rho^{-O(\ddim)}} |\Schild(C)|.
    \]

    We conclude that the time complexity of computing the cost $\Gamma$ for a single combination \eqref{eqn:combination_of_configurations} is 
    \[
    \rho^{-O(\ddim)} + 2^{\rho^{-O(\ddim)}} |\Schild(C)|.
    \]
    Since there are a total of $(\epsilon \rho)^{-\rho^{-O(\ddim)}}$ such combinations, 
    the time complexity of computing $g(C, \veca_C, \vecb_C)$ is 
    \[
    (\epsilon \rho)^{-\rho^{-O(\ddim)}} \cdot \left(
        \rho^{-O(\ddim)} + 2^{\rho^{-O(\ddim)}} |\Schild(C)| 
    \right)
    \leq (\epsilon \rho)^{-\rho^{-O(\ddim)}} \cdot (1 + |\Schild(C)|).
    \]

    The total time complexity of filling the DP table is 
    \begin{align*}
        & \qquad \sum_{\ell = 0}^L \sum_{\substack{C \in \modifydecom_\ell \\ \text{$C$ is non-ornament}}} 
        (\epsilon \rho)^{-2 |P_C|} \cdot 
        (\epsilon \rho)^{-\rho^{-O(\ddim)}} \cdot (1 + |\Schild(C)|) \\
        & = (\epsilon \rho)^{-\rho^{-O(\ddim)}} 
        \sum_{\ell = 0}^L \sum_{\substack{C \in \modifydecom_\ell \\ \text{$C$ is non-ornament}}} 
        (1 + |\Schild(C)|) \\
        & = (\epsilon \rho)^{-\rho^{-O(\ddim)}}
        \sum_{\ell = 0}^L \Big(|\{C \in \modifydecom_\ell \colon \text{$C$ is non-ornament}\}|
        + |\{y \in Y \colon h(y) = \ell\}|\Big) \\
        & = \left(\frac{\ddim}{\epsilon}\right)^{\left(\frac{\ddim}{\epsilon}\right)^{O(\ddim)}} O(n + m) \log \Delta.
    \end{align*}

    Combining the analysis above, we conclude that the time complexity of our algorithm is $2^{2^t} \cdot \tilde{O}(n + m) \log \Delta$, for 
    \begin{align*}
        t = O\left(\ddim \log\frac{\ddim}{\epsilon}\right).
    \end{align*}
    As mentioned in \Cref{sec:fl_bounded_clients_alg_correctness}, the bootstrap only introduces a $\poly(\ddim)$ overhead in the running time, which can be charged to the $2^{2^t}$ factor.
    \Cref{thm:FL_bounded_clients} follows with $\Delta = \poly(n, m)$.
\end{proof}
\section{$k$-Median with Low-dimensional Centers}
\label{sec:kmedian_bounded_centers}

In this section, we extend \Cref{thm:FL_bounded_centers} to $k$-median problem in the setting where the candidate center set has low doubling dimension.
Our main result of this section is the following.

\begin{theorem}\label{thm:kmedian_bounded_centers}
    There is a randomized algorithm that, given as input $\epsilon \in (0, \tfrac{1}{2})$, $n, m, k \in \NN$ and  $(X \cup Y, \dist)$ with  $|X| = n, |Y| = m, \ddim(Y) \leq \ddim$, computes a $(1 + \epsilon)$-approximation of $k$-median in time 
    $2^{2^t} \cdot \tilde{O}(n + m)$
    with constant success probability, where
    \begin{equation*}
        t \in O\left(\ddim \log\frac{\ddim}{\epsilon}\right).
    \end{equation*}
\end{theorem}

\subsection{Structural Lemmas}

Let $F^* \subseteq Y, |F^*| \leq k$ be the optimal solution for $k$-median, and 
$S \subseteq Y, |S| \leq k$ be a constant approximate solution for $k$-median.
We define bad clients $\Bad_X \subseteq X$ and bad facilities $\Bad_S \subseteq S$ w.r.t. $F^*$, $S$ and $\decom$.
Analogous to \Cref{lemma:fl_bounded_centers_new_instance}, we have the following lemma that eliminates bad clients by moving them to the nearest facility in $S$.
The proof of \Cref{lemma:kmedian_bounded_centers_new_instance} is the same as \Cref{lemma:fl_bounded_centers_new_instance}, and thus is omitted.

\begin{lemma}[New instance]\label{lemma:kmedian_bounded_centers_new_instance}
    Given $(X \cup Y, \dist)$, with $\ddim(Y) \leq \ddim$, and a constant approximate solution $S \subseteq Y$ for $k$-median, 
    construct a new 
    (multi-)set of clients $X' \subseteq X \cup Y$ as $X' := \phi(X)$ for 
    \begin{align*}
        \phi(x) := \begin{cases}
            x, &\text{if } x \notin \Bad_X; \\
            \proj{S}{x}, & \text{if } x \in \Bad_X,
        \end{cases}
    \end{align*} 
    namely, $X'$ is constructed from $X$ by moving every bad client $x$ to $\proj{S}{x}$.
    Then with probability $0.99$,
    \begin{align*}
        \forall F \subseteq Y, |F| \leq k, \qquad \Big|\median(X, F) - \median(X', F)\Big|
        \leq \epsilon \optkmed(X, Y).
    \end{align*}
\end{lemma}

The following lemma guarantees the existence of $(1 + \epsilon)$-approximate portal-respecting solution.
It is the same as \Cref{lemma:fl_bounded_centers_good_solution}, with an extra requirement that we now need $F$ to have bounded size $F \leq k$.

\begin{lemma}[Good portal-respecting solution]\label{lemma:kmedian_bounded_centers_good_solution}
    Let $(X \cup Y, \dist)$ with $\ddim(Y) \leq \ddim$ be a $k$-median instance, $S \subseteq Y$ be a constant approximate solution and $\decom$ be the hierarchical decomposition on $Y$ in \Cref{lemma:hierarchical_decomposition} with portal parameter $\rho = \epsilon^{10}/\ddim^2$.
    Let $X' \subseteq X$ be the new client set constructed by \Cref{lemma:kmedian_bounded_centers_new_instance}.
    For $x \in X'$ and $F \subseteq Y$, define the portal-respecting connection cost of $x$ as
    \begin{align*}
        \ccostport(x, F) := \begin{cases}
            \dportH(x, F), & \text{if } x = \proj{S}{x}; \\
            \dporthatH(x, F), &\text{if } x \neq \proj{S}{x}.
        \end{cases}
    \end{align*}
    Then with probability $0.9$, there exists a solution $F \subseteq Y, |F| \leq k$, such that 
    \begin{equation}\label{eqn:kmedian_bounded_centers_good_solution}
        \sum_{x \in X'} \ccostport(x, F) \leq (1 + \epsilon) \optkmed(X, Y).
    \end{equation}
\end{lemma}

To prove \Cref{lemma:kmedian_bounded_centers_good_solution}, we first show in the following lemma that there exists a facility set of size at most $k$ containing $\Bad_S$.
\Cref{lemma:kmedian_bounded_centers_good_solution_construction} is essentially a combination of~\cite[Claims 20 and 21]{Cohen-AddadFS21}.
We provide the proof in \Cref{appendix:proof_of_lemma_kmedian_bounded_centers_construction} for completeness.

\begin{restatable}{lemma}{lemmakmedianboundedcentersgoodsolutionconstruction}
    \label{lemma:kmedian_bounded_centers_good_solution_construction}
Let $(X \cup Y, \dist)$ with $\ddim(Y) \leq \ddim$ be a $k$-median instance, and $S \subseteq Y$ be a constant approximate solution.
Then there exists a set of facilities $F \subseteq Y$, which satisfies the following properties.
\begin{enumerate}[label=(\alph*)]
    \item $\Bad_S \subseteq F$.
    \item For every $x \in X \cup Y$, $\dist(x, F) \leq 6 \dist(x, F^*) + 5 \dist(x, S)$.
    \item With probability $0.99$, $|F| \leq k$ and $\median(X, F) \leq (1 + \epsilon) \optkmed(X, Y)$.
\end{enumerate}
\end{restatable}

Based on \Cref{lemma:kmedian_bounded_centers_good_solution_construction}, we are ready to prove \Cref{lemma:kmedian_bounded_centers_good_solution}.
In fact, we show the facility set in \Cref{lemma:kmedian_bounded_centers_good_solution_construction} satisfies the conditions in \Cref{lemma:kmedian_bounded_centers_good_solution}.

\begin{proof}[Proof of \Cref{lemma:kmedian_bounded_centers_good_solution}]
    We show that the facility set in \Cref{lemma:kmedian_bounded_centers_good_solution_construction} satisfies \eqref{eqn:kmedian_bounded_centers_good_solution}.
    Rewrite the LHS of \eqref{eqn:kmedian_bounded_centers_good_solution} as 
    \[
    \sum_{x \in X} \ccostport(\phi(x), F) + \ocost(F).
    \]
    For every $x \in X$, we first bound the highest level where $\{\proj{S}{\phi(x)}, \proj{F}{\phi(x)}\}$ is cut.
    Specifically, let $\ell(x)$ be the highest level where $\{\proj{S}{\phi(x)}, \proj{F}{\phi(x)}\}$
    is cut; we show that 
    \[
    \ell(x) \leq \log(\dist(\phi(x), S) /\epsilon + 10 \dist(\phi(x), F^*)) + \epsddim.
    \]
    Consider the following cases:

    \paragraph{Case 1: $x$ is a bad client, and $\proj{S}{x}$ is a bad facility.}
    By the construction of $X'$, $\phi(x) = \proj{S}{x}$. 
    Since $\proj{S}{x}$ is a bad facility, it is in $\Bad_S$, and thus in $F$.
    Therefore, $\proj{S}{\phi(x)} = \phi(x) = \proj{F}{\phi(x)}$; $\{\proj{S}{\phi(x)}, \proj{F}{\phi(x)}\}$ is never cut.

    \paragraph{Case 2: $x$ is a bad client, and $\proj{S}{x}$ is not a bad facility.}
    By the construction of $X'$, $\phi(x) = \proj{S}{x}$. 
    Since $\proj{S}{x}$ is not a bad facility, the ball 
    \[B_Y(\proj{S}{x}, 10 \dist(\proj{S}{x}, F^*)) = B_Y(\phi(x), 10 \dist(\phi(x), F^*))\] 
    is cut at level at most $\log(10 \dist(\phi(x), F^*)) + \epsddim$.
    By \Cref{lemma:kmedian_bounded_centers_good_solution_construction}, 
    \[
    \dist(\phi(x), \proj{F}{\phi(x)}) \leq 6 \dist(\phi(x), F^*) + 5 \dist(\phi(x), S) = 6 \dist(\phi(x), F^*).
    \] 
    Then $\proj{F}{\phi(x)} \in B_Y(\phi(x), 10 \dist(\phi(x), F^*))$.
    Hence, the highest level where $\{\proj{S}{\phi(x)}, \proj{F}{\phi(x)}\}$ is cut can be bounded by $\ell(x) \leq \log(10 \dist(\phi(x), F^*)) + \epsddim$.

    \paragraph{Case 3: $x$ is not a bad client, and $\proj{S}{x}$ is a bad facility.}
    In this case $\phi(x) = x$, and the ball $B_Y(\proj{S}{x}, \dist(x,S) / \epsilon)$ is cut at level at most $\log(\dist(x,S)/\epsilon) + \epsddim$.
    On the other hand, since $\proj{S}{x}$ is a bad facility, we have $\proj{S}{x} \in F$.
    Therefore, 
    \begin{align*}
        \dist(\proj{S}{x}, \proj{F}{x}) \leq \dist(x, \proj{S}{x}) + \dist(x, \proj{F}{x}) \leq 2 \dist(x,S),
    \end{align*}
    which implies $\proj{F}{x} \in B_Y(\proj{S}{x}, \dist(x,S) / \epsilon)$.
    Therefore, $\{\proj{S}{x}, \proj{F}{x}\}$ is cut at level at most $\ell(x) \leq \log(\dist(x,S)/\epsilon) + \epsddim = \log(\dist(\phi(x), S)/\epsilon) + \epsddim$.

    \paragraph{Case 4: $x$ is not a bad client in $X$, and $\proj{S}{x}$ is not a bad facility.}
    In this case $\phi(x) = x$, and the ball $B_Y(\proj{S}{x}, \dist(x, S)  / \epsilon)$ is cut at level at most $\log(\dist(x, S) /\epsilon) + \epsddim$, 
    and the ball $B_Y(\proj{S}{x}, 10 \dist(\proj{S}{x}, F^*))$ is cut at level at most $\log(10 \dist(\proj{S}{x}, F^*)) + \epsddim$.
    Consider the following sub-cases.

    If $\dist(x, F^*) \leq \dist(x, S) /(10 \epsilon)$, then 
    \begin{align*}
        \dist(\proj{S}{x}, \proj{F}{x}) 
        &\leq \dist(x, \proj{S}{x}) + \dist(x, \proj{F}{x}) \\
        &\leq \dist(x, S)  + 6 \dist(x, F^*) + 5 \dist(x, S)   &&\text{By \Cref{lemma:kmedian_bounded_centers_good_solution_construction} }\\
        &\leq \dist(x, S) /\epsilon &&\text{Since $\dist(x, F^*) \leq \dist(x, S)  / (10 \epsilon)$}.
    \end{align*}
    Hence, $\proj{F}{x} \in B_Y(\proj{S}{x}, \dist(x,S) / \epsilon)$, and thus the highest level where $\{\proj{S}{x}, \proj{F}{x}\}$ is cut is at most $\log(\dist(x, S) /\epsilon) + \epsddim = \log(\dist(\phi(x), S) /\epsilon) + \epsddim$.

    If $\dist(x, F^*) > \dist(x, S) /(10\epsilon)$, first note that
    \begin{align*}
        \dist(x, F^*) \leq \dist(\proj{S}{x}, F^*) + \dist(x, \proj{S}{x})
        = \dist(\proj{S}{x}, F^*) + \dist(x, S) 
        \leq \dist(\proj{S}{x}, F^*) + 10 \epsilon \dist(x, F^*),
    \end{align*}
    which implies $\dist(x, F^*) \leq \tfrac{1}{1 - 10 \epsilon} \dist(\proj{S}{x}, F^*)$.
    Then 
    \begin{align*}
        \dist(\proj{S}{x}, \proj{F}{x}) 
        &\leq \dist(x, \proj{S}{x}) + \dist(x, \proj{F}{x}) \\
        &\leq \dist(x, S)  + 6 \dist(x, F^*) + 5 \dist(x, S)   &&\text{By \Cref{lemma:kmedian_bounded_centers_good_solution_construction} } \\
        &\leq 6(1 + 10 \epsilon) \dist(x, F^*) && \text{Since $\dist(x, S)  < 10 \epsilon \dist(x, F^*)$} \\
        &\leq \frac{6 (1 + 10 \epsilon)}{1 - 10 \epsilon} \dist(\proj{S}{x}, F^*) \\
        &\leq 10 \dist(\proj{S}{x}, F^*) &&\text{for } \epsilon \leq 1/100.
    \end{align*}
    This implies $\proj{F}{x} \in B_Y(\proj{S}{x}, 10 \dist(\proj{S}{x}, F^*))$.
    Therefore, the highest level where $\{\proj{S}{x}, \proj{F}{x}\}$ is cut is at most $\log(10 \dist(\proj{S}{x}, F^*)) + \epsddim \leq \log(10 \dist(\phi(x), F^*) + 10 \dist(\phi(x), S)) + \epsddim$.

    \medskip

    In conclusion, we bound the highest level where $\{\proj{S}{\phi(x)}, \proj{F}{\phi(x)}\}$ is cut by $\ell(x) \leq \log(\dist(\phi(x), S) /\epsilon + 10 \dist(\phi(x), F^*)) + \epsddim$.
    
    We are now ready to prove \eqref{eqn:kmedian_bounded_centers_good_solution}.
    Fix $x \in X$.
    If $\phi(x) = \proj{S}{\phi(x)}$, since $\ell(x)$ is the highest level where $\{\proj{S}{\phi(x)}, \proj{F}{\phi(x)}\}$ is cut, 
    by \Cref{lemma:detour} we have 
    \[\ccostport(\phi(x), F) = \dportH(\phi(x), F) \leq \dist(\phi(x), F) + O(\rho) 2^{\ell(x)}.\]
    If $\phi(x) \neq \proj{S}{\phi(x)}$, then $x$ is not a bad client.
    By \Cref{lemma:connection_cost_portal}, we have  
    \begin{align*}
        \ccostport(\phi(x), F) 
        &= \dporthatH(\phi(x), F) 
    \leq \dporthatH(\phi(x), \proj{F}{\phi(x)}) \\
    &\leq (1 + O(\epsilon))\dist(\phi(x), F) + O(\rho) 2^{\ell(x)} + O(\rho) \expepsddim \cdot \frac{\dist(\phi(x), S)}{\epsilon}.
    \end{align*}
     For $\rho = \epsilon^{10}/\ddim^2$ we can bound the LHS of \eqref{eqn:kmedian_bounded_centers_good_solution} as 
    \begin{align*}
        &\qquad \sum_{x \in X} \ccostport(\phi(x), F) \\
        &\leq \sum_{\substack{x \in X \\ \phi(x) = \proj{S}{\phi(x)}}} (
            \dist(\phi(x), F) + O(\rho) 2^{\ell(x)}
        ) \\
        & \qquad + \sum_{\substack{x \in X \\ \phi(x) \neq \proj{S}{\phi(x)}}} 
        \left(
            (1 + O(\epsilon)) \dist(\phi(x), F) + O(\rho) 2^{\ell(x)} + O(\rho) \frac{\ddim \dist(\phi(x), S)}{\epsilon^2}
        \right) \\
        & \leq (1 + O(\epsilon)) \floc(\phi(X), F) + O(\rho) \expepsddim \sum_{x \in X} \left(
            \frac{2 \dist(\phi(x), S)}{\epsilon} + 10 \dist(\phi(x), F^*) 
        \right)\\
        & \leq (1 + O(\epsilon)) \median(X', F) + 
        O(\epsilon) \median(X', S) 
        + O(\epsilon) \median(X', F^*)  \\
        & \leq (1 + O(\epsilon)) \median(X, F) + 
        O(\epsilon) \median(X, S) 
        + O(\epsilon) \median(X, F^*) + O(\epsilon) \optkmed(X, Y) &&\text{ \Cref{lemma:kmedian_bounded_centers_new_instance}} \\
        &\leq (1 + O(\epsilon)) \median(X, F) + O(\epsilon) \optkmed(X, Y) \\
        & \leq (1 + O(\epsilon)) \optkmed(X, Y) &&\text{ \Cref{lemma:kmedian_bounded_centers_good_solution_construction}}
    \end{align*}

    We conclude that with probability $0.9$, 
    \[
    \sum_{x \in X'} \ccostport(x, F) \leq (1 + \epsilon) \optkmed(X, Y).
    \]

\end{proof}

\subsection{The Algorithm}
\label{sec:alg_kmedian_bounded_centers}
Our $k$-median algorithm is a modification of the facility location algorithm in \Cref{sec:fl_alg}.

\paragraph{Preprocessing stage.}
By \Cref{sec: constant-approx-algos}, we can compute an $O(1)$-approximate solution $S \subseteq Y$ for $k$-median in near-linear time, when $\ddim(Y)$ is bounded.
The rest of the preprocessing stage is exactly the same as \Cref{sec:fl_alg},
where we 
1) construct the hierarchical decomposition $\decom$ with parameter $\rho = \epsilon^{10}/\ddim^2$,
2) construct the new instance $(X', Y)$, 
3) compute the revealed cluster $C(x)$ for every $x \in X'$ and 
4) construct the proxy set $N_x$ for every $x \in X'$.

\paragraph{Dynamic program.}
The most significant difference of $k$-median from facility location is that the number of facilities is always bounded by $k$, which makes the dynamic program more complicated.
Analogous to the facility location algorithm, each entry of the DP table is represented by a cluster $C$ on $\decom$, together with a configuration $\veca_C = \braket{a^1_C, a^2_C, \dots, a^{|P_C|}_C}, \vecb_C = \braket{b^1_C, b^2_C, \dots, b^{|P_C|}_C}$, which encodes the minimum (portal-respecting) distance from each portal to facilities inside and outside cluster $C$, respectively.
The entry is additionally encoded by a value $\Gamma_C \in [\median(X, S)/n, 2 \median(X, S)]$, which corresponds to the total revealed cost inside cluster $C$.

Given an entry encoded by $(C, \veca_C, \vecb_C, \Gamma_C)$, the value stored in it is the minimum number of facilities required to be placed in $C$, such that the revealed cost inside $C$ is at most $\Gamma_C$.
Formally, 
\begin{align}
    g(C, \veca_C, \vecb_C, \Gamma_C) := \min\Bigg\{
        K \colon \exists F\subset Y \text{ that is consistent with } \veca_C, \vecb_C, \text{ s.t. } 
        |F \cap C| = K \notag \\
        \text{ and }
        \sum_{x \in X' \colon C(x) \subseteq C} \ccostport(x, F) \leq \Gamma_C
    \Bigg\}.
    \label{eqn:kmedian_DP_table_entry_value}
\end{align}
Note that once we have the table entries corresponding to the root node $X'$, it suffices to output the minimum $\Gamma_{X'}$ with $g(X', \veca_{X'}, \vecb_{X'}, \Gamma_{X'}) \leq k$ as the $k$-median value.

As in \Cref{sec:fl_alg}, for every level $\ell$ cluster $C \in \decom_\ell$, we discretize the configuration $a_i, b_i$ to be multiplications of $\rho 2^\ell$ in the range $[0, 2^\ell/\epsilon]$.
Besides, $\Gamma$ is discretized to be powers of $(1 + \tfrac{\epsilon}{2^{O(\ddim)} \log \Delta})$ in the range $[\median(X, S)/n, 2 \median(X, S)]$.

\paragraph{Base case.}
Consider a leaf node $C = \{y\}$ and the corresponding configuration $a, b$ and $\Gamma$.
(We can wlog assume that the only portal of $\{y\}$ is $y$ itself.)
It is easy to check consistency for $a, b$, since there is either a facility on $y$ (then $a = 0$) or no facility on $y$ (then $a = \infty$).

The revealed cost in $\{y\}$ can be easily computed.
Recall in the preprocessing stage, we have already marked the number of copies of $y$ in $X'$, denoted by $w(y)$.
For each of these copies, its connection cost can be read from the configuration, specifically, $\ccostport(y, F) = \min\{a, b\}$.
Therefore, the revealed cost in $\{y\}$ is $w(y) \cdot \ccostport(y, F)$. 

Finally, we compare the revealed cost $w(y) \cdot \ccostport(y, F)$ with $\Gamma$.
If $w(y) \cdot \ccostport(y, F) > \Gamma$, then we set $g(\{y\}, a, b, \Gamma) = \infty$, indicating this is not a consistent configuration.
If $w(y) \cdot \ccostport(y, F) \leq \Gamma$, then we set $g(\{y\}, a, b, \Gamma)$ to be $0$ or $1$, depending on $a = \infty$ or $0$ (i.e., whether there is a facility on $y$).

Next, we show how to compute $g(C, \veca_C, \vecb_C, \Gamma_C)$ for a higher level cluster $C$.
We first give an intuitive but inefficient algorithm.
Then following the idea of~\cite{Cohen-AddadFS21}, we show how to accelerate the algorithm using an auxiliary DP.

\paragraph{Updating the DP table -- an inefficient solution.}
Consider a higher level cluster $C \in \decom_\ell$ and a configuration $\veca_C = \braket{a^1_C, \dots, a^{|P_C|}_C}, \vecb_C = \braket{b^1_C, \dots, b^{|P_C|}_C}$, together with value $\Gamma_C$.
Following \Cref{sec:fl_alg}, the newly revealed cost in $C$ can be computed from $\veca_C$ and $\vecb_C$.
Let \begin{align*}
    \Gamma'_C = \Gamma_C - \sum_{x \in X' \colon C(x) = C} \ccostport(x, F).
\end{align*}
Note that $\Gamma'_C$ is the maximal cost that $C$ is allowed to inherit from its child clusters.

To compute $g(C, \veca_C, \vecb_C, \Gamma)$, our algorithm enumerates all possible combinations of configurations for the child clusters of $C$.
For each of these combinations $\{(D, \veca_D, \vecb_D, \Gamma_D)\}_{D \in \child(C)}$, the consistency of $\veca_D$ and $\vecb_D$ with $\veca_C, \vecb_C$ can be checked in the same way as \Cref{sec:fl_alg} (i.e., criteria \ref{it:criteria1} and \ref{it:criteria2}).
We additionally require that $\sum_{D \in \child(C)} \Gamma_D \leq \Gamma'_C$.
The DP table is then updated by 
\begin{align}
    \qquad g(C, \veca_C, \vecb_C, \Gamma_C) 
    = \min_{\substack{
    \{(D, \veca_D, \vecb_D, \Gamma_D)\}_{D \in \child(C)} \\
    \text{ is consistent} 
    \text{ and } \sum_D \Gamma_D \leq \Gamma'_C
    }}
    \sum_{D \in \child(C)} 
    g(D, \veca_D, \vecb_D, \Gamma_D).
    \label{eqn:kmedian_naive_update}
\end{align}

One issue of \eqref{eqn:kmedian_naive_update} is that there are too many combinations of configurations to enumerate.
To be specific, the total number of combinations for each cluster $C$ is 
\begin{align*}
    \prod_{D \in \child(C)} (\epsilon \rho)^{-2 |P_D|} \log^2 \Delta
    = \left(\frac{\ddim}{\epsilon}\right)^{\left(\frac{\ddim}{\epsilon}\right)^{O(\ddim)}} \log^{2^{O(\ddim)}} \Delta,
\end{align*}
which will introduce an $O(\log^{2^{O(\ddim)}} \Delta)$ factor in the 
total running time.

\paragraph{Acceleration by an auxiliary DP.}
To reduce this dependency of $\log^{2^{O(\ddim)}} \Delta$, we follow~\cite{Cohen-AddadFS21} to use an auxiliary DP to accelerate the computation.
In this auxiliary DP, children of $C$ are given an arbitrary order $D_1, D_2, \dots, D_{|\child(C)|}$.
Every DP table entry is encoded by a cluster $C$, one of its child cluster $D_i$, a set of configurations $\veca_C, \vecb_C, \veca_1, \vecb_1, \veca_2, \vecb_2, \dots, \veca_i, \vecb_i$, and a value $\Gamma$. 
The value in that table entry equals to the minimum number of facilities needed to be placed in $D_{i + 1} \cup D_{i + 2} \cup \dots \cup D_{|\child(C)|}$, such that the cost in $D_{i + 1} \cup D_{i + 2} \cup \dots \cup D_{|\child(C)|}$ is at most $\Gamma$, given the configurations of $C, D_1, \dots, D_i$.

The auxiliary DP table $h(C, D_i, \veca_C,\vecb_C,\veca_1, \vecb_1, \dots, \veca_i, \vecb_i, \Gamma)$ can be updated by enumerating the configuration of the \emph{next} child cluster, namely, $D_{i + 1}$.
\begin{align}
    h(C, D_i, \veca_C, \vecb_C, \veca_1, \vecb_1, \dots, \veca_i, \vecb_i, \Gamma)
    = & \min_{\veca_{i + 1}, \vecb_{i + 1}, \Gamma_{i + 1}}
    \Big\{
        g(D_{i + 1}, \veca_{i + 1}, \vecb_{i + 1}, \Gamma_{i + 1}) \notag \\
        &+ h(C, D_{i + 1}, \veca_C, \vecb_C, \veca_1, \vecb_1, \dots, \veca_{i + 1}, \vecb_{i + 1}, \Gamma - \Gamma_{i + 1})
    \Big\}
    \label{eqn:compute_h}
\end{align}
\begin{align}
    h(C, D_{|\child(C)|}, \veca_C, \vecb_C, \veca_1, \vecb_1, \dots, \veca_{|\child(C)|}, \vecb_{|\child(C)|}, \Gamma) 
    &= \begin{cases}
        0, & \text{if  configuration is consistent;} \\
        \infty, & \text{otherwise.}
    \end{cases} 
    \label{eqn:compute_h_base}
\end{align}
Moreover, we have 
\begin{align}
    g(C, \veca_C, \vecb_C, \Gamma_C)
    = \min_{\veca_1, \vecb_1, \Gamma_1}
    \Big\{
        g(D_1, \veca_1, \vecb_1, \Gamma_1)
        + h(C, D_1, \veca_C, \vecb_C, \veca_1, \vecb_1, \Gamma'_C - \Gamma_1)
    \Big\}.
    \label{eqn:compute_g_C_accelerated}
\end{align}

\subsubsection{Proof of Correctness}
Denote by $\widehat{F} \subseteq Y$ the solution returned by the dynamic program.
We note that $\widehat{F}$ does not necessarily minimize
$
    \sum_{x \in X'} \ccostport(x, F)
$, and it is even non-trivial to see if it gives a $(1 + \epsilon)$-approximation for the minimum cost.
This potential misalignment is mainly due to the discretization tricks used in our dynamic program.
Recall that we only consider powers of $(1 + \epsilon')$ for the value $\Gamma$.
Therefore, when we enumerate $\{\Gamma_D\}$ with $\sum_D \Gamma_D \leq \Gamma$, every $\Gamma_D$ is also a power of $(1 + \epsilon')$.
This discretization leads to the loss of other possible combinations of $\{\Gamma_D\}$,
which could result in a $\pm O(1)$ additive error in the computation of $g(C, \veca_C, \vecb_C, \Gamma_C)$.
The error accumulates with levels and could become as large as $\poly(n)$ on the root cluster $X'$.
Nevertheless, we show the following lemma, whose proof is provided in \Cref{appendix:proof_of_lemma:discretization_is_ok} for completeness.

\begin{restatable}{lemma}{lemmadiscretizationisok}
\label{lemma:discretization_is_ok}
Let $\widehat{F} \subseteq Y$ be the set of facilities returned by the algorithm in \Cref{sec:alg_kmedian_bounded_centers}.
Then 
\[
\sum_{x \in X'} \ccostport(x, \widehat{F}) \leq (1 + O(\epsilon)) \min_{F \subseteq Y, |F| \leq k} \sum_{x \in X'} \ccostport(x, F).
\]
\end{restatable}

Based on \Cref{lemma:discretization_is_ok}, we are ready to prove the correctness of our algorithm.

\begin{proof}[Proof of \Cref{thm:kmedian_bounded_centers} (correctness)]
    We will show that with constant probability, the facility set $\widehat{F}$ computed by our algorithm satisfies $\median(X, \widehat{F}) \leq (1 + \epsilon) \optkmed(X, Y)$.

    By \Cref{lemma:kmedian_bounded_centers_good_solution}, with probability $0.9$, there exists a solution $F \subseteq Y, |F| \leq k$, such that $
    \sum_{x \in X'} \ccostport(x, F) \leq (1 + \epsilon) \optkmed(X, Y).
    $

    Therefore, 
    \begin{align*}
        \median(X, \widehat{F}) 
        &\leq \median(X', \widehat{F}) + \epsilon \optkmed(X, Y) &&\text{By \Cref{lemma:kmedian_bounded_centers_new_instance}, w.p. $0.99$} \\
        &\leq \sum_{x \in X'} \ccostport(x, \widehat{F}) + \epsilon \optkmed(X, Y) \\
        &\leq (1 + O(\epsilon)) \sum_{x \in X'} \ccostport(x, F) + \epsilon \optkmed(X, Y) &&\text{By \Cref{lemma:discretization_is_ok}} \\
        &\leq (1 + O(\epsilon)) \optkmed(X, Y) &&\text{By \Cref{lemma:kmedian_bounded_centers_good_solution}}.
    \end{align*}
    Rescaling $\epsilon$ completes the proof.
\end{proof}

\subsubsection{Time Complexity}

We prove the time complexity of our algorithm is $2^{2^t} \tilde{O}(n + m)$ for $t = O(\ddim \log(\ddim/\epsilon))$.
By \Cref{lem: transform_to_bounded_aspect_ratio,lem:merge_instances}, we can wlog assume that the aspect ratio of $X \cup Y$ is $\Delta = \poly(n, m)$.

\begin{proof}[Proof of \Cref{thm:kmedian_bounded_centers} (time complexity)]
    The time complexity of the preprocessing stage is the same as the facility location algorithm, which is $(\ddim/\epsilon)^{O(\ddim)} \tilde{O}(n + m) \log \Delta$.

    We focus on the time complexity of filling the DP table.
    Fix a cluster $C \in \decom_\ell$.
    We calculate the complexity of filling all table entries regarding $C$.
    
    and a configuration $\veca_C = \braket{a^1_C, \dots, a^{|P_C|}_C}, \vecb_C = \braket{b^1_C, \dots, b^{|P_C|}_C}$, and a value $\Gamma_C \in [\median(X, S)/n, 2\median(X, S)]$.

    The computation of $g(C, \veca_C, \vecb_C, \Gamma_C)$ includes first computing \[\Gamma'_C = \Gamma_C - \sum_{x \in X' \colon C(x) = C} \ccostport(x, F).\]
    By \Cref{sec:fl_alg_time}, this can be done in time 
    \begin{align*}
        \epsilon^{-O(\ddim)} \cdot |\{x \colon C(x) = C\}|.
    \end{align*}
    Then, the algorithm enumerates $\veca_1, \vecb_1, \Gamma_1$.
    Recall that $a_1^p, b_1^p$ are multiplications of $\rho 2^\ell$ in the range $[0, 2^\ell/\epsilon]$, and $\Gamma_1$ is a power of $(1 + \frac{\epsilon}{2^{O(\ddim)} \log \Delta})$ in the range $[\median(X, S)/n, 2 \median(X, S)]$.
    Hence, there are $(\epsilon \rho)^{-2 |P_{D_1}|} \log^2 \Delta = (\epsilon \rho)^{-\rho^{-O(\ddim)}} \log^2 \Delta$ such combinations.
    For each of these combinations, the value of $g(\cdot) + h(\cdot)$ can be calculated in $O(1)$ time, thus a total of $(\epsilon \rho)^{-\rho^{-O(\ddim)}} \log^2 \Delta$ time.

    Analogously, by~\eqref{eqn:compute_h}, each entry in the auxiliary DP 
    $h(C, D_i, \veca_C, \vecb_C, \veca_1, \vecb_1, \dots, \veca_i, \vecb_i, \Gamma)$
    can be computed in time  $(\epsilon \rho)^{-\rho^{-O(\ddim)}} \log^2 \Delta$.

    Therefore, the time complexity of the computation regarding a single cluster $C \in \modifydecom_\ell$ is 
    \begin{align*}
        &\qquad (\epsilon \rho)^{-2 |P_C|} \log^2 \Delta \cdot \Big(
            \epsilon^{-O(\ddim)} \cdot |\{x \colon C(x) = C\}| + 
         (\epsilon \rho)^{-2 |P_{D_1}|} \log^2 \Delta 
        \Big)\\
        &\qquad \qquad + \sum_{i = 1}^{|\child(C)| - 1}
        (\epsilon \rho)^{-2 |P_C| - 2\sum_{j = 1}^i |P_{D_j}|} \log^2 \Delta \cdot (\epsilon \rho)^{-2 |P_{D_{i + 1}}|} \log^2 \Delta \\
        & \leq (\epsilon \rho)^{-\rho^{-O(\ddim)}} \log^4 \Delta \cdot (1 + |\{x \colon C(x) = C\}|).
    \end{align*}

    The total complexity of filling the DP table is 
    \begin{align*}
        &\qquad \sum_{\ell = 0}^L \sum_{C \in \decom_\ell} 
        (\epsilon \rho)^{-\rho^{-O(\ddim)}} \log^4 \Delta \cdot (1 + |\{x \colon C(x) = C\}|) \\
        & = (\epsilon \rho)^{-\rho^{-O(\ddim)}} \log^4 \Delta
        \sum_{\ell = 0}^L \Big(|\decom_\ell|
        + |\{x \in X \colon x \text{ is revealed at level } \ell\}|\Big) \\
        & = \left(\frac{\ddim}{\epsilon}\right)^{\left(\frac{\ddim}{\epsilon}\right)^{O(\ddim)}} \tilde{O}(n + m) \log^4 \Delta.
    \end{align*}

    Finally, the algorithm finds the smallest $\Gamma_{X'}$ such that there exists $\veca_{X'}, \vecb_{X'}$ satisfying $g(X', \veca_{X'}, \vecb_{X'}, \Gamma_{X'}) \leq k$.
    This takes extra $(\epsilon \rho)^{-\rho^{-O(\ddim)}} \log^2 \Delta$ time.

    Combining the analysis above, we conclude that the time complexity of our algorithm is $2^{2^t} \cdot \tilde{O}(n + m) \log^4 \Delta$, for 
    \begin{align*}
        t = O\left(\ddim \log\frac{\ddim}{\epsilon}\right).
    \end{align*}
    \Cref{thm:kmedian_bounded_centers} follows with $\Delta = \poly(n, m)$.
\end{proof}
\section{$k$-Median with Low-dimensional Clients}
\label{sec:kmedian_bounded_clients}

In this section, we extend our results in \Cref{sec:fl_bounded_clients} to the $k$-median problem.
Our main result is the following.

\begin{theorem}\label{thm:kmedian_bounded_clients}
    There is a randomized algorithm that, given as input $\epsilon \in (0, \tfrac{1}{2})$, $n, m, k \in \NN$ and  $(X \cup Y, \dist)$ with $|X| = n, |Y| = m, \ddim(X) \leq \ddim$, computes a $(1 + \epsilon)$-approximation of $k$-median in time 
    $2^{2^t} \cdot \tilde{O}(n + m)$
    with constant success probability, where
    \begin{equation*}
        t \in O\left(\ddim \log\frac{\ddim}{\epsilon}\right).
    \end{equation*}
\end{theorem}

\subsection{Structural Lemmas}

Let $F^* \subseteq Y, |F^*| \leq k$ be the optimal $k$-median solution, and 
Let $S \subseteq Y, |S| \leq k$ be a constant approximate solution.
We define  bad clients $\Bad_X \subseteq X$ and bad facilities $\Bad_S \subseteq S$ w.r.t. $F^*$, $S$ and $\modifydecom$.
We will use the same definition of bad clients and bad facilities in \Cref{def:two_types_of_bad_points}, 
w.r.t. $F^*$ and $S$.
We first state the following lemma which eliminate bad clients in $X$.

\begin{lemma}[New instance]\label{lemma:kmedian_bounded_clients_new_instance}
    Given $(X \cup Y, \dist)$, with $\ddim(X) \leq \ddim$, and a constant approximate solution $S \subseteq Y$ for $k$-median, 
    construct a new 
    (multi-)set of clients $X' \subseteq X$ as $X' := \phi(X)$ for 
    \begin{align*}
        \phi(x) := \begin{cases}
            x, &\text{if } x \notin \Bad_X; \\
            \proj{X}{\proj{S}{x}}, & \text{if } x \in \Bad_X,
        \end{cases}
    \end{align*} 
    namely, $X'$ is constructed from $X$ by moving every bad client $x$ to $\proj{X}{\proj{S}{x}}$.
    Then with probability $0.99$,
    \begin{align*}
        \forall F \subseteq Y, |F| \leq k,  \qquad \Big|\median(X, F) - \median(X', F)\Big|
        \leq \epsilon \optkmed(X, Y).
    \end{align*}
\end{lemma}

The proof of \Cref{lemma:kmedian_bounded_clients_new_instance} is analogous to \Cref{lemma:fl_bounded_centers_new_instance}, and thus is omitted.
The following lemma shows the existence of $(1 + \epsilon)$-approximate portal-respecting solution.
It is similar to \Cref{lemma:fl_bounded_clients_good_solution}, with the only extra requirement that $|F| \leq k$.

\begin{lemma}[Good portal-respecting solution]\label{lemma:kmedian_bounded_clients_good_solution}
    Let $(X \cup Y, \dist)$ with $\ddim(X) \leq \ddim$ be a $k$-median instance, $S \subseteq Y$ be a constant approximate solution and $\modifydecom$ be the hierarchical decomposition on $X \cup Y$ in \Cref{lemma:new_decomposition_for_general_metric} with portal parameter $\rho = \epsilon^{10}/\ddim^2$.
    Let $X' \subseteq X$ be the new client set constructed by \Cref{lemma:kmedian_bounded_clients_new_instance}.
    Then with probability $0.9$, there exists a solution $F \subseteq Y, |F| \leq k$, such that 
    \begin{equation}\label{eqn:kmedian_bounded_clients_good_solution}
        \sum_{x \in X'} \dportP(x, F)
        \leq (1 + \epsilon) \optkmed(X, Y).
    \end{equation}
\end{lemma}

To prove \Cref{lemma:kmedian_bounded_clients_good_solution}, we use the following lemma, which claims the existence of a $(1 + \epsilon)$-approximate solution that contains all bad facilities in $S$.
\Cref{lemma:kmedian_bounded_clients_good_solution_construction} is essentially a combination of~\cite[Claims 20 and 21]{Cohen-AddadFS21}.
We provide the proof in \Cref{appendix:proof_of_kmedian_bounded_clients_construction} for completeness.

\begin{restatable}{lemma}{lemmakmedianboundedclientsgoodsolutionconstruction}
\label{lemma:kmedian_bounded_clients_good_solution_construction}
Let $(X \cup Y, \dist)$ with $\ddim(X) \leq \ddim$ be a $k$-median instance, and $S \subseteq Y$ be a constant approximate solution.
Then there exists a set of facilities $F \subseteq Y$, which satisfies the following properties.
\begin{enumerate}[label=(\alph*)]
    \item $\Bad_S \subseteq F$.
    \item For every $x \in X$, $\dist(x, F) \leq 6 \dist(x, F^*) + 5 \dist(x, S)$.
    \item With probability $0.99$, $|F| \leq k$ and $\median(X, F) \leq (1 + \epsilon) \optkmed(X, Y)$.
\end{enumerate}
\end{restatable}

We prove \Cref{lemma:kmedian_bounded_clients_good_solution} based on \Cref{lemma:kmedian_bounded_clients_good_solution_construction}.
In fact, we show that the solution $F$ in \Cref{lemma:kmedian_bounded_clients_good_solution_construction} satisfies \eqref{eqn:kmedian_bounded_clients_good_solution}.
The proof is very similar to \Cref{lemma:fl_bounded_clients_good_solution}.

\begin{proof}[Proof of \Cref{lemma:kmedian_bounded_clients_good_solution}]
    Recall that $F^* \subseteq Y$ is the optimal solution and that $\Bad_S \subseteq S$ is the set of bad facilities in $S$.
    We show that with probability $0.9$,  the solution $F$ in \Cref{lemma:kmedian_bounded_clients_good_solution_construction} satisfies \eqref{eqn:kmedian_bounded_clients_good_solution}.
    By \Cref{lemma:kmedian_bounded_clients_new_instance}, $X' = \phi(x)$.
    We thus rewrite the LHS of \eqref{eqn:kmedian_bounded_clients_good_solution} as 
    \[
    \sum_{x \in X} \dportP(\phi(x), F).
    \]

    To bound the connection cost $\dportP(\phi(x), F)$ for each $x \in X$, we consider the following cases.

    \paragraph{Case 1: $x$ is a bad client, and $\proj{S}{x}$ is a bad facility.}
    By the construction of $X'$, $\phi(x) = \proj{X}{\proj{S}{x}}$. 
    Since $\proj{S}{x}$ is a bad facility, it is in $\Bad_S$, and thus in $F$.
    If $\proj{S}{x} \in X$, then $\proj{X}{\proj{S}{x}} = \proj{S}{x}$, and thus $\dportP(\phi(x), F) = 0$.
    If $\proj{S}{x} \in Y \setminus X$, then by the construction of $\modifydecom$, the set $\{\proj{X}{\proj{S}{x}}, \proj{S}{x}\}$ is cut exactly at level $h(\proj{S}{x})$, the level where $\proj{S}{x}$ is a leaf node on $\modifydecom$.
    By \Cref{lemma:detour_complicated}, 
    \begin{align}
        \dportP(\phi(x), F)
        &\leq \dportP(\proj{X}{\proj{S}{x}}, \proj{S}{x}) \notag\\
        &\leq \dist(\proj{X}{\proj{S}{x}}, \proj{S}{x}) + O(\rho) 2^{h(\proj{S}{x})} \notag\\
        &\leq \dist(x, \proj{S}{x}) + O(\rho) \frac{\dist(\proj{S}{x}, X)}{\sqrt{\rho}} 
        \notag \\
        &\leq (1 + O(\sqrt{\rho})) \dist(x, \proj{S}{x}) \notag\\
        &\leq 2 \dist(x, S). \label{eqn:kmedian_case_1}
    \end{align}

    \paragraph{Case 2: $x$ is a bad client, and $\proj{S}{x}$ is not a bad facility.}
    By the construction of $X'$, $\phi(x) = \proj{X}{\proj{S}{x}}$.
    Our plan is to upper bound the highest level where $\{\phi(x), \proj{X}{\proj{F}{\phi(x)}}\}$ is cut w.r.t. $\modifydecom$.

    Since $\proj{S}{x}$ is not a bad facility, by \Cref{def:badly_cut,def:two_types_of_bad_points}, the ball 
    \begin{align*}
        &B_X\Big(\proj{X}{\proj{S}{x}}, 100 \dist(\proj{X}{\proj{S}{x}}, F^*) + 100 \dist(\proj{X}{\proj{S}{x}}, S)\Big) \\
        = &B_X\Big(\phi(x), 100 \dist(\phi(x), F^*) + 100 \dist(\phi(x), S)\Big)
    \end{align*}
    is cut at level at most $\ell \leq \log(100 \dist(\phi(x), F^*) + 100 \dist(\phi(x), S)) + \epsddim$.
    Note that 
    \begin{align*}
        \dist(\phi(x), \proj{X}{\proj{F}{\phi(x)}})
        &\leq \dist(\phi(x), \proj{F}{\phi(x)}) + \dist(\proj{F}{\phi(x)}, \proj{X}{\proj{F}{\phi(x)}}) \\
        &\leq 2 \dist(\phi(x), F) \\
        &\leq 12 \dist(\phi(x), F^*) + 10 \dist(\phi(x), S) &&\text{By \Cref{lemma:kmedian_bounded_clients_good_solution_construction}.}
    \end{align*}
    Hence, $\proj{X}{\proj{F}{\phi(x)}} \in B_X(\phi(x), 100 \dist(\phi(x), F^*) + 100 \dist(\phi(x), S))$.
    Therefore, the highest level where $\{\phi(x), \proj{X}{\proj{F}{\phi(x)}}\}$ is cut is at most 
    $\ell \leq \log(100 \dist(\phi(x), F^*) + 100 \dist(\phi(x), S)) + \epsddim$.
    By \Cref{lemma:detour_by_projection}, 
    \begin{align}
        \dportP(\phi(x), F)
        &\leq (1 + \sqrt{\rho}) \dist(\phi(x), F) + O(\sqrt{\rho}) 2^\ell \notag \\
        &\leq (1 + \sqrt{\rho}) \dist(\phi(x), F) + O(\sqrt{\rho}) \expepsddim (\dist(\phi(x), F^*) + \dist(\phi(x), S)) \notag \\
        &\leq (1 + \sqrt{\rho}) \dist(\phi(x), F) + O(\epsilon) (\dist(\phi(x), F^*) + \dist(\phi(x), S)). \label{eqn:kmedian_case_2}
    \end{align}

    \paragraph{Case 3: $x$ is not a bad client, and $\proj{S}{x}$ is a bad facility.}
    In this case $\phi(x) = x$, and the ball $B_X(x, \dist(x, S) / \epsilon)$ is cut at level at most $\ell \leq \log(\dist(x, S)/\epsilon) + \epsddim$.
    On the other hand, since $\proj{S}{x} \in \Bad_S$, we have $\proj{S}{x} \in F$ by \Cref{lemma:kmedian_bounded_clients_good_solution_construction}.
    Therefore, 
    \[
        \dist(x, \proj{X}{\proj{F}{x}})
        \leq 2 \dist(x, F) 
        \leq 2 \dist(x, S).
    \]
    We have $\proj{X}{\proj{F}{x}} \in B_X(x, \dist(x, S) / \epsilon)$.
    Hence, the highest level where $\{x, \proj{X}{\proj{F}{x}}\}$ is cut is at most  
    $\ell \leq \log(\dist(x, S)/\epsilon) + \epsddim$.

    By \Cref{lemma:detour_by_projection},
    \begin{align}
        \dportP(\phi(x), F) 
        &= \dportP(x, F) \notag \\
        &\leq (1 + O(\sqrt{\rho})) \dist(x, F) + O(\sqrt{\rho}) 2^\ell \notag \\
        &\leq (1 + O(\sqrt{\rho})) \dist(x, F) + O(\sqrt{\rho}) \expepsddim \frac{\dist(x, S)}{\epsilon} \notag \\
        & \leq (1 + O(\sqrt{\rho})) \dist(x, F) + O(\epsilon) \dist(x, S). \label{eqn:kmedian_case_3}
    \end{align}

    \paragraph{Case 4: $x$ is not a bad client, and $\proj{S}{x}$ is not a bad facility.}
    In this case $\phi(x) = x$.
    Following the proof in \Cref{sec:fl_bounded_centers,sec:fl_bounded_clients,sec:kmedian_bounded_centers}, consider the following two sub-cases:

    If $\dist(x, F^*) \leq \dist(x, S) / (20 \epsilon)$, 
    first observe that 
    the ball $B_X(x, \dist(x, S) / \epsilon)$ is cut at level at most $\ell \log(\dist(x, S)/\epsilon) + \epsddim$.
    On the other hand,
    \[
        \dist(x, \proj{X}{\proj{F}{x}}) 
        \leq 2 \dist(x, F) 
        \leq 12 \dist(x, F^*) + 10 \dist(x, S)
        \leq \frac{\dist(x, S)}{\epsilon}.
    \]
    Hence $\proj{X}{\proj{F}{x}} \in B_X(x, \dist(x, S) / \epsilon)$.
    Therefore, the highest level where $\{x, \proj{X}{\proj{F}{x}}\}$ is cut is at most  
    $\ell \log(\dist(x, S)/\epsilon) + \epsddim$.
    Using a similar argument as Case 3, we can show that 
    \begin{equation}\label{eqn:kemdian_case_4.1}
        \dportP(\phi(x), F)
        \leq (1 + O(\sqrt{\rho})) \dist(x, F) + O(\epsilon) \dist(x, S).
    \end{equation}

    If $\dist(x, F^*) > \dist(x, S) / (20 \epsilon)$, we plan to utilize the fact that 
    $
        B_X(\proj{X}{\proj{S}{x}}, 100 \dist(\proj{X}{\proj{S}{x}}, F^*) + 100 \dist(\proj{X}{\proj{S}{x}}, S))
    $ 
    is cut at level at most $\ell \leq \log(100 \dist(\proj{X}{\proj{S}{x}}, F^*) + 100 \dist(\proj{X}{\proj{S}{x}}, S)) + \epsddim$, by showing 
    both $x$ and $\proj{X}{\proj{F}{x}}$ is contained in this ball.

    First observe that 
    \begin{align*}
        \dist(x, F^*) 
        &\leq \dist(x, S) + \dist(\proj{S}{x}, F^*) \\
        &\leq \dist(x, S) + \dist(\proj{S}{x}, \proj{X}{\proj{S}{x}}) + \dist(\proj{X}{\proj{S}{x}}, F^*) \\
        &\leq 2 \dist(x, S) + \dist(\proj{X}{\proj{S}{x}}, F^*) \\
        & \leq 40 \epsilon \dist(x, F^*) + \dist(\proj{X}{\proj{S}{x}}, F^*).
    \end{align*}
    Hence, $\dist(x, F^*) \leq \frac{1}{1 - 40 \epsilon} \dist(\proj{X}{\proj{S}{x}}, F^*)$.
    
    Furthermore, note that 
    \begin{align*}
        \dist(x, \proj{X}{\proj{S}{x}})
        \leq 2 \dist(x, S) 
        < 40 \epsilon \dist(x, F^*)
        \leq \frac{40 \epsilon}{1 - 40 \epsilon}  \dist(\proj{X}{\proj{S}{x}}, F^*)
        \leq 100 \dist(\proj{X}{\proj{S}{x}}, F^*), 
    \end{align*}
    and that 
    \begin{align*}
        \dist(\proj{X}{\proj{F}{x}}, \proj{X}{\proj{S}{x}})
        &\leq \dist(x, \proj{X}{\proj{F}{x}}) + \dist(x, \proj{X}{\proj{S}{x}}) \\
        &\leq 2 \dist(x, F) + 2 \dist(x, S) \\
        &\leq 12 \dist(x, F^*) + 12 \dist(x, S) &&\text{By \Cref{lemma:kmedian_bounded_clients_good_solution_construction}} \\
        &\leq 12(1 + 20 \epsilon) \dist(x, F^*)\\
        & \leq \frac{12(1 + 20 \epsilon)}{1 - 40\epsilon} \dist(\proj{X}{\proj{S}{x}}, F^*) \\
        & \leq 100 \dist(\proj{X}{\proj{S}{x}}, F^*),
    \end{align*}
    for $\epsilon \leq 1/100$.
    We conclude that $x, \proj{X}{\proj{F}{x}} \in B_X(\proj{X}{\proj{S}{x}}, 100 \dist(\proj{X}{\proj{S}{x}}, F^*) + 100 \dist(\proj{X}{\proj{S}{x}}, S))$.
    Thus the highest level where $\{x, \proj{X}{\proj{F}{x}}\}$ is cut is at most  
    $\ell \leq \log(100 \dist(\proj{X}{\proj{S}{x}}, F^*) + 100 \dist(\proj{X}{\proj{S}{x}}, S)) + \epsddim$.
    Similar to Case 2, we can show that 
    \begin{align}
        \dportP(\phi(x), F)
        &\leq (1 + O(\sqrt{\rho})) \dist(x, F) + O(\epsilon) \Big(\dist(\proj{X}{\proj{S}{x}}, F^*) + \dist(\proj{X}{\proj{S}{x}}, S) \Big) \notag \\
        &\leq (1 + O(\sqrt{\rho})) \dist(x, F) + O(\epsilon) \dist(x, F^*) + O(\epsilon) \dist(x, S). \label{eqn:kmedian_case_4.2}
    \end{align}

    Combining \eqref{eqn:kemdian_case_4.1} with \eqref{eqn:kmedian_case_4.2}, we conclude that in this case, 
    \begin{equation}\label{eqn:kmedian_case_4}
        \dportP(\phi(x), F)
        \leq (1 + O(\sqrt{\rho})) \dist(x, F) + O(\epsilon) \dist(x, F^*) + O(\epsilon) \dist(x, S).
    \end{equation}

    \medskip

    Combining above four cases (i.e., \eqref{eqn:kmedian_case_1}, \eqref{eqn:kmedian_case_2}, \eqref{eqn:kmedian_case_3} and \eqref{eqn:kmedian_case_4}), we have 
    \begin{align}
        &\qquad \sum_{x \in X} \dportP(\phi(x), F) \notag \\
        &\leq \sum_{\substack{x \in \Bad_X, \\ \proj{S}{x} \in \Bad_S}} 2 \dist(x, S) \notag \\
        & \hspace{5em} + \sum_{\substack{x \notin \Bad_X \text{ or} \\ \proj{S}{x} \notin \Bad_S}} \Big(
            (1 + O(\sqrt{\rho})) \dist(\phi(x), F)
            + O(\epsilon) \dist(\phi(x), F^*)
            + O(\epsilon) \dist(\phi(x), S)
        \Big) \notag\\
        &\leq \sum_{\substack{x \in \Bad_X, \\ \proj{S}{x} \in \Bad_S}} 2 \dist(x, S) 
        + (1 + O(\sqrt{\rho}))\median(X', F) 
        + O(\epsilon)\median(X', F^*) 
        + O(\epsilon) \median(X', S) \label{eqn:kmedian_combine_4_cases}
    \end{align}

    By \Cref{lemma:kmedian_bounded_clients_new_instance},
    \begin{align}
        \text{with probability } 0.99, \quad 
        &\median(X', F) \leq \median(X, F) + O(\epsilon) \optkmed(X, Y), \notag \\ 
        &\median(X', F^*) \leq \median(X, F^*) + O(\epsilon) \optkmed(X, Y) = (1 + O(\epsilon)) \optkmed(X, Y) \notag\\
        &\median(X', S) \leq \median(X, S) + O(\epsilon) \optkmed(X, Y)
        \leq O(1) \optkmed(X, Y) \label{eqn:kmedian_first_0.99}
    \end{align}
    hold simultaneously.

    By \Cref{lemma:kmedian_bounded_clients_good_solution_construction},
    \begin{equation}\label{eqn:kmedian_second_0.99}
        \text{with probability } 0.99, \qquad |F| \leq k \text{ and } 
    \median(X, F) \leq (1 + O(\epsilon)) \optkmed(X, Y).
    \end{equation}

    Finally, observe that 
    \begin{align*}
        \E \sum_{\substack{x \in \Bad_X, \\ \proj{S}{x} \in \Bad_S}} 2 \dist(x, S)
        &=  \sum_{x \in X} 2 \dist(x, S) \cdot \Pr[x \in \Bad_X, \proj{S}{x} \in \Bad_S] \\
        & \leq O(\epsilon) \cdot \sum_{x \in X} 2 \dist(x, S) &&\text{By \Cref{lemma:badly_cut_ambient}} \\
        & \leq O(\epsilon) \cdot \median(X, S) \\
        & \leq O(\epsilon) \cdot \optkmed(X, Y).
    \end{align*}
    Applying Markov's inequality, we have 
    \begin{equation}\label{eqn:kmedian_third_0.99}
        \text{with probability } 0.99, \qquad  
        \sum_{\substack{x \in \Bad_X, \\ \proj{S}{x} \in \Bad_S}} 2 \dist(x, S) 
        \leq O(\epsilon) \optkmed(X, Y).
    \end{equation}

    Combining \eqref{eqn:kmedian_first_0.99}, \eqref{eqn:kmedian_second_0.99} and \eqref{eqn:kmedian_third_0.99} with \eqref{eqn:kmedian_combine_4_cases}, we conclude that with probability $0.97$, $|F| \leq k$ and  
    \begin{align*}
        \sum_{x \in X} \dportP(\phi(x), F)
        &\leq O(\epsilon) \optkmed(X, Y)
        + (1 + O(\epsilon)) \median(X, F) + O(\epsilon) \optkmed(X, Y) \\
        &\leq (1 + O(\epsilon)) \optkmed(X, Y).
    \end{align*}
    Rescaling $\epsilon$ concludes the proof.
    
\end{proof}

\subsection{The Algorithm}
\label{sec:alg_kmedian_bounded_clients}

Our $k$-median algorithm is a dynamic program on top of the new decomposition $\modifydecom$.
The framework is similar to the algorithm in \Cref{sec:alg_kmedian_bounded_centers}, with the complication of handling ornaments on $\modifydecom$.

\paragraph{Preprocessing stage.}
By \Cref{sec: constant-approx-algos}, we can compute an $O(1)$-approximate solution $S \subseteq Y$ for $k$-median in near-linear time, when $\ddim(X)$ is bounded.
The rest of the preprocessing stage is exactly the same as \Cref{sec:alg_fl_bounded_clients}, where we
1) use \Cref{alg:modified_decomposition} to compute the decomposition $\modifydecom$ on top of $X \cup Y$ with scaling parameter $\rho = \epsilon^{10}/\ddim^2$, 
and 2) construct the new dataset $X'$ by moving every bad client $x \in X$ to $\proj{X}{\proj{S}{x}}$.

\paragraph{Dynamic program.}
Analogous to \Cref{sec:alg_kmedian_bounded_centers}, our algorithm consists of a main DP $g(\cdot)$ and an auxiliary DP $h(\cdot)$.

In the main DP, each entry in the DP table is encoded by $(C, \veca_C, \vecb_C, \Gamma_C)$.
The configurations $\veca_C = \braket{a^1_C, a^2_C, \dots, a^{|P_C|}_C}, \ \vecb_C = \braket{b^1_C, b^2_C, \dots, b^{|P_C|}_C}$ encode the minimum (portal-respecting) distance from each portal to facilities inside and outside cluster $C$, respectively.
The real value $\Gamma_C \in [\median(X, S)/n, 2 \median(X, S)]$ corresponds to the total connection cost inside $C$.
The DP value $g(C, \veca_C, \vecb_C, \Gamma_C)$ equals to the minimum number of facilities required to be placed in $C$, such that the connection cost inside $C$ is at most $\Gamma_C$.
Formally, 
\begin{align}
    g(C, \veca_C, \vecb_C, \Gamma_C) := \min\Bigg\{
        K \colon \exists F \subseteq Y 
        &\text{ that is consistent with } 
        \veca_C, \vecb_C, \notag \\
        &\text{ s.t. }
        |F \cap C| = K 
        \text{ and } 
        \sum_{x \in X' \cap C} \dportP(x, F) \leq \Gamma_C
    \Bigg\}. \label{eqn:kmedian_bounded_clients_main_DP}
\end{align}
Note that once we have the table entries corresponding to the root node $X'$, it suffices to output the minimum $\Gamma$ with $g(X', \veca, \vecb, \Gamma) \leq k$ as the $k$-median value.

To define the auxiliary DP, as in \Cref{sec:alg_kmedian_bounded_centers}, we need an order for the non-ornament child clusters of $C$.
Recall that $\Nchild(C)$ is the set of non-ornament child clusters of $C$ and $\Schild(C)$ is the set of ornament child clusters of $C$.
Consider an arbitrary order of $\Nchild(C)$ as $D_1, D_2, \dots, D_{|\Nchild(C)|}$.
Each table entry of the auxiliary DP $h$ is encoded by a cluster $C$, its $i$-th non-ornament child cluster $D_i$, a set of configurations $\veca_C, \vecb_C, \veca_1, \vecb_1, \dots, \veca_i, \vecb_i$, and a value $\Gamma$.
The value in that table entry 
$h(C, D_i, \veca_C, \vecb_C, \veca_1, \vecb_1, \dots, \veca_i, \vecb_i, \Gamma)$
equals to the minimum number of facilities needed to be placed in 
\[
C \setminus \bigcup_{j = 1}^i D_j 
= \bigcup_{j = i+ 1}^{|\Nchild(C)|} D_{j} \cup \bigcup_{\{y\} \in \Schild(C)} \{y\}, 
\] 
such that the cost in $C \setminus \bigcup_{j = 1}^i D_j$ is at most $\Gamma$, given the configurations of $C, D_1, \dots, D_i$.

For a level $\ell$ cluster $C \in \modifydecom_\ell$, we discretize the configurations $a_p, b_p$ to be multiplications of $\rho 2^\ell$ in the range $[0, 2^\ell/\epsilon]$.
Besides, $\Gamma$ is discretized to be powers of $(1 + \tfrac{\epsilon}{2^{O(\ddim)} \log \Delta})$ in the range $[\median(X, S)/n, 2 \median(X, S)]$.

\paragraph{Base case.}
The base case of the dynamic program corresponds to the leaf nodes of $\modifydecom$ which are in $X'$.
Consider $x \in X'$ and its corresponding node $\{x\} \in \modifydecom_0$, together with a configuration $(a, b, \Gamma)$. 
(We can wlog assume $\{x\}$ has itself as its only portal.)
It is easy to check if $(a, b)$ is valid.
Recall that $a$ encodes the distance from $x$ to the closest facility in $\{x\}$.
Therefore, if $x \in Y \cap X'$, then $a$ is either $0$ or $\infty$;
if $x \in X' \setminus Y$, then $a$ must be $\infty$.
We set $g(\{x\}, a, b, \Gamma) = \infty$ if $(a, b)$ is invalid.

To further check if $\Gamma$ is consistent with $(a, b)$, we can compare the connection cost of $x$ with $\Gamma$.
Specifically, the connection cost of $x$ is $\dportP(x, F) = \min\{a, b\}$ multiplied by $w(x)$, the number of copies of $x$ in $X'$.
If $\Gamma < w(x) \cdot \min\{a, b\}$, i.e., the connection cost inside $\{x\}$ already exceeds $\Gamma$, 
then $\Gamma$ is inconsistent with $(a, b)$.
We set $g(\{x\}, a, b, \Gamma) = \infty$

If $(a, b, \Gamma)$ is consistent, we set $g(a, b, \Gamma)$ to be the number of facilities in $\{x\}$, which is either $0$ or $1$, depending on whether $a$ is $\infty$ or $0$.

\paragraph{Updating the DP table.}
Consider a higher level cluster $C \in \modifydecom_\ell$ with $\ell \geq 1$, together with its configuration $\veca_C = \braket{a^1_C, \dots, a^{|P_C|}_C}, \vecb_C = \braket{b^1_C, \dots, b^{|P_C|}_C}$ and value $\Gamma_C$.
To compute $g(C, \veca_C, \vecb_C, \Gamma_C)$, the algorithm enumerates the configuration of its first non-ornament child $D_1$; the number of facilities inside $C$ equals to the number of facilities in $D_1$ plus the number of facilities in $C \setminus D_1$.
Formally, 
\begin{equation}\label{eqn:kmedian_bounded_clients_g}
    g(C, \veca_C, \vecb_C, \Gamma_C) := \min_{\veca_1, \vecb_1, \Gamma_1} \Big\{
    g(D_1, \veca_1, \vecb_1, \Gamma_1) + h(C, D_1, \veca_C, \vecb_C, \veca_1, \vecb_1, \Gamma_C - \Gamma_1)
\Big\}.
\end{equation}

For every $1 \leq i \leq |\Nchild(C)| - 1$, to compute the auxiliary DP entry $h(C, D_i, \veca_C, \vecb_C, \veca_1, \vecb_1, \dots, \veca_i, \vecb_i, \Gamma)$, our algorithm enumerates the configuration of $C$'s $(i + 1)$-th non-ornament child cluster $D_{i + 1}$.
The number of facilities inside $C \setminus \bigcup_{j = 1}^i D_j$ equals to the number of facilities in $D_{i + 1}$ plus the number of facilities in $C \setminus \bigcup_{j = 1}^{i + 1} D_j$.
Formally, 
\begin{align}
    h(C, D_i, \veca_C, \vecb_C, \veca_1, \vecb_1, \dots, \veca_i, \vecb_i, \Gamma)
    = & \min_{\veca_{i + 1}, \vecb_{i + 1}, \Gamma_{i + 1}}
    \Big\{
        g(D_{i + 1}, \veca_{i + 1}, \vecb_{i + 1}, \Gamma_{i + 1}) \notag \notag \\
        &+ h(C, D_{i + 1}, \veca_C, \vecb_C, \veca_1, \vecb_1, \dots, \veca_{i + 1}, \vecb_{i + 1}, \Gamma - \Gamma_{i + 1})
    \Big\}.
    \label{eqn:kmedian_bounded_clients_h}
\end{align}

Finally, we discuss how to compute 
\[
    h(C, D_{|\Nchild(C)|}, \veca_C, \vecb_C, \veca_1, \vecb_1, \dots, \veca_{|\Nchild(C)|}, \vecb_{|\Nchild(C)|}, \Gamma).
\]
This corresponds to the number of facilities in $\bigcup_{\{y\} \in \Schild(C)} \{y\}$.
We use the (almost) same 4-step procedure as \Cref{sec:alg_fl_bounded_clients}.

\begin{itemize}
    \item \textbf{Step 0: Initialize.}
    Maintain a universe $U$, which is set to $\emptyset$ initially.
    For every ornament $\{y\} \in \Schild(C)$, maintain a set $R_y$, which is set to $\emptyset$ initially.

    Roughly speaking, $U$ is the set of portals that need to be served by a facility in $\Schild(C)$, 
    and $R_y$ contains the portals which can be served by $y$.

    \item \textbf{Step 1: Check consistency for $p \in P_C$.}
    For every $p \in P_C$, the algorithm finds $D \in \Nchild(C)$ and $q \in P_D$, such that $a_C^p = a_D^q + \dist(p, q)$.
    If such $D$ and $q$ exists, then check for the next $p \in P_C$.

    Otherwise, such $D$ and $q$ do not exist. 
    We add the pair $(C, p)$ to $U$, indicating portal $p$ needs to be served by a facility in $\Schild(C)$.
    Our algorithm then finds all ornaments $\{y\} \in \Schild(C)$ that satisfy $a_C^p = \dist(p, y)$, and adds the pair $(C, p)$ to $R_y$, indicating $p$ can be served by $y$.
    If no such $y$ exists, we claim that the configuration
    $(\veca_C, \vecb_C, \veca_1, \vecb_1, \dots, \veca_{|\Nchild(C)|}, \vecb_{|\Nchild(C)|})$
    is inconsistent, and return $h(C, D_{|\Nchild(C)|}, \veca_C, \vecb_C, \veca_1, \vecb_1, \dots, \veca_{|\Nchild(C)|}, \vecb_{|\Nchild(C)|}, \Gamma) = \infty$.
    
    \item \textbf{Step 2: Check consistency for $D \in \Nchild(C)$ and $q \in P_D$.}
    For every $D \in \Nchild(C)$ and $q \in P_D$, our algorithm first tries to find $p \in P_C$, such that $b^q_D = b_C^p + \dist(p, q)$.
    If such $p$ exists, it means $q$ connects to a facility outside $D$ first via $p$, then to a facility outside $C$.
    We can then check for the next $D$ and $q$.

    If such $p$ does not exist, our algorithm tries to find a non-ornament cluster $D' \in \Nchild(C), D' \neq D$ and a portal $q' \in P_{D'}$, such that $b^q_D = a^{q'}_{D'} + \dist(q, q')$.
    If such $D', q'$ exists, it means $q$ connects to a facility outside $D$ first via $q'$, then to a facility inside $D'$.
    We can then check for the next $D$ and $q$.

    Assume there exist neither $p \in P_C$, such that $b^q_D = b_C^p + \dist(p, q)$, nor $D' \in \Nchild(C), D' \neq D$ and $q' \in P_{D'}$, such that $b^q_D = a^{q'}_{D'} + \dist(q, q')$.
    We add the pair $(D, q)$ to the universe $U$, indicating that portal $q \in P_D$ needs to be served by a facility in $\Schild(C)$.
    Our algorithm then finds all ornaments $\{y\} \in \Schild(C)$ that satisfy $b_D^q = \dist(q, y)$, and adds $(D, q)$ to $R_y$, indicating $q$ can be served by $y$.
    If no such $y$ exists, we claim that the configuration
    $(\veca_C, \vecb_C, \veca_1, \vecb_1, \dots, \veca_{|\Nchild(C)|}, \vecb_{|\Nchild(C)|})$
    is inconsistent, 
    and return $h(C, D_{|\Nchild(C)|}, \veca_C, \vecb_C, \veca_1, \vecb_1, \dots, \veca_{|\Nchild(C)|}, \vecb_{|\Nchild(C)|}, \Gamma) = \infty$.

    \item \textbf{Step 3: Compute the number of facilities in $\Schild(C)$.}
    At this stage, 
    the universe $U$ is a collection of (cluster, portal) pairs:
    \begin{equation}\label{eqn:kmedian_universe}
        U \subseteq \{C\} \times P_C \cup \bigcup_{D \in \Nchild(C)} \{D\} \times P_D.
    \end{equation}
    For every $\{y\} \in \Schild(C)$, $R_y$ is a subset of $U$, and $\bigcup_{y \in \Schild(C)} R_y = U$.
    We solve the unweighted set cover problem for universe $U$ and set system $\{R_y \colon \{y\}\in \Schild(C)\}$.
    $h(C, D_{|\Nchild(C)|}, \veca_C, \vecb_C, \veca_1, \vecb_1, \dots, \veca_{|\Nchild(C)|}, \vecb_{|\Nchild(C)|}, \Gamma)$
    is set to 
    the optimal cost for this set cover instance, i.e., the minimum number of sets required to cover the universe $U$.
\end{itemize}

\subsubsection{Proof of Correctness}

Let $\widehat{F} \subseteq Y$ be the set of facilities returned by the algorithm.
Similar to \Cref{lemma:discretization_is_ok}, we can show the following lemma, whose proof is deferred to \Cref{appendix:proof_of_kmedianboundedclientsdiscretizationisok}:
\begin{restatable}{lemma}{lemmakmedianboundedclientsdiscretizationisok}
\label{lemma:kmedian_bounded_clients_discretization_is_ok}
    Let $\widehat{F} \subseteq Y$ be the set of facilities returned by the algorithm in \Cref{sec:alg_kmedian_bounded_clients}.
    Then \[
    \sum_{x \in X'} \dportP(x, \widehat{F})
    \leq (1 + O(\epsilon)) \min_{F \subseteq Y, |F| \leq k}
    \sum_{x \in X'} \dportP(x, F).
    \]
\end{restatable}

Based on \Cref{lemma:kmedian_bounded_clients_discretization_is_ok}, we prove the correctness of our algorithm.

\begin{proof}[Proof of \Cref{thm:kmedian_bounded_clients} (correctness)]
    We will show that with constant probability, the facility set $\widehat{F}$ computed by our algorithm satisfies $\median(X, \widehat{F}) \leq (1 + \epsilon) \optkmed(X, Y)$.

    By \Cref{lemma:kmedian_bounded_clients_good_solution}, with probability $0.9$, there exists a solution $F \subseteq Y, |F| \leq k$, such that $
    \sum_{x \in X'} \dportP(x, F) \leq (1 + \epsilon) \optkmed(X, Y).
    $

    Therefore, 
    \begin{align*}
        \median(X, \widehat{F}) 
        &\leq \median(X', \widehat{F}) + \epsilon \optkmed(X, Y) &&\text{By \Cref{lemma:kmedian_bounded_clients_new_instance}, w.p. $0.99$} \\
        &\leq \sum_{x \in X'} \dportP(x, \widehat{F}) + \epsilon \optkmed(X, Y) \\
        &\leq (1 + \epsilon) \sum_{x \in X'} \dportP(x, F) + \epsilon \optkmed(X, Y) &&\text{By \Cref{lemma:kmedian_bounded_clients_discretization_is_ok}} \\
        &\leq (1 + O(\epsilon)) \optkmed(X, Y) &&\text{By \Cref{lemma:kmedian_bounded_clients_good_solution}}.
    \end{align*}
    Rescaling $\epsilon$ completes the proof.
\end{proof}

\subsubsection{Time Complexity}

We prove the time complexity of our algorithm is $2^{2^t} \tilde{O}(n + m)$ for $t = O(\ddim \log(\ddim/\epsilon))$.
By \Cref{lem: transform_to_bounded_aspect_ratio,lem:merge_instances}, we can wlog assume that the aspect ratio of $X \cup Y$ is $\Delta = \poly(n, m)$.

\begin{proof}[Proof of \Cref{thm:kmedian_bounded_clients} (time complexity)]
    The preprocessing stage is the same as \Cref{sec:alg_fl_bounded_clients}, and thus has time complexity $(\ddim/\epsilon)^{O(\ddim)} \tilde{O}(n + m) \log \Delta$.

    We focus on the time complexity of filling the DP table.
    Fix a cluster $C \in \modifydecom_\ell$ and a configuration $\veca_C = \braket{a^1_C, \dots, a^{|P_C|}_C}, \vecb_C = \braket{b^1_C, \dots, b^{|P_C|}_C}$, and a value $\Gamma_C \in [\median(X, S)/n, 2\median(X, S)]$.

    The computation of $g(C, \veca_C, \vecb_C, \Gamma_C)$ includes an enumeration of $\veca_1, \vecb_1, \Gamma_1$.
    Recall that $a_p^C, b_p^C$ are multiplications of $\rho 2^\ell$ in the range $[0, 2^\ell/\epsilon]$, and $\Gamma_C$ is a power of $(1 + \tfrac{\epsilon}{2^{O(\ddim)} \log \Delta})$ in the range $[\median(X, S)/n, 2 \median(X, S)]$.
    Hence, there are $(\epsilon \rho)^{-2 |P_{D_1}|} \log^2 \Delta = (\epsilon \rho)^{-\rho^{-O(\ddim)}} \log^2 \Delta$ such combinations.
    For each of these combinations, the value of $g(\cdot) + h(\cdot)$ can be calculated in $O(1)$ time, thus a total of $(\epsilon \rho)^{-\rho^{-O(\ddim)}} \log^2 \Delta$ time.

    For $1 \leq i \leq |\Nchild(C)| - 1$, the computation of $h(C, D_i, \veca_C, \vecb_C, \veca_1, \vecb_1, \dots, \veca_i, \vecb_i, \Gamma)$ is analogous, thus $(\epsilon \rho)^{-\rho^{-O(\ddim)}} \log^2 \Delta$ time in total.

    The computation of $
    h(C, D_{|\Nchild(C)|}, \veca_C, \vecb_C, \veca_1, \vecb_1, \dots, \veca_{|\Nchild(C)|}, \vecb_{|\Nchild(C)|}, \Gamma)
    $
    is more involved.
    By \Cref{sec:time_complexity_fl_bounded_clients}, it can be computed in time 
    \[
    \rho^{-O(\ddim)} + 2^{\rho^{-O(\ddim)}} |\Schild(C)|.
    \]

    Therefore, the time complexity of the computation regarding a single cluster $C \in \modifydecom_\ell$ is 
    \begin{align*}
        &\qquad (\epsilon \rho)^{-2 |P_C|} \log^2 \Delta \cdot (\epsilon \rho)^{-2 |P_{D_1}|} \log^2 \Delta \\
        &\qquad \qquad + \sum_{i = 1}^{|\Nchild(C)| - 1}
        (\epsilon \rho)^{-2 |P_C| - 2\sum_{j = 1}^i |P_{D_j}|} \log^2 \Delta \cdot (\epsilon \rho)^{-2 |P_{D_{i + 1}}|} \log^2 \Delta \\
        &\qquad \qquad + (\epsilon \rho)^{-2 |P_C| - 2\sum_{j = 1}^{|\Nchild(C)|} |P_{D_j}|} \log^2 \Delta
        (\rho^{-O(\ddim)} + 2^{\rho^{-O(\ddim)}} |\Schild(C)|) \\
        & \leq (\epsilon \rho)^{-\rho^{-O(\ddim)}} \log^4 \Delta \cdot (1 + |\Schild(C)|).
    \end{align*}

    The total complexity of filling the DP table is 
    \begin{align*}
        &\qquad \sum_{\ell = 0}^L \sum_{\substack{C \in \modifydecom_\ell \\ \text{$C$ is non-ornament}}}  
        (\epsilon \rho)^{-\rho^{-O(\ddim)}} 
        \log^4 \Delta \cdot (1 + |\Schild(C)|) \\
        & = (\epsilon \rho)^{-\rho^{-O(\ddim)}} \log^4 \Delta
        \sum_{\ell = 0}^L \Big(|\{C \in \modifydecom_\ell \colon \text{$C$ is non-ornament}\}|
        + |\{y \in Y \colon h(y) = \ell\}|\Big) \\
        & = \left(\frac{\ddim}{\epsilon}\right)^{\left(\frac{\ddim}{\epsilon}\right)^{O(\ddim)}} \tilde{O}(n + m) \log^4 \Delta.
    \end{align*}

    Combining the analysis above, we conclude that the time complexity of our algorithm is $2^{2^t} \cdot \tilde{O}(n + m) \log^4 \Delta$, for 
    \begin{align*}
        t = O\left(\ddim \log\frac{\ddim}{\epsilon}\right).
    \end{align*}
    \Cref{thm:kmedian_bounded_clients} follows with $\Delta = \poly(n, m)$.
\end{proof}

\section{$(k,\ell)$-Median Clustering of Polygonal Curves under Discrete Fréchet Distance}
\label{sec:discreteFrechet}
In this section, we apply our method to obtain new results for the $(k,\ell)$-median problem of polygonal curves under the discrete Fr\'echet distance. This is a standard variant of the $k$-median problem, where the centers are constrained to have $\ell$ vertices. We first establish the general result and then specialize to the case 
$d=1$, in which we achieve an improved running time by means of a novel complexity-reduction technique that may be of independent interest.

The following statement about the doubling dimension of the discrete Fr\'echet distance is probably folklore, but we include a proof for completeness. 
\begin{proposition}
\label{prop:ddim}
The metric space defined on the equivalence classes of polygonal curves in $\bX_\comp^{d}$ that have pairwise discrete Fr\'echet distance $0$, 
equipped with the discrete Fr\'echet distance, has doubling dimension $\Theta(d\comp)$.
\end{proposition}
\begin{proof}
We first show that the doubling dimension is at least $d\comp$. For each $i \in [d]$, let $e_i$ be the $i$-th vector of the standard basis in $\RR^d$, i.e., $e_i$ has $1$ appearing in the $i$-th position and $0$ in all other positions.  
Consider a polygonal curve $\pi = \langle x_1,\dots,x_\comp \rangle $ where each $x_i = (3\cdot i,0\ldots, 0) $ for each $i\in [\comp]$. Let $B$ be the discrete Fr\'echet ball of radius $1$ centered at $\pi$, i.e., $B = \{\tau \in \bX_{\comp}^d ~|~ \dist_{dF}(\pi,\tau) \leq 1\}$. All polygonal curves defined by sequences in $\prod_{i=1}^\comp \bigcup_{j=1}^{d} \{x_i-e_j,x_i+e_j\}$ are in $B$, but no two of them can be covered by the same ball of radius $1/2$. Hence, we need at least $2^{d\comp}$ balls of radius $1/2$ to cover $B$, implying that the doubling dimension is at least $d\comp$. 

For the upper bound, consider any polygonal curve $\pi=\langle x_1,\ldots,x_\comp \rangle $ and any radius $r>0$. Let $B_r$ be the discrete Fr\'echet ball of radius $r$ centered at $\pi$. For each $i\in [\comp]$, let $S_i$ be a minimal $r/2$-covering set of the Euclidean ball of radius $r$ centered at $x_i$, i.e., for any $p \in \RR^d$ satisfying $\|p-x_i\|_2\leq r$, there exists a $p' \in S_i$ such that $\|p-p'\|_2\leq r/2$. It is known (see e.g.~\cite{GuptaKL03}) that $|S_i| = 2^{O(d)}$. 
Let  $T^{\ast}$ be an optimal traversal between $\pi$ and an arbitrary polygonal curve $\tau= \langle q_1,\ldots, q_\comp\rangle \in B_r$.  For each $j\in [\comp]$, let $i_{j}$ be the first index (pointing to the $i_{j}$-th vertex of $\pi$) paired with the $j$-th vertex of $\tau$ in $T^{\ast}$. It must hold that there is a $y_j \in S_{i_j}$ such that $\|q_j-y_j\|_2 \leq r/2 $. 
Hence, the sequence $\langle y_1,\ldots, y_\comp \rangle$ defines a polygonal curve with discrete Fr\'echet distance within $r/2$ from $\tau$. By taking into account all traversals and for each traversal all relevant combinations of points of $S_i$, we cover the entire $B_r$ with balls of radius $r/2$. By \cite[Lemma 4]{FFK23}, 
the number of traversals is at most $4^\comp$, and for each traversal, we consider at most $2^{O(d\comp)}$ polygonal curves as centers of the balls of radius $r/2$, implying an upper bound of $O(d\comp)$ on the doubling dimension. 
\end{proof}

We begin with an auxiliary lemma on discretizing the set of polygonal curves that have complexity $\ell$ and are within distance $r$ from an arbitrary polygonal curve of complexity $\comp$. This was essentially proven and used by Filtser et al.~\cite{FFK23} in the context of approximate nearest neighbor data structures. Since there is no standalone lemma with the exact same statement in~\cite{FFK23}, we include a proof for completeness. 

\begin{lemma}
\label{lemma:candidate_single_scale}
Given a polygonal curve $\pi\in \bX_\comp^d$, threshold $r>0$ and $\eps \in (0,1)$, we can compute a set of polygonal curves $\cS_{\pi,r,\eps}$, in time $O(1/\eps)^{d\ell} + O(d\comp\log \comp)$, such that $|\cS_{\pi,r,\eps}| \in O(1/\eps)^{d\ell}$ and for all $\tau\in \bX_{\ell}^d$, if $\dist_{dF}(\pi,\tau)\leq r$ then there exists $\tau'\in \cS_{{\pi},r,\eps}$ such that $\dist_{dF}(\tau',\tau)\leq \eps r$. 
\end{lemma}

\begin{proof}
We first compute a polygonal curve $\pi'\in \bX_{\ell}^d$ such that $\dist_{dF}(\pi,\pi')\leq 2r$ or determine that there is no curve in $\bX_{\ell}^d$ within distance $r$ from $\pi$, using \cite[Lemma 10]{FFK23}, in time $O(d\comp\log \comp)$. If there is no curve in $\bX_{\ell}^d$ within distance $r$ from $\pi$ then we return $\cS_{\pi,r,\eps}\gets \emptyset$. Otherwise, we proceed as follows.  
    Let $\tau = \langle q_1,\ldots, q_{\ell}\rangle$, $\pi'=\langle\tilde{x}_1,\ldots, \tilde{x}_{\ell}\rangle$ and $\mathcal{G}_{\eps r}$ be the regular grid in $\RR^d$ with cell width $\eps r d^{-1/2}$. We define  $\tilde{I}_j := \{p\in \RR^d\mid \|p-\tilde{x}_j\|_2 \leq (3+\eps)r\}\cap \mathcal{G}_{\eps r}$, for any $j \in [\ell]$.
    For each traversal $T$ of two polygonal curves of complexity $\ell$
    we compute a set of polygonal curves $\cS_T$ as follows: for each $j\in [\ell]$ let $i_{j}$ be the index of the first vertex matched to the vertex at index $j$ in the traversal, i.e., $i_{j}= \min\{i:~ (i,j)\in T \}$. We compute $\cS_T = \prod_{j=1}^{\ell}\tilde{I}_{i_j}$ and we output $\cS_{\pi,r,\eps} = \bigcup_{T} \cS_T$. 

    Each $\tilde{I}_j$ contains at most $O(1/\eps)^d$ points. Hence, for a fixed traversal $T$, there are $O(1/\eps)^{d\ell}$  combinations of vertices defining the polygonal curves included in $\cS_T$. By \cite[Lemma 4]{FFK23}, the number of different traversals is at most $4^{\ell}$. Hence, the running time of computing the set $\cS_{\pi,r,\varepsilon}$ is $O(1/\eps)^{d\ell} + O(d\comp\log \comp)$ and its size is upper bounded by $O(1/\eps)^{d\ell}$.
    
    To show correctness, we first observe that if $\dist_{dF}(\pi,\tau)\leq r$, then by the triangle inequality $\dist_{dF}(\pi',\tau) \leq \dist_{dF}(\pi,\tau)+\dist_{dF}(\pi,\pi')  \leq 3r$. Hence, there is an optimal traversal $T^{\ast}$ matching the vertices of $\tau$ with vertices of $\pi'$ with cost at most $3r$, which implies that there exists a polygonal curve in $\cS_{T^{\ast}}$ with vertices $q_1',\ldots,q_{\ell}'$ such that for any $j\in [\ell]$, $\|q_j'-q_j\|_2 \leq \eps r$. Therefore, $\dist_{dF}(\tau,\tau')\leq \eps r$. 

\end{proof}

We require the following definition of simplifications.
\begin{definition}[minimum-error $\ell$-simplification]
For a curve $x \in \bX^d_z$ a curve $\tilde{x} \in \bX^d_\ell$ is a minimum-error $\ell$-simplification of $x$ if for any curve $y \in \bX^d_\ell$ it holds that $\dist_{dF}(x,\tilde{x}) \leq \dist_{dF}(x,y)$.
\end{definition}
 The minimum-error $\ell$-simplification $\tilde x$ of some curve $x \in \bX^d_z$ can be computed in time $O(2^{O(d)}\ell \comp \log (\comp) \log (\comp/\ell) )$
\cite{bereg2008simplifying}.

Now we present a simple constant factor approximation algorithm, which was previously used and analyzed in \cite{DKS16} in a slightly different setting. We give a proof for the sake of completeness.

\begin{algorithm}
    \caption{$(k,\ell)$-MedianConstantApprox}
    \label{alg: kl-median-constant-approx}
    
    \KwIn{$\{x_1,\dots,x_n\} \in \bX^d_z$, $\ell \in \NN$, $k \in \NN$}

    Let $\tilde{P} \gets \{\tilde{x}_i\mid i\in[n]\}$, where $\tilde{x}_i$ is a
    minimum-error $\ell$-simplification of $x_i$

    Let $C$ be the solution computed by a constant factor approximation algorithm for $k$-median on $(\tilde{P},\dist_{dF})$

    \Return $\sum_{i = 1}^n \dist_{dF}(\tilde{x}_i,C) + \dist_{dF}(\tilde{x}_i,x_i)$
   
\end{algorithm}

\begin{lemma}\label{lem: kl-median_const_approx}
    Given $\{x_1,\dots,x_n\} \in \bX^d_z$, $\ell \in \NN$, $k \in \NN$. Algorithm \ref{alg: kl-median-constant-approx} computes in time $O(2^{O(d\ell)} d\ell ^2nz(\log(z) \log(z/\ell)+\log^9n))$ a constant factor approximation for the $(k,\ell)$-median problem.
\end{lemma}
\begin{proof}
We start by showing that the computed value is indeed a constant approximation to the optimal cost.
Let $C^*$ be an optimal solution for the $(k,\ell)$-median problem with cost $\delta^*$ and let $c^*_i \in C^*$  be the closest center to $x_i$.
Additionally let $\lambda \in \RR$ denote the approximation factor of $C$.
By triangle inequality we get the following lower bound.
\begin{align*}
    \sum_{i = 1}^n \dist_{dF}(\tilde{x}_i,C) + \dist_{dF}(\tilde{x}_i,x_i) &\geq \sum_{i = 1}^n \dist_{dF}(x_i,C)\\
    &\geq \delta^*
\end{align*}
Next we can upper bound the returned value as follows, by utilizing the definition of a minimum-error $\ell$-simplification
\[\sum_{i = 1}^n \dist_{dF}(\tilde{x_i},C) + \dist_{dF}(\tilde{x_i},x_i)  \leq  \sum_{i = 1}^n \dist_{dF}(\tilde{x_i},C) + \delta^*.\]

Then, again by using triangle inequality and the definition of minimum-error $\ell$-simplification we get
\begin{align*}
\delta^* &= \sum_{i=1}^n \dist_{dF}(x_i,c^*_i)\\
&\geq \sum_{i=1}^n \dist_{dF}(\tilde{x}_i,c^*_i) - \dist_{dF}(\tilde{x}_i,x_i)\\
&\geq \sum_{i=1}^n \dist_{dF}(\tilde{x}_i,c^*_i) - \delta^*\\
&\geq \frac{1}{\lambda} \sum_{i=1}^n \dist_{dF}(\tilde{x}_i,C) - \delta^*,
\end{align*}
which yields $\sum_{i = 1}^n \dist_{dF}(\tilde{x_i},C) + \dist_{dF}(\tilde{x_i},x_i)  \leq (2\lambda+1)\delta^*$.

Next we analyze the running time of Algorithm \ref{alg: kl-median-constant-approx}. Calculating the set $\tilde{P}$ takes $O(n2^{O(d)}\ell \comp \log (\comp) \log (\comp/\ell) )$ time \cite{bereg2008simplifying}. Computing $C$ takes $O(d\ell^22^{O(d\ell)}n \log 9n)$ time by using  the constant approximation algorithm in \cite{Cohen-AddadFS21}, Proposition \ref{prop:ddim} and the fact that computing the discrete Fréchet distance between to curves in $\bX^d_\ell$ takes $O(d\ell^2)$ time using a straightforward dynamic programming algorithm \cite{agarwal2014computing}.
The total running time is therefore $O(2^{O(d\ell)} d\ell ^2nz(\log(z) \log(z/\ell)+\log^9n))$
\end{proof}

We now proceed by computing a set of candidate centers, i.e., a set of polygonal curves with complexity $\ell$, from which the $k$ centers can be chosen, while only sacrificing an arbitrarily small approximation factor from the cost of the optimal solution. 

\begin{lemma}
\label{lemma:mediancandidates}
    Given a set of $n$ polygonal curves $\Pi \subset \bX_\comp^d$ and $k,\ell \in \NN$, and $\eps \in (0,1)$, we can compute in time $O(2^{O(d\ell)} d\ell ^2nz(\log(z) \log(z/\ell)+\log^9n))  + O(n\log n)\cdot O(1/\eps)^{d\ell}$ a set $\cS \subset \bX_\ell^d$ of size $O( n \log n) \cdot O(1/\eps)^{d\ell}$ such that there exists a set $\cC \subset \cS$, $|\cC|=k$, that satisfies 
    \[
    \sum_{\pi\in \Pi} \min_{\tau\in \cC}\dist_{dF}(\pi,\tau) \leq (1+3\eps)\optkl,
    \]
    where $\optkl$ is the optimal $(k,\ell)$-median cost.  
\end{lemma}

\begin{proof}
    We first compute $\delta$ such that $\optkl \leq \delta \leq \alpha \cdot \optkl$ using Algorithm \ref{alg: kl-median-constant-approx}, where $\alpha$ is some constant. For each $ i \in \{0,\dots , \lceil  \log( \alpha  n)\rceil\}$, let $r_i = \frac{2^i\cdot \delta}{\alpha n}$. For each $\pi\in \Pi$, we compute $\cS_{\pi,r_i,\eps}$ using \Cref{lemma:candidate_single_scale}. We output $\cS=\bigcup_{\pi\in \Pi}\bigcup_{i=1}^{\lceil  \log( \alpha n)\rceil} \cS_{\pi,r_i,\eps}$.
The time needed to compute $\delta$ is in $O(2^{O(d\ell)} d\ell ^2nz(\log(z) \log(z/\ell)+\log^9n))$ \ref{lem: kl-median_const_approx}. 
    By \Cref{lemma:candidate_single_scale}, the running time to create $\cS$  is in $O(n\log n) \cdot O\left((1/\eps)^{d\ell} +d\comp\log \comp \right)$ since we consider $O(\log n)$ different parameters $r_i$, and its size is in $O(n\log n) \cdot O\left((1/\eps)^{d\ell} \right)$. 

    To show correctness, consider any polygonal curve $\tau\in \cC^{\ast}$, where $\cC^{\ast}$ is an optimal solution for $(k,\ell)$-median. For any $\pi\in \cP$ that has $\tau$ as its closest center from $\cC^*$ it holds that $\dist_{dF}(\pi,\tau)\leq \optkl$. 
    Now let $\pi_{\tau}$ be the polygonal curve in $\Pi$ which has the  smallest discrete Fr\'echet distance to $\tau$  (ties are broken arbitrarily)  and let $i^{\ast}$ be the smallest value $i$ such that 
    $\dist_{dF}(\pi_{\tau},\tau)\leq r_i$. By \Cref{lemma:candidate_single_scale}, $\cS_{\pi_{\tau},r_{i^\ast},\eps}$ contains a polygonal curve $\tau'\in \bX_{\ell}^d$ such that $\dist_{dF}(\tau,\tau')\leq \eps r_{i^{\ast}}$. 
    If $i^{\ast}=0$, then $\dist_{dF}(
    \tau,\tau') \leq \frac{\eps \optkl}{n}$ and by the triangle inequality, for any $\pi\in \Pi$, $\dist_{dF}(\pi,\tau')\leq \dist_{dF}(\pi,\tau)+ \frac{\eps \optkl}{n}$. 
    If $i^{\ast}>0$, then by the triangle inequality, for any $\pi\in \Pi$ that has $\tau$ as its closest center from $\cC^*$,
    \begin{align*}
        \dist_{dF}(\pi,\tau')&\leq \dist_{dF}(\pi,\tau)+\dist_{dF}(\tau,\tau')\\
        &\leq  \dist_{dF}(\pi,\tau)+ \eps \cdot r_{i^{\ast}}\\
        &\leq  \dist_{dF}(\pi,\tau)+ 2\eps \cdot \dist_{dF}(\pi_{\tau},\tau)\\
        & \leq (1+2\eps)\cdot  \dist_{dF}(\pi,\tau).
    \end{align*}
 Hence, for any $\tau\in \cC^{\ast}$ there is a $\tau' \in \cS$ such that for any $\pi\in \Pi$ that has $\tau$ as its closest center, $\dist_{dF}(\pi,\tau') \leq \max ( (1+2\eps)\cdot  \dist_{dF}(\pi,\tau) ,  \dist_{dF}(\pi,\tau)+ \frac{\eps \optkl}{n})$. 
 Now, let $\cC$ be a set containing one such $\tau'$ for each $\tau \in \cC^{\ast}$. The cost of this solution is
 \begin{align*}
     \sum_{\pi\in \Pi} \min_{\tau'\in \cC}\dist_{dF}(\pi,\tau') 
     &\leq 
     \sum_{\pi\in \Pi} \min_{\tau\in \cC^{\ast}} \max \left (  (1+2\eps)\cdot  \dist_{dF}(\pi,\tau) ,  \dist_{dF}(\pi,\tau)+ \frac{\eps \optkl}{n} \right )
     \\
     &\leq  \sum_{\pi\in \Pi} \min_{\tau\in \cC^{\ast}}   \left((1+2\eps)\cdot  \dist_{dF}(\pi,\tau) + \frac{\eps \optkl}{n}\right)
     \\
     &\leq  (1+2\eps)\cdot  \sum_{\pi\in \Pi} \min_{\tau\in \cC^{\ast}}    \dist_{dF}(\pi,\tau) + n\cdot \frac{\eps \optkl}{n}
     \\
     &\leq (1+3\eps)\optkl. 
 \end{align*}
\end{proof}
We can now prove our main result for the $(k,\ell)$-median problem. Note that the discrete Fréchet distance between two curves with ambient dimension $d$ of complexities $z$ and $\ell$ can be computed in $O(dz\ell)$ time \cite{agarwal2014computing}.

\thmklmedian *

\begin{proof}
    The result follows directly by using the set $\cS$ computed as in \Cref{lemma:mediancandidates}, as a set of center candidates in the algorithm of \Cref{thm:kmedian_bounded_centers}.
\end{proof}

In the next subsections, we focus on the case $d=1$ and show an alternative approach for solving the problem. For the case $d=1$ we develop a novel complexity-reduction technique that may be of independent interest. We give an additional  application of this technique in the context of coresets for $(k,\ell)$-median. 

\subsection{Complexity Reduction for Time Series}

In this section, we will establish a complexity-reduction technique for time series.  Time series of complexity $z\in \mathbb{N}$ are polygonal curves of $\bX_z^{1}$ and can be simply seen as $z$-dimensional real vectors, i.e., members of $\RR^z$. 
In a first step, we will show that one can quantize the entries of a time series from $\mathbb R^\comp$
to $O(\ell/\varepsilon)$ distinct values, while guaranteeing that the distance to any time series of complexity $\ell$ is preserved up to a factor of $(1\pm \varepsilon)$.

\begin{algorithm}
\caption{ReduceValueDomain}
\label{alg: ReduceValueDomain}

\KwIn{$x \in \RR^\comp$, $\ell \in \NN$, $\varepsilon \in \RR_{> 0}$}

 Let $\tilde{x}$ be a minimum-error $\ell$-simplification of $x$

 $\delta \gets \dist_{dF}(x,\tilde{x})$

 Obtain $x'$
from $x$ by rounding up each entry of $x$ to its next multiple of $\varepsilon\delta$ 

\Return $x'$

\end{algorithm}

\begin{lemma}\label{lem: reduce value domain}
Let $x \in \RR^\comp$, $\ell \in \NN$ and $1\ge\varepsilon>0$. Algorithm \ref{alg: ReduceValueDomain} computes in time $O(\ell \comp  \log(\comp)  \log(\comp/\ell))$ a time series $x' \in X^\comp$ with $X \subset \RR$ and $|X| \in O(\ell /\varepsilon)$ s.t. for every $y \in \RR^\ell$, \[ (1-\varepsilon) \dist_{dF}(x,y) \le \dist_{dF}(x',y) \le (1+ \varepsilon)  \dist_{dF}(x,y).\]
\end{lemma}

\begin{proof}
   The minimum-error $\ell$-simplification $\tilde x$ of $x$ can be computed in $O(\ell \comp \log (\comp) \log (\comp/\ell))$ time using an algorithm from 
\cite{bereg2008simplifying}
and the distance $\dist_{dF}(x,\tilde x)$ can be computed in $O(\ell \comp)$ time \cite{agarwal2014computing}.
 The remaining steps are linear in $\comp$ and hence the running time follows.
Now let $\tilde{x}$ be the computed minimum-error $\ell$-simplification of $x$ and let $\delta = \dist_{dF}(x,\tilde{x})$. 
    For $1 \leq j \leq \ell$ we define $S_j = [ \tilde{x}_j - \delta, \tilde{x}_j + (1+\varepsilon)\delta ]$. Then it holds for all $1 \leq i \leq \comp$ that there is $1 \leq j \leq \ell$ s.t. $x_i \in S_j$. It follows that the number of distinct values of $x'$
    is $O(\ell/\varepsilon)$. 
By the triangle inequality we have 
\[|\dist_{dF}(x,y) - \dist_{dF}(x',y)|\leq \dist_{dF}(x,x')\leq \Vert x-x'\Vert_\infty \leq \varepsilon \delta \leq \varepsilon\dist_{dF}(x,y)\] for every $y \in \RR^\ell$.
\end{proof}

In our complexity reduction we are interested in maintaining 
the discrete Fr\'echet distance of a time series $x\in \RR^\comp$
to every time series of length at most $\ell$. Thus, we can use the previous lemma to reduce the number of distinct values of every fixed time series to $O(\ell/\varepsilon)$. We will 
therefore focus in the remainder of this section on such time series and exploit this property for our dimension reduction.

For $x\in \RR^\comp$, $x= (x_1,\dots, x_\comp)$, define
$\vectorset(x) = \{x_i : 1\le i \le \comp\}$.

\begin{definition}{(traversal sectors)}
Let $x \in \RR^\comp$, $y \in \RR^\ell$ and $T$ be a traversal between $x$ and $y$.
For $j\in [\ell]$ 
we define $S^{(x,T)}_j := \{x_i ~|~ i \in [\comp] \text{ and } (i,j) \in T\}.$
We call the sequence $(S^{(x,T)}_1,\dots, S^{(x,T)}_\ell)$ the traversal sectors of $x$ and $T$.
\end{definition}

Furthermore,  $(S_1,\dots, S_\ell)$ are called  traversal sectors of $x$, if there exists a traversal $T$ with traversal sectors $(S_1,\dots, S_\ell)$. 
We observe that for a given 
traversal $T$ we get
\[\max_{(i,j)\in T} |x_i - y_j| = \max_{1 \leq i \leq \ell} \max_{a \in  S^{(x,T)}_i} |a - y_i| = \max_{1 \leq i \leq \ell} \max\{|\min(S^{(x,T)}_i) - y_i|, | \max (S^{(x,T)}_i) - y_i|\}.\]

Thus, to determine the discrete Fr\'echet distance between time series $x$ and $y$ it suffices to consider the minimum and maximum value in each traversal sector. This will be used in the following definition.

\begin{definition}[$\ell$-profile]
Let $\ell \in \NN$ and 
 $x\in \RR^\comp$.
For an arbitrary $y\in \RR^\ell$ and
a traversal $T$ between $x$ and $y$ we call the sequence 
\[ 
\big(\min (S^{(x,T)}_1),\max (S^{(x,T)}_1)\big), \dots, \big(\min (S^{(x,T)}_\ell), \max (S^{(x,T)}_\ell)\big)
\] the $\ell$-profile of $(x,T)$. 
\end{definition}

See Figure \ref{fig: example profile}
for an example of traversal sectors and $\ell$-profile.

\begin{figure}[h]
    \centering
    \includegraphics[width=0.48\linewidth]{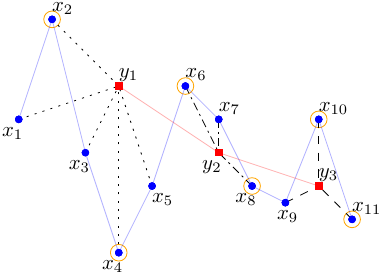}
    \caption{Depicted are two time series $x$ and $y$ and the matched vertices through some traversal $T$. The traversal sectors of $(x,T)$ are $S^{(x,T)}_1 = \{x_1,x_2,x_3,x_4,x_5\}, S^{(x,T)}_2 =\{x_6,x_7,x_8\}$ and $S^{(x,T)}_3 = \{x_9,x_{10},x_{11}\}$. The orange highlighted vertices are the extrema values located in each traversal sector, which define the profile, i.e. the sequence $((x_2,x_4),(x_6,x_8),(x_{10},x_{11}))$ is the $3$-profile of $(x,T)$. }
    \label{fig: example profile}
\end{figure}

Let $D_{(x,\ell)}$ denote the set of all $\ell$-profiles of a time series $x \in \RR^\comp$
over all traversals in $\mathcal{T}_{\comp,\ell}$. 

Since the information stored in an $\ell$-profile of some time series $x$ and a traversal $T$ preserves $\max_{(i,j)\in T} |x_i - y_j|$ for any $y \in \RR^\ell$ we can argue that two time series $x$ and $x'$ with $\vectorset(x) = \vectorset(x')$
that have the same set of $\ell$-profiles are interchangeable  w.r.t. the discrete Fréchet distance. This is stated in the following lemma.

\begin{restatable}{lemma}{lemmatchingprofiles}\label{lem: matching profiles}
Let $\comp,\comp',\ell \in \NN$.
Let $x\in \RR^{\comp}$ and $x'\in \RR^{\comp'}$ be two time series with the same set of $\ell$-profiles, i.e. 
$D_{(x,\ell)} = D_{(x',\ell)}$.
 Then for every $y \in \RR^\ell$ we have $\dist_{dF}(x,y) = \dist_{dF}(x',y).$
\end{restatable}

\begin{proof}       
    Since $D_{(x,\ell)} = D_{(x',\ell)}$, for any traversal $T \in \mathcal{T}_{(\comp,\ell)}$ there is a traversal $T'= \in \mathcal{T}_{(\comp',\ell)}$  s.t. for all $i \in [\ell]$,  $\max(S^{(x,T)}_i) = \max(S^{(x',T')}_i)$ and $\min(S_i^{(x,T)}) = \min(S^{(x',T')}_i)$ and vice versa.

    Consider some arbitrary $y \in \RR^\ell$ and let $T\in \mathcal{T}_{\comp,\ell)}$ be s.t. $\max_{(i,j)\in T} |x_i - y_j| = \dist_{dF}(x,y)$.
Let $T'\in \mathcal{T}_{(\comp',\ell)}$ s.t. $(x',T')$ has the same $\ell$-profile as $(x,T)$. 

    Then we get
    \begin{align*}
    \dist_{dF}(x,y) &= \min_{(i,j) \in T} |x_i - y_j| \\
    &= \max_{i \in [\ell]} \max \{|\min(S^{(x,T)}_i) - y_i|, | \max(S^{(x,T)}_i) - y_i|\}\\
    &= \max_{i \in [\ell]} \max \{|\min(S^{(x',T')}_i) - y_i|, | \max(S^{(x',T')}_i) - y_i|\}\\
    &= \max_{(i,j) \in T'} |x'_i - y_j| \\
    &\ge \dist_{dF}(x',y).
    \end{align*}

 The other direction holds by a symmetric argument, which concludes the proof.
\end{proof}

In the following we show how to construct for any time series $x \in \RR^\comp$ and $\ell \in \NN$  a time series $x' \in \RR^{\comp'}$ s.t. $D_{(x,\ell)} = D_{(x',\ell)}$ and $\comp' \in O(2^{|\vectorset(x)|^{2\ell+2} })$. To do so we  introduce the concept of \emph{prefix profiles}. For some time series $x \in \RR^\comp$ and $t \in [\comp]$ we write $x(t):= (x_1,\dots,x_t)^T$ to denote the prefix of $x$ of complexity $t$. 

\begin{definition}[prefix profile]
Let $\ell \in \NN$,
 $x\in \RR^\comp$ and $t \in [\comp]$. For arbitrary $y\in \RR^\ell$ and some
traversal $T$ between $x$ and $y$ let  $r \in [\ell]$ be the largest value s.t. $S_r^{(x(t),T)} \neq \emptyset$. Then we call the sequence 
\[ 
\big(\min (S^{(x(t),T)}_1),\max (S^{(x(t),T)}_1)\big), \dots, \big(\min (S^{(x(t),T)}_r), \max (S^{(x(t),T)}_r)\big)
\] the \emph{prefix profile} of $(x(t),T)$ with length $r$. 
\end{definition}
Observe that for a time series $x \in \RR^\comp$ there are $O(|\vectorset(x)|^{2r})$ many possible prefix profiles with length $r$. This implies that the number of distinct prefix profiles over all lengths $r\in [\ell ]$ is 
\[ \sum^\ell_{r = 1} O(|\vectorset(x)|^{2r}) \in O(|\vectorset(x)|^{2\ell +2}).\] 
Now let $D_{x(t)}$ be the set of all prefix profiles for $(x(t),T)$ over all traversals $ T \in \mathcal{T}_{\comp,\ell}$ and let $D_x = \bigcup_{t=1}^\comp D_{x(t)} $. Then there is a unique function $\delta_x:2^{D_x} \times \vectorset(x) \rightarrow 2^{D_x}$ that defines how a set of prefix profiles can evolve to a potentially different set of prefix profiles by including a further vertex of $x$. To see this, consider some arbitrary prefix profile of length $r$ for $x(t)$ and a traversal $T$. Then either $(t+1,r)\in T$ and $x_{t+1} \in S_r^{(x(t+1),T)}$ or $(t+1,r+1) \in T$ and $S^{(x(t+1),T)}_{r+1} = \{x_{t+1}\}$. 
Both possibilities correspond to different traversals in $\mathcal{T}_{m,\ell}$ and therefore potentially different prefix profiles in $D_x$.

We are now ready to show the following lemma.

\begin{lemma}\label{lem: existence matching profiles} For every $x \in \RR^\comp$ and $\ell \in \NN$  there exists a time series $x' \in \vectorset(x)^{\comp'}$ with $\comp' \in O(2^{|\vectorset(x)|^{(2\ell + 2)}})$ s.t. 
$D_{(x,\ell)} = D_{(x',\ell)}$.
\end{lemma}

\begin{proof}
     Let $x' \in \vectorset(x)^{\comp'}$ be the time series with smallest complexity $\comp'$ s.t. $D_{(x,\ell)} = D_{(x',\ell)}$ and observe that the number of subsets of the prefix profiles of $x'$ is bounded by $|2^{D_{x'}}|\in O(2^{|\vectorset(x)|^{(2\ell +2)}}$). Assume there exist $1\leq t_1 < t_2 \leq \comp'$ with $D_{x'(t_1)} =D_{x'(t_2)}$, then it holds that $\delta_{x'}(D_{x'(t_1)},x'_{t_2+1}) = \delta_{x'}(D_{x'(t_2)},x'_{t_2+1})$. Recursively applying $\delta_{x'}$ with the remaining vertices of $x'$ implies that the time series $x'' = (x'_1,\dots x'_{t_1},x'_{t_2+1},\dots x'_{\comp'})^T$ has the same set of $\ell$-profiles
     as $x'$ and therefore $x$, i.e. $D_{(x'',\ell)} = D_{(x',\ell)} = D_{(x,\ell)}$. This is a contradiction to the assumption that $x'$ is the time series of smallest complexity with $D_{(x',\ell)} = D_{(x,\ell)}$. Therefore for all $1\leq t_1 < t_2\leq \comp'$ it holds that $D_{x'(t_1)} \neq D_{x'(t_2)}$ and since the number of distinct prefix profiles of $x'$ is $O(2^{|\vectorset(x)|^{2\ell + \ell}})$ we get that $\comp' \in O(2^{|\vectorset(x)|^{2\ell + \ell}})$.
    
\end{proof}

The following corollary results from combining Lemma \ref{lem: reduce value domain} and Lemma \ref{lem: existence matching profiles}.

\begin{corollary} 
For every $\varepsilon \in (0,1)$, $\ell \in \NN, \comp \in \NN$ and every time series $x\in \RR^\comp$ there exists a 
time series $x'
\in \RR^{\comp'}$ with $\comp' \in O(2^{O(\ell/\eps)^{(2\ell+2)}})$ such that for every $y\in \RR^\ell$ we have
\begin{equation*}
(1-\varepsilon) \dist_{dF}(x,y) \le \dist_{dF}(x',y) \le (1+\varepsilon) \dist_{dF}(x,y).
\end{equation*}
\end{corollary}

\subsubsection{Algorithm}

In this section we turn the previously discussed existential result into an algorithmic one.

Given some time series $x \in \RR^{\comp}$ the solution will be to iterate over all  increasing complexities $t$ and all time series
in $\vectorset(x)^t$ until we find one that has the same set of  $\ell$-profiles as $x$.
Our previous results guarantee that this algorithm terminates with 
 $t \in O(2^{|\vectorset(x)|^{(2\ell + 2)}})$. By Lemma \ref{lem: existence matching profiles} we have that $\dist_{dF}(x,q) = \dist_{dF}(y,q)$, for all $q \in \RR^{\ell}$. 
The challenge is to efficiently compute the set of $\ell$-profiles for a given time series, which we will discuss in the remaining part of the section.

To compute the set $D_{(x,\ell)}$, for some $x \in \RR^\comp$, we iterate over all sequences $p \in ({\vectorset(x)^2})^\ell$ of potential $\ell$-profiles and decide if there exists a traversal $T$ s.t. $p$ is an $\ell$-profile of $(x,T)$. To do so it is sufficient to decide the existence of the corresponding traversal sectors $S^{(x,T)}_{1},\dots, S^{(x,T)}_{\ell}$. 

 Consider a sequence $p = ((\min_1,\max_1),\dots,(\min_\ell,\max_\ell)) \in ({\vectorset(x)^2})^\ell$. Note that $\min_i$ and $\max_i$ are supposed to be the minimum and maximum element according to the considered $\ell$-profile. The objective is to decide if there exist traversal sectors
$S_1,\dots,S_\ell$ of $x$ that are consistent with $p$, i.e. for all $i \in [\ell]$, $\min_i,\max_i \in S_i$ and $S_i \subseteq [\min_i,\max_i]$. This is done in a recursive way and stated as dynamic program in the form of Algorithm \ref{alg: Partition of Profile}.

The idea is as follows. For some $t \in [\comp]$ and $h \in [\ell]$ we would like to know if there are  traversal sectors $S_1,\dots, S_h$ for $(x_1,\dots,x_t)$ that are compatible with $p$.
However, in order to set up a recursion, we also need to know whether the minimum and/or maximum value of $S_h$
has already appeared in $x_1,\dots, x_t$.
For this purpose, we introduce two Boolean variables $a$ and $b$.

Concretely, Algorithm~\ref{alg: Partition of Profile} is a dynamic programm that checks for the existence of partial solutions of the following form. 

\begin{definition}[compatible traversal sectors]
 Given a time series $x=(x_1,\dots,x_\comp) \in \RR^\comp$, sequence  $p=\big((\min_1,\max_1),\dots,(\min_\ell,\max_\ell)\big) \in (\vectorset(x)^2)^\ell$, $h \in [\ell]$, $t \in [\comp]$.
 Traversal sectors $S_1,\dots,S_h$ of $(x_1,\dots,x_t)$ are compatible with $p$ if they satisfy: 
  \begin{enumerate}
    \item $\min_i,\max_i \in S_i$ for $1\leq i\leq h-1$
    \item $S_i \subseteq [\min_i,\max_i]$ for $1\leq i \leq h$
    \end{enumerate}
\end{definition}

Note that the last sector $S_h$ is not required to fulfill the conditions of the profile, in that we don't require $\min_h$ and $\max_h$ to be contained in $S_h$. This is needed as they may not have appeared in the sequence at that time of the algorithm, yet.
Lemma~\ref{lem:partitionDP_invariant} below shows correctness of Algorithm~\ref{alg: Partition of Profile}. The crucial observation is that partial solutions can be combined as follows. 

\begin{observation}\label{obs:combination_ab}
For any $h \in [\ell], t \in [\comp]$, let $S_1,\dots,S_h$ of $(x_1,\dots,x_t)$ be compatible traversal sectors of  $(x_1,\dots,x_t)$ with $\min_h \in S_h$ and let $S'_1,\dots,S'_h$ be compatible traversal sectors of  $(x_1,\dots,x_t)$ with $\max_h \in S'_h$, then either $S_h \subseteq S'_h$ and therefore also $\min_h \in S'_h$, or $S'_h \subseteq S_h$ and therefore also $\max_h \in S_h$. This follows because the last sector in each case consist of the elements of a prefix of the same sequence.
\end{observation}

\begin{algorithm}
\caption{DecideProfile}
\label{alg: Partition of Profile}
\KwIn {$(x_1,\dots,x_\comp)$, $((\min_1,\max_1),\dots,(\min_\ell,\max_\ell))$}

 $\feas \gets \{\text{False}\}^{(\ell+1)\times (\comp+1)}$, $\hmin \gets \{\text{False}\}^{(\ell+1)\times (\comp+1)}$, $\hmax \gets \{\text{False}\}^{(\ell+1)\times (\comp+1)}$
 
$\feas[0,0] \gets \text{True}$, $\hmin[0,0] \gets \text{True}$, $\hmax[0,0] \gets \text{True}$ \label{line: decideProfile_initTables}

\For{$h = 1 \textbf{ to } \ell $}{
    \For{$t = 1 \textbf{ to } \comp $}{
       \If{$x_t \in [\min_h,\max_h]$}{\label{line: decideProfile_if_inRange}
            \If{$\feas[h,t-1]$}{\label{line: decideProfile_if_newSector}
                $\feas[h,t] \gets \text{True}$
                  
                \If{$\hmin[h,t-1] \lor x_t = \min_h$}{
                    $\hmin[h,t] \gets \text{True}$
                }
                \If{$\hmax[h,t-1] \lor x_t = \max_h$}{ 
                    $\hmax[h,t] \gets \text{True}$
                }
            }
            \Else{ 
                \If{$\feas[h-1,t-1] \land \hmin[h-1,t-1] \land \hmax[h-1,t-1]$}{\label{line: decideProfile_if_oldSector}
                    $\feas[h,t] \gets \text{True}$
        
                    \If{$x_t = \min_h$}{ 
                        $\hmin[h,t] \gets \text{True}$
                    }
                    \If{$x_t = \max_h$}{
                        $\hmax[h,t] \gets \text{True}$
                    }
                }
            }
        }
    }
}                                   
\Return $\feas[\ell,\comp] \land \hmin[\ell,\comp] \land \hmax[\ell,\comp]$

\end{algorithm}

\begin{restatable}{lemma}{lemassignmentdp}\label{lem:partitionDP_invariant}
Given $x \in \RR^\comp$ and $p \in {(\vectorset(x)^2)}^\ell$.
Let $\feas, \hmin$ and $\hmax$ be the tables constructed during Algorithm \ref{alg: Partition of Profile} with input $(x,p)$.
For every $h \in [\ell], t \in [\comp]$, after the corresponding iteration of the for-loop, we have $\feas[h,t] = True$
if and only if there exist 
compatible traversal sectors $S_1,\dots,S_h$ of $(x_1,\dots,x_t)$. In addition, we have that $\hmin[h,t] = True$ (resp. $\hmax[h,t] = True$) if and only if there exists such a compatible solution with $\min_h \in S_h$ (resp. $\max_h \in S_h$).
\end{restatable}

\begin{proof}
We prove the lemma by induction on the iterations of the inner for-loop. For the base case consider the first iteration with $h=1$ and $t=1$. Assume $x_1 \in [\min_1,\max_1]$. In this case, the clause in Line \ref{line: decideProfile_if_inRange} evaluates to True, but the clause in Line \ref{line: decideProfile_if_newSector} evaluates to False. However, the clause in Line \ref{line: decideProfile_if_oldSector} evaluates to True, because these Booleans were set in Line \ref{line: decideProfile_initTables} to True to initialize the algorithm. In this case, $\feas[1,1]$ is set to True, which is correct since $S_1=\{x_1\}$ is a compatible traversal sector. Furthermore, $\hmin[1,1]$ and $\hmax[1,1]$ are set correctly. Otherwise, if $x_1 \notin [\min_1,\max_1]$, then there exists no compatible traversal sectors and $\feas[1,1],\hmin[1,1],$ as well as $\hmax[1,1]$ remain set to False.

For the induction step consider any $h,t \geq 1$ with $h >1$ or $t > 1$. Assume $x_t \in [\min_h,\max_h]$. 
If $\feas[h,t-1]$ is True, then by induction there exist compatible traversal sectors $S_1,\dots,S_h$ for $(x_1,\dots,x_{t-1})$. Therefore there exist compatible traversal sectors  of $(x_1,\dots,x_t)$ by adding $x_t$ to $S_h$. Furthermore, $\hmin[h,t]$ and $\hmax[h,t]$ are set correctly.

Otherwise, if $\feas[h,t-1]$ is False, then we check in Line \ref{line: decideProfile_if_oldSector}, if
 $\hmin[h-1,t-1]$ and $\hmax[h-1,t-1]$ are True. By induction and Observation \ref{obs:combination_ab} this is the case if and only if there exist compatible traversal sectors $S_1,\dots,S_{h-1}$ for $(x_1,\dots,x_{t-1})$ and it holds that $\min_{h-1} \in S_{h-1}$ as well as $\max_{h-1} \in S_{h-1}$.
 As such, there exist compatible traversal sectors $S_1,\dots,S_{h}$ with $S_h=\{x_t\}$ for $(x_1,\dots,x_{t})$. Furthermore, $\hmin[h,t]$ and $\hmax[h,t]$ are set correctly with respect to $S_h$. 

Now, assume $x_t \notin [\min_h,\max_h]$. In this case, there exists no compatible traversal sectors and $\feas[h,t],\hmin[h,t],$ as well as $\hmax[h,t]$ remain set to False.
\end{proof}

\begin{lemma} \label{lem: check profile feasible}
Given $x \in \RR^\comp$ and $p =\big( (min_1, max_1),\dots, (min_\ell, max_\ell)\big) \in {(\vectorset(x)^2)}^\ell$. Then Algorithm \ref{alg: Partition of Profile} with input $(x,p)$ takes  $O(\comp\ell)$ time and returns True iff  there exists a traversal $T$ s.t. $p$ is an  $\ell$-profile for $(x,T)$.
\end{lemma}

\begin{proof}
Let $\feas, \hmin$ and $\hmax$ be the tables constructed by Algorithm \ref{alg: Partition of Profile} and assume 
the algorithm returns True. Then, it must be that $\hmin[\ell,\comp]$ and $\hmax[\ell,\comp]$, as well as $\feas[\ell,\comp]$ are all set to True.
By Lemma \ref{lem:partitionDP_invariant} and Observation~\ref{obs:combination_ab},  there exist traversal sectors $S_1,\dots,S_\ell$ of $x$ that are compatible with $p$ and it holds that $\min_\ell \in S_\ell$ and $\min_\ell \in S_\ell$. By definition there exists a traversal  $T$ s.t. $S^{(x,T)}_i = S_i$, for $1\leq i \leq \ell$ implying that $p$ is an $\ell$-profile of $(x,T)$.

Next assume that there exists a traversal $T$ s.t. $p$ is an $\ell$-profile for $(x,T)$. Then  $S^{(x,T)}_1,\dots,S^{(x,T)}_\ell$ are traversal sectors for $x$ that are compatible with $p$ and it holds that $\min_\ell \in S^{(x,T)}$ and $\max_\ell \in S^{(x,T)}$. By Lemma \ref{lem:partitionDP_invariant}, 
$\hmin[\ell,\comp]$ and $\hmax[\ell,\comp]$, as well as $\feas[\ell,\comp]$ are correctly set to True.

The initialization of  $\feas$ takes $O(\ell \comp)$ time and the algorithm takes $O(\ell \comp)$ many iterations of the inner for-loop, which require constant time each.
\end{proof}

After establishing how to compute the set of $\ell$-profiles we are ready to state the complete procedure, which is given by Algorithm \ref{alg: ComplexityReduction} and analyzed in Theorem \ref{theorem:dimreduction}.

\begin{algorithm}[H]
\caption{ComplexityReduction}
\label{alg: ComplexityReduction}
\KwIn{$x \in \RR^\comp$, $\ell \in \NN$, $\varepsilon \in \RR_{\geq 0}$}

$\tilde{x}\gets \text{ReduceValueDomain}(x,\ell,\varepsilon)$

 Let $\comp'$ be smallest value s.t. it exists $x' \in \vectorset(\tilde{x})^{\comp'}$ with $D_{(x',\ell)} = D_{(\tilde{x},\ell)}$  

\Return $x'$
\end{algorithm}

\thmdimreduction*
\begin{proof}
    Let $\tilde{x}$ be the time series returned by ReduceValueDomain$(x,\ell,\varepsilon)$. Then by Lemma \ref{lem: reduce value domain} for all $y \in \RR^\ell$ 
    \[(1-\varepsilon) \dist_{dF}(x,y) \leq \dist_{dF}(\tilde{x},y) \leq (1+\varepsilon) \dist_{dF}(x,y)\]
    with  $|\vectorset(x')| \in O(\ell/\varepsilon)$.  
    By Lemma \ref{lem: existence matching profiles} there exists a time series $x' \in \vectorset(\tilde{x})^{\comp'}$, with $D_{(x',\ell)} = D_{(\tilde{x},\ell)}$
    which further implies by Lemma \ref{lem: matching profiles} that for arbitrary $y\in \RR^\ell$ we have $\dist_{dF}(x',y) = \dist_{dF}(\tilde{x},y)$. It follows that  \[(1-\varepsilon) \dist_{dF}(x,y) \leq \dist_{dF}(x',y) \leq (1+\varepsilon) \dist_{dF}(x,y).\]

    By \Cref{lem: reduce value domain}, the time needed to compute $\tilde{x}$ is in $O(\ell \comp\log^2 \comp)$.   
    Then, by \Cref{lem: check profile feasible}, computing $D_{(\tilde{x},\ell)}$ takes $O(\ell/{\eps})^{2\ell}\cdot \comp$ time by enumerating all $\vectorset(\tilde{x})^{2\ell}$ candidate $\ell$-profiles and checking if they are valid $\ell$-profiles for $\tilde{x}$ using \Cref{alg: Partition of Profile}. 
    To find the time series $x'$ with smallest complexity with its vertices in $\vectorset(\tilde{x})$ such that $D_{(x',\ell)} = D_{(\tilde{x},\ell)}$, we enumerate all vectors over 
    $\vectorset(\tilde{x})$ in increasing length until we find a vector with the same set of $\ell$-profiles as $\tilde{x}$. For each increasing value of $i=1,2,\ldots,z'$, for each vector $y\in \vectorset(\tilde{x})^i$ we compute its set of $\ell$-profiles, by enumerating all $\vectorset(\tilde{x})^{2\ell}$ candidate $\ell$-profiles and checking if they are valid $\ell$-profiles for $x'$ using \Cref{alg: Partition of Profile}.  By definition, we have $ \comp' \in O(2^{O(\ell/\eps)^{(2\ell+2)}})$. 
    The overall running time of this step is in $O({\comp'}^2)\cdot O(\ell/\eps)^{\comp'}\cdot O(\ell/\eps)^{2\ell} \subseteq O(\ell/\eps)^{2\comp'}$. 
\end{proof}

\subsection{Near Linear Time Approximation Scheme for $(k,\ell)$-Median with $d = 1$}

First, we reduce the complexity $\comp$ of the input time series to a complexity that only depends on $\ell,\eps$ using \Cref{alg: ComplexityReduction}, to obtain a set of clients $X$ with constant complexity. Then, 
we compute a set of candidate centers $Y$ using \Cref{lemma:mediancandidates}, which has a size near linear in $n$, and independent of $\comp$. 

Strictly speaking, the discrete Fr\'echet distance is a pseudo-metric as distance between distinct curves can be $0$. Therefore, one has to consider a metric space defined on the equivalence classes of curves with pairwise distance $0$. Each equivalence class has a corresponding representative time series and, for some given time series $x$ its representative can be computed by removing all duplicates of neighboring values in $x$. 

\begin{algorithm}
\caption{NLTAS \label{alg:ltas}}
\KwIn{$P\subset \RR^\comp$, $k,\ell \in \NN$, $\eps \in (0,1)$}

$X\gets $ for each $x\in P$, run \Cref{alg: ComplexityReduction} with input $x,\ell,\eps/16$ and store output in $X$

$Y\gets$ output of the algorithm of \Cref{lemma:mediancandidates} with input $X,k,\ell$ and $\eps/12$.

 $Sol \gets$ output of the algorithm of \Cref{thm:kmedian_bounded_centers} with input $
X,Y,k,\eps/4 $

\Return $Sol$

\end{algorithm}

\begin{theorem}
\label{theorem:final_formal}
Let $\eps\in (0,1/2]$ and $\ell \in \mathbb N$ be constants. Given a set $P$ of $n$ real-valued time series of complexity $\comp$, and parameter $k$ Algorithm \ref{alg:ltas} computes in time
$\tilde{O}( n\comp)$
and with success probability at least $1-\eps$ a $(1+\eps)$-approximation to the $(k,\ell)$-median problem under discrete Fréchet distance. 
\end{theorem}
\begin{proof}
By \Cref{theorem:dimreduction}, \Cref{alg: ComplexityReduction} runs in $O(\comp \ell \log^2 \comp) + O(\ell/\eps)^{2 \comp'}$ time and outputs 
a time series $x' \in \RR^{\comp'}$, where $\comp' \in O(2^{O(\ell/\eps)^{(2\ell+2)}})$. Hence, computing $X\subset \RR^{\comp'}$ costs 
$O(n\comp \ell \log^2 \comp) + n\cdot O(\ell/\eps)^{2 \comp'}$ 
time. 
The algorithm of \Cref{lemma:mediancandidates} runs in time $O(2^{O(\ell)} \ell ^2n\comp'(\log(\comp') \log(\comp'/\ell)+\log^9n))  + O(n\log n)\cdot O(1/\eps)^{\ell}$ and outputs a set $Y \subset \RR^\ell$ of size $O( n \log n) \cdot O(1/\eps)^{\ell}$. 
Finally, we run the algorithm of \Cref{thm:kmedian_bounded_centers} on $X,Y$, where $|X\cup Y| \in  O(n\log n)\cdot O(1/\eps)^{\ell}$. 
By \Cref{prop:ddim}, the doubling dimension of the ambient space of $Y$ is in $O(\ell)$. By \Cref{thm:kmedian_bounded_centers} and assuming constant time distance evaluations, the running time of this step is 
$2^{2^t} \cdot \tilde{O}(n)$, where $t = O\left(\ell \log\frac{\ell}{\epsilon}\right).$
Since $\eps, \ell$ and $\comp'$ are constants, distances can be computed in constant time and furthermore we can simplify the running time to $\tilde{O}(n).$

To show correctness, first observe that by \Cref{theorem:dimreduction}, the cost of any solution for $P$ is preserved up to a factor of $(1\pm\frac{\eps}{16})$ after reducing the complexity of each $x\in P$ resulting in $X$. The optimal solution for $X$ corresponds to a solution for $P$ which is at most $\frac{1+\eps/16}{1-\eps/16}\leq 1+\eps/4$ from the optimal. 
Now by \Cref{lemma:mediancandidates}, there will be a solution consisting of medians from $Y$ which has cost at most $(1+\frac{\eps}{4})$ times that of the optimal solution for $X$. Therefore, the execution of the algorithm of \Cref{thm:kmedian_bounded_centers} on $X,Y$ will find a solution consisting of $k$ time series from $Y$ that has a cost of at most  $(1+\frac{\eps}{4})^2$ times that of the optimal solution for $X$. This solution has a cost of at most  $(1+\frac{\eps}{4})^3 \leq 1+\eps$ times that of the optimal solution for $P$. 
\end{proof}

\subsection{Coreset for $(k,\ell)$-Median with $d = 1$ }

In this section, we discuss an additional implication of our complexity reduction result in the context of constructing coresets for $(k,\ell)$-median. Given a set $P$ of $n$ time series, an $\eps$-coreset of $P$ for the $(k,\ell)$-median problem, is a weighted set $S\subseteq P$ such that for any $C\subset \RR^{\ell}$, $|C|=k$, 
\[(1-\eps)\sum_{x \in P}\min_{c \in C} \dist_{dF}(x,c)\leq \sum_{x \in S} w_x \cdot \min_{c \in C}  \dist_{dF}(x,c) \leq (1+\eps)\sum_{x \in P}\min_{c \in C} d_{dF}(x,c) ,\] 
where $w_x$ is the weight associated with $x$. 

A recent result of Cohen-Addad et al.~\cite[Corollary 7.2]{Cohen-Addad2025Coreset} implies coresets for the $(k,\ell)$-median problem under the discrete Fr\'echet distance of size $\tilde{O}(\eps^{-2} k \ell \log(\comp) )$.  By combining this result with \Cref{theorem:dimreduction}, we obtain coresets of size $\tilde{O}(\eps^{-2}k \ell \log (\comp'))$, i.e., completely independent of the size of the input since $\comp'$ is a function of $\ell,\eps$. The result is formally stated as follows.
\begin{corollary}
        Let $\varepsilon \in (0,1)$ and $k,\ell \in \NN$. For any set $P$ of time series there exists an $\eps$-coreset for the $(k,\ell)$-median problem of size $\tilde{O}(\eps^{-2}k \ell \log (\comp'))$, where $\comp' \in  O(2^{O(\ell/\eps)^{(2\ell+2)}})$. 
\end{corollary}
\addcontentsline{toc}{section}{References}
\bibliographystyle{alphaurl}
\begin{small}	
    \bibliography{ref}
\end{small}

\begin{appendices}
    \crefalias{section}{appendix}
    \section{Constant Approximation Algorithms}
\label{sec: constant-approx-algos}

It has been shown in \cite{Cohen-AddadFS21} that one can obtain near linear time algorithms to compute a constant approximation to the $k$-median problem as well as the facility location problem when the input points come from a doubling space. In detail, the running time for the $k$-median problem is $O(2^{O(\ddim)}n \log^9 n)$ (where $n=|X\cup Y|)$ and combines the spanner construction (with parameter $\epsilon=1/2$ from \cite{Har-PeledM06})  with a graph-based algorithm by Thorup \cite{thorup2005}. 
For the facility location problem one can obtain a $O(n \log n)$ time $2^{O(\ddim)}$-approximation algorithm, where $n = |X\cup Y|$ is the size of the instance~\cite[Section A.3]{Cohen-AddadFS21}. 

We would like to extend their result to our setting when only one of the two sets has bounded doubling dimension. Assume first that this is the set $Y$ of candidate centers. In this case, we first compute for every $x\in X$ its $2$-approximate nearest neighbor in $Y$ using the algorithm in \Cref{lemma:ANN_doubling_data} which runs in time $O(2^{O(\ddim)} \cdot n \log \Delta)$. Then we replace every client by its approximate nearest neighbor and solve the resulting problem using one of the algorithms described above. We claim that the resulting solution is also a constant factor approximation for the original problem.

Indeed, moving a client by a distance $D$ while maintaining the assignment to the same cluster center/facility will change the objective function by a value of at most $D$. Since the distance of $x\in X$
to its closest candidate facility is a lower bound on its connection cost, this implies that our construction changes the cost of any fixed solution by at most $2\opt$.  Thus, the cost of the optimal solution will become at most $3\opt$ and the solution of a $c$-approximation algorithm
on the new instance will have cost at most $3c \opt$. Going back to the original data set changes the cost of this solution by at most $2\opt$, so it is a $(3c+2)$-approximation. 
The above argumentation holds for both the $k$-median and the facility location problem.

It remains to deal with the case that $X$ has bounded doubling dimension and $Y$ not.
In this case, we define the following metric embedding as in \cite{Har-PeledK13}.
For $y\in Y$ let $\bar y$ denote a $2$-approximate nearest neighbor of $y$ in $X$, otherwise if $y \in X$ define $\bar{y} = y$.
We define a new metric space $(X\cup Y, \dist')$ where 
$\dist'(x,y) = |\dist(x,\bar{x}) - \dist(y,\bar{y})| + \dist(\bar{x},\bar{y})$ for $x,y\in X\cup Y$. Then, by Lemma 2.4 in \cite{Har-PeledK13} and the fact that $\dist(y,\bar{y}) \leq 2 \dist(y,X)$ for all $y \in X\cup Y$ we get the following Lemma.
\begin{lemma} \label{lem: constant_approx_embedding} For the metric space $(X \cup Y, \dist')$ the following holds.
    \begin{enumerate}
        \item For $x,y \in X $, $\dist'(x,y) = \dist(x,y)$
        \item For $x \in X$ and $y \in Y$, $\dist(x,y) \leq \dist'(x,y) \leq 6 \dist(x,y)$
        \item The doubling dimension is $O(\ddim(X))$
    \end{enumerate}
\end{lemma}

Computing $\bar{y}$ for some $y \in Y$ can be done in $O(2^{O(\ddim(X))}n\log \Delta)$ time using the data structure of Lemma $\ref{lemma:ANN_doubling_data}$.

We have designed an embedding of a metric space into a doubling metric in such a way that pairwise distances between $X$ and $Y$ and within $X$ are maintained upto a constant factor. Since only these distances appear in the objective function, we obtain that any constant factor solution w.r.t. $(X\cup Y, \dist')$
is also a constant factor approximation in $(X\cup Y, \dist)$.
Since we can evaluate $\dist'(p,q)$ for arbitrary $p,q\in X\cup Y$ in constant time per query (assuming constant time access to $\dist(p,q)$) we can use the previously known algorithms and get the same running times.

\section{Removing Dependency on Aspect Ratio}\label{section: remove_aspect_ratio}

In the following we will show how to transform our input instance in near linear time to some instance that has aspect ratio poly-logarithmic in $n$ and $m$, while only incurring a factor $(1+\eps)$ on the pairwise distances. Our construction is independent of the case wether $X$ or $Y$ is doubling. We will formulate it with respect to $X$ being doubling but the construction for the other case is symmetric.

We start by stating a data-structure to compute an $O(n)$-ANN in $X$ for a point $y$ that has construction and query time independent of the aspect ratio. We will further utilize this to compute a $O(n)$-approximation for the facility location or $k$-median problem, which will be useful for our instance transformation. 

We show that the data-structure of \cite[Lemma 4.2]{Har-PeledM06} works for partially doubling metrics. 
\begin{lemma}\label{lem: linear-ann}
    Let $(X \cup Y,\dist)$ be a metric space with $|X| = n$ and $\ddim(X) \leq \ddim$. Then there is an algorithm that builds in $2^{O(\ddim)}O(n\log n)$ time a data-structure, that given a query point $y\in Y$, returns a $2n$-ANN of $y \in X$. The query time is $2^{O(\ddim)}O(\log n)$.
\end{lemma}
The following proof is identical to that of Lemma 4.2 from \cite{Har-PeledM06}. We give our own formulation to keep the paper self-contained.
\begin{proof}

By applying the data-structure of Lemma 4.2 in \cite{Har-PeledM06} on $X$ we get a binary search tree $S$, in which each vertex of the tree $v$ is associated with a point $x_v \in X$ and some radius $r \in \RR_{\geq 0}$. Each vertex $v$ in $S$ satisfies that
\[\frac{n}{2\ddim^3} \leq |B_X(x_v,r_v)|\leq (1-\frac{1}{2\ddim^3})n\] and $B_X(x_v,(1+\frac{1}{2n}r_v)) \setminus B_X(x_v,(1-\frac{1}{2n}r_v)) = \emptyset$. The left subtree of $v$ is associated with points in $B_X(x_v,r_v)$ and the right subtree with points in $X \setminus B_X(x_v,r_v)$. The depth of $S$ is $O(\ddim^3 \log n)$.

Given some $y \in Y$ we can find a $2n$-ANN in $X$ in time $O(\ddim^3 \log n)$ as follows. Let $u$ be the root of $S$, then if $\dist(y,x_u) \leq r_u$ recurse of the left subtree and otherwise on the right. Return the nearest point to $y$ among all $x_v$, where $v$ is a vertex in $S$ that was traversed this way.

Let $y^*$ be the nearest neighbor of $y$ in $X$ and let $y'$ be the output of the described procedure. Let $w$ be the lowest common ancestor of the vertices in $S$ that correspond to $y$ and $y'$. Assume that $y'$ is contained in the left subtree and $y^*$ in the right subtree of $w$. Then,  by the properties of the tree  $\dist(y,x_w) \leq r_w$ and $\dist(y,y^*)\geq r_w/2n$ implying that $x_w$ is already a $2n$-ANN.
Now consider the other case where $y^*$ is contained in the left subtree and $y'$ in the right subtree of $w$. Then $\dist(y,y') \leq \dist(y,x_w)$ and by the property of the tree it holds $\dist(y,y^*) \geq r_w/2n + (\dist(x_w,y) - r_w)$ and the ratio of these two terms is at most $2n$ and $x_w$ is therefore a $2n$-ANN of $y$ in $X$.
\end{proof}

Now let $(X\cup Y,\dist'')$ be the metric spaced introduced in section \ref{sec: constant-approx-algos} with the exception that $\bar{x}$ is an $O(n)$-ANN for $x \in Y$.
Then Lemma \ref{lem: constant_approx_embedding} implies the following Corollary.

\begin{corollary}

 \label{cor: linear_approx_embedding} For the metric space $(X \cup Y, \dist'')$ the following holds.
    \begin{enumerate}
        \item For $x,y \in X $, $\dist''(x,y) = \dist(x,y)$
        \item For $x \in X$ and $y \in Y$, $\dist(x,y) \leq \dist''(x,y) \leq O(n) \dist(x,y)$
        \item The doubling dimension is $O(\ddim(X))$
 \end{enumerate}
\end{corollary}
The embedding can be computed in $O(2^{O(\ddim(X))}n \log n)$ time using the data structure of Lemma \ref{lem: linear-ann}.
One can now combine the embedding to metric $(X \cup Y,\dist'')$ and constant factor algorithms for $k$-median or facility location problem \cite{Cohen-AddadFS21} to get an $O(n)$-approximation for the respective clustering problem in near linear time.

Next we give the transformation to an input instance that has bounded aspect ratio. The idea is to first apply the metric embedding of Lemma \ref{cor: linear_approx_embedding} to $(X\cup Y,\dist)$ to get the metric space $(X\cup Y,\dist'')$. 
and then apply the transformation given by Lemma 32 in \cite{Cohen-AddadFS21}.

\begin{lemma}\label{lem: transform_to_bounded_aspect_ratio}
    Let $P$ denote the facility location problem or the $k$-Median problem.
    Given $(X\cup Y,\dist)$ with $\ddim(X) \leq \ddim$ and $|X \cup Y| = n$, $\eps >0$. Then one can compute in time $\tilde{O}(n+m)$ a set of instances $(X_1\cup Y_1 ,\dist_1),\dots,  (X_\ell\cup Y_\ell ,\dist_\ell)$ s.t.

    \begin{enumerate}
        \item $X= \bigcup_{i=1}^\ell X_i$ and $Y = \bigcup_{i=1}^\ell Y_i$
        \item For $i \in [\ell]$, the instance $(X_i\cup Y_i,\dist_i)$ has aspect ratio $O(n^6/\varepsilon)$
        \item Let $(X\cup Y,\bar{\dist})$ be the metric where $\bar{\dist}(x,y) = \dist_i(x,y)$ for $x,y\in X_i\cup Y_i$ and  $\bar{\dist} = \infty$ otherwise. Let $\opt$ be the optimal value for problem $P$. Then it holds
        \begin{itemize}
            \item There exists a solution on $(X\cup Y,\bar{\dist})$ with cost $(1+\eps/n) \opt$
            \item Every solution on $(X\cup Y,\bar{\dist})$ for problem $P$ with cost $A$ induces a solution with cost at most $A + \eps \opt /n$ for $(X \cup Y,\dist)$
        \end{itemize}
    \end{enumerate}
\end{lemma}

\begin{proof}
    We first compute the metric embedding $(X\cup Y,\dist'')$, which preserves all distances $\dist(x,y)$ with $x \in X$ and $y \in X\cup Y$ within factor $O(n)$. 
    We will further equip the embedding $(X \cup Y,\dist'')$ with the following mechanism. Let $\gamma$ be the cost of an constant approximation for a problem $P$.
    Whenever we evaluate a distance $\dist''(x,y) < \eps \gamma/n^3$ we set the distance $\dist''(x,y) = \eps \gamma/n^3$. By doing so we ensure that the smallest distance that is queried on $(X\cup Y,\dist'')$ is at least $\eps \gamma/n^3$. Note that this satisfies the triangle inequality. Doing so induces an additive error of $O(\eps \opt /n^3)$.
    We then apply the construction given in \cite{Cohen-AddadFS21} on $(X \cup Y, \dist'')$ and partition $(X \cup Y,\dist)$ accordingly. By using  $\dist''$ we incur an additional factor $O(n)$ to the aspect ratio.
\end{proof}

The following Lemma corresponds to Lemma 33 in \cite{Cohen-AddadFS21} and the proof is analogous. It will allow us to trade the logarithmic dependency on $\Delta$ for additional factors that are only poly-logarithmic in the size of the input instance.

\begin{lemma}\label{lem:merge_instances}
    Let $P$ denote the facility location problem or the $k$-Median problem. Given a set of instances $(X \cup Y,\dist)$ and $(X_1\cup Y_1 ,\dist_1),\dots,  (X_\ell\cup Y_\ell ,\dist_\ell)$ as in Lemma \ref{lem: transform_to_bounded_aspect_ratio} and an algorithm with running time $O(n_i(\log n_i)^c f(\Delta))$ to solve $P$ on instances with $n_i$ points and aspect ratio $\Delta$. Then there exists an algorithm that has running time $O(n(\log n)^{c+2} f(O(n^6/\eps)))$ to solve $P$ on $(X \cup Y,\dist)$.
\end{lemma}

\section{Missing Proofs in \Cref{sec:kmedian_bounded_centers}}

\subsection{Proof of \Cref{lemma:kmedian_bounded_centers_good_solution_construction}}
\label{appendix:proof_of_lemma_kmedian_bounded_centers_construction}

\lemmakmedianboundedcentersgoodsolutionconstruction*

\begin{proof}[Proof of \Cref{lemma:kmedian_bounded_centers_good_solution_construction}]

    We first state how to construct $F$.
    Let $F^* \subseteq Y$ be the optimal $k$-median solution.
    Wlog, assume $|F^*| = |S| = k$, since otherwise we can add arbitrary facilities to them.
    For $f \in F^*$, recall that $\proj{S}{f}$ is the point in $S$ closest to $f$.
    For $s \in S$, denote $S^{-1}_{F^*}(s) = \{f \in F^* \colon \proj{S}{f} = s\}$.
    Denote $S_0 = \{s \in S \colon |S^{-1}_{F^*}(s)| = 0\}$, 
    $S_1 = \{s \in S \colon |S^{-1}_{F^*}(s)| = 1\}$
    and $S_{\geq 2} = \{s \in S \colon |S^{-1}_{F^*}(s)| \geq 2\}$.
    Denote $S_{\geq 1} := S_1 \cup S_{\geq 2}$.
    For $s \in S_{\geq 1}$, denote $f_s := \argmin_{f \in S^{-1}_{F^*}(s)} \dist(f, s)$, i.e., $f_s$ is the closest point to $s$ among $S^{-1}_{F^*}(s)$, breaking ties arbitrarily.
    Denote $F_S := \{f \in F^* \colon \exists s \in S, \text{ s.t. } f = f_s\}$.

    Following~\cite{Cohen-AddadFS21}, the center set $F$ is constructed as follows:
    \begin{itemize}
        \item \textbf{Step 1:} Among the centers in $F^* \setminus F_S$, remove from $F^*$ a subset of size $1000 c \epsilon |F^* \setminus F_S|$ that yields the minimum cost increase after removal, where $c$ is the same constant in \Cref{lemma:badly_cut}.\footnote{This step requires us to assume $\epsilon \leq 1/(1000c)$. Since $1000 c$ is a constant, there is no loss of generality to make such assumption.}
        The resulting center set is denoted as $\overline{F}^*$.
        Let $F = \overline{F}^*$.
        \item \textbf{Step 2:} For $s \in S_0$, if $s$ is a bad center, then add $s$ to $F$.
        \item \textbf{Step 3:} For $s \in S_{\geq 1}$, if $s$ is a bad center, then add $s$ to $F$, and remove $f_s$ from $F$.
    \end{itemize}
    Formally, $F := \overline{F}^* \setminus \{f_s \colon s \in \Bad_S \cap S_{\geq 1}\}  \cup \Bad_S$.
    We note that $F$ contains all bad facilities in $S$.
    This proves property (a).

    \paragraph{Property (b).}
    We first show that $\dist(x, \overline{F}^*) \leq 3 \dist(x, F^*) + 2 \dist(x, S)$.
    For $x \in X \cup Y$, if $\proj{F^*}{x}$ is not removed, then $\dist(x, \overline{F}^*) = \dist(x, F^*)$.
    If $\proj{F^*}{x}$ is removed, then $\proj{F^*}{x} \notin F_S$.
    Denote $s = \proj{S}{\proj{F^*}{x}}$. 
    Recall that $f_s$ is the closest point to $s$ in $S_{F^*}^{-1}(s)$, and $f_s \in \overline{F}^*$.
    We have 
    \begin{align*}
        \dist(x, \overline{F}^*) 
        &\leq \dist(x, f_s) &&\text{Since $f_s \in \overline{F}^*$}\\
        &\leq \dist(x, \proj{F^*}{x}) + \dist(\proj{F^*}{x}, s) + \dist(s, f_s) &&\text{By triangle inequality}\\ 
        &\leq \dist(x, F^*) + 2 \dist(\proj{F^*}{x}, s) &&\text{By definition of $f_s$}\\
        &= \dist(x, F^*) + 2 \dist(\proj{F^*}{x}, S) && \text{By definition of $s$}\\ 
        &\leq 3 \dist(x, F^*) + 2 \dist(x, S) &&\text{By triangle inequality}.
    \end{align*}
    Hence, 
    \begin{align}
        \dist(x, \overline{F}^*) \leq \begin{cases}
            \dist(x, F^*), & \text{if } \proj{F^*}{x} \in \overline{F}^*; \\
            3 \dist(x, F^*) + 2 \dist(x, S). &\text{if } \proj{F^*}{x} \notin \overline{F}^*.
        \end{cases}
        \label{eqn:d_x_Fbar}
    \end{align}

    Next, we bound $\dist(x, F)$.
    For $x \in X \cup Y$, if $\proj{\overline{F}^*}{x} \in F$, then $\dist(x, F) \leq \dist(x, \overline{F}^*)$.
    If $\proj{\overline{F}^*}{x} \notin F$, then there exists $s \in \Bad_S \cap S_{\geq 1}$, such that $\proj{\overline{F}^*}{x} = f_s$.
    Moreover, $s \in F$.
    Therefore, 
    \begin{align*}
        \dist(x, F) \leq \dist(x, s)
        \leq \dist(x, f_s) + \dist(f_s, s) 
        = \dist(x, \overline{F}^*) + \dist(\proj{\overline{F}^*}{x}, S)
        \leq 2 \dist(x, \overline{F}^*) + \dist(x, S).
    \end{align*}
    Hence, 
    \begin{align}
        \dist(x, F) \leq \begin{cases}
            \dist(x, \overline{F}^*), &\text{if } \proj{\overline{F}^*}{x} \in F; \\
            2 \dist(x, \overline{F}^*) + \dist(x, S), &\text{if } \proj{\overline{F}^*}{x} \notin F.
        \end{cases}
        \label{eqn:d_x_F}
    \end{align}

    Combining \eqref{eqn:d_x_F} with \eqref{eqn:d_x_Fbar}, we have $\dist(x, F) \leq 6 \dist(x, F^*) + 5 \dist(x, S)$.

    \paragraph{Property (c): size of $F$.}
    By step 2, every bad center in $S_0$ increases the size of $F$ by $1$.
    By step 3, bad centers in $S_{\geq 1}$ do not increase the size of $F$.
    Therefore, 
    \[
    |F| = |F^*| - 1000 c \epsilon |F^* \setminus F_S| + |\Bad_S \cap S_0|.
    \]

    By \Cref{lemma:badly_cut}, 
    \begin{align*}
        \E[|\Bad_S \cap S_0|] = \sum_{s \in S_0} \Pr[s \in \Bad_S]
        \leq c \epsilon \cdot |S_0|.
    \end{align*}
    By Markov's inequality, with probability $0.999$, 
    $|\Bad_S \cap S_0| \leq 1000c \epsilon |S_0|$.
    Observe that $|S_0| = |F^* \setminus F_S|$.
    Therefore, with probability $0.999$, 
    \[
    |F| \leq |F^*| - 1000 c \epsilon |F^* \setminus F_S| + 1000 c \epsilon |S_0| 
    = |F^*|
    \leq k.
    \]
   
    \paragraph{Property (c): cost of $F$.}
    We first show that $\median(X, \overline{F}^*) \leq (1 + \epsilon) \optkmed(X, Y)$ by an averaging argument:
    Consider removing from $F^*$ a \emph{random} subset of size $\epsilon |F^* \setminus F_S|$ and denote the resulting (random) set to be $\widetilde{F}^*$.
    We can obtain a similar bound as \eqref{eqn:d_x_Fbar}.
    \begin{align*}
        \dist(x, \widetilde{F}^*) \leq \begin{cases}
            \dist(x, F^*), & \text{if } \proj{F^*}{x} \in \widetilde{F}^*; \\
            3 \dist(x, F^*) + 2 \dist(x, S). &\text{if } \proj{F^*}{x} \notin \widetilde{F}^*.
        \end{cases}
    \end{align*}
    Moreover, $\proj{F^*}{x} \notin \widetilde{F}^*$ happens with probability at most $\epsilon$.
    Therefore, 
    \begin{align*}
        \E[\median(X, \widetilde{F}^*)]
        &= \sum_{x \in X} \E[\dist(x, \widetilde{F}^*)] 
        \leq \sum_{x \in X} \dist(x, F^*) + O(\epsilon) \cdot (3 \dist(x, F^*) + 2 \dist(x, S)) \\
        &= (1 + O(\epsilon)) \optkmed(X, Y) + O(\epsilon) \median(X, S) \\
        &\leq (1 + O(\epsilon)) \optkmed(X, Y).
    \end{align*}
    Since $\overline{F}^*$ minimizes the cost increase after removal, we have 
    \begin{align*}
        \median(X, \overline{F}^*) \leq \E[\median(X, \widetilde{F}^*)] \leq (1 + O(\epsilon)) \optkmed(X, Y).
    \end{align*}

    We next show that $\median(X, F) \leq (1 + O(\epsilon)) \median(X, F^*)$ with probability $0.999$.
    Note that in \eqref{eqn:d_x_F}, 
    Moreover, $\proj{\overline{F}^*}{x} \notin F$ only if $\proj{S}{\proj{\overline{F}^*}{x}} \in \Bad_S$, which happens with probability $c \epsilon$ by \Cref{lemma:badly_cut}.
    Therefore, 
    \begin{align*}
        \E[\median(X, F)] 
        &= \sum_{x \in X} \E[\dist(x, F)] \\ 
        &\leq \sum_{x \in X} \left[
            \dist(x, \overline{F}^*) + O(\epsilon) (2 \dist(x, \overline{F}^*) + \dist(x, S))
        \right]\\
        &= (1 + O(\epsilon)) \median(X, \overline{F}^*)
        + O(\epsilon) \median(X, S) \\
        &\leq (1 + O(\epsilon)) \optkmed(X, Y).
    \end{align*}
    Applying Markov's inequality over (non-negative) random variable $\median(X, F) - \optkmed(X, Y)$, we have 
    $\median(X, F) \leq (1 + O(\epsilon)) \optkmed(X, Y)$ holds with probability $0.999$.
    Rescaling $\epsilon$ concludes the proof.
\end{proof}

\subsection{Proof of \Cref{lemma:discretization_is_ok}: Cost of the Returned Solution}
\label{appendix:proof_of_lemma:discretization_is_ok}

\lemmadiscretizationisok*

Consider an entry $(C, \veca_C, \vecb_C, \Gamma_C)$ in the DP table.
Let $g(C, \veca_C, \vecb_C, \Gamma_C)$ be the value computed by the dynamic program (\Cref{eqn:compute_h_base,eqn:compute_h,eqn:compute_g_C_accelerated}).
Let $G(C, \veca_C, \vecb_C, \Gamma_C)$ be the true value of that entry, i.e., the
minimum number of facilities required to be placed in $C$, such that the revealed cost inside $C$ is at most $\Gamma_C$, which is formally defined in \eqref{eqn:kmedian_DP_table_entry_value}.
Recall that $G(C, \veca_C, \vecb_C, \Gamma_C)$ can be different from $g(C, \veca_C, \vecb_C, \Gamma_C)$, mainly because of the discretization of $\Gamma_C$ into powers of $(1 + \epsilon')$, where $\epsilon' = \frac{\epsilon}{2^{O(\ddim)} \log \Delta}$.
We show in the following lemma that $G$ can be two-sided bounded by $g$.

\begin{lemma}\label{lemma:G_bounded_by_g}
    Let $0 \leq \ell \leq L$ and $C \in \decom_\ell$ be a level $\ell$ cluster.
    Then for every configuration $\veca_C, \vecb_C$, and $\Gamma_C \in [\median(X, S)/n, 2 \median(X, S)]$ being powers of $(1 + \epsilon')$, it holds
    \begin{align*}
        g(C, \veca_C, \vecb_C, \alpha_\ell \Gamma_C) \leq G(C, \veca_C, \vecb_C, \Gamma_C) \leq g(C, \veca_C, \vecb_C, \Gamma_C),
    \end{align*}
    for $\alpha_\ell = (1 + \epsilon')^{\ell \cdot 2^{O(\ddim)}}$.
\end{lemma}

We use \Cref{lemma:G_bounded_by_g} to prove \Cref{lemma:discretization_is_ok}.

\begin{proof}[Proof of \Cref{lemma:discretization_is_ok}]
    Applying \Cref{lemma:G_bounded_by_g} to $X' \in \decom_L$ and noting $\alpha_{L} = (1 + \epsilon')^{2^{O(\ddim)} \log \Delta} \leq 1 + \epsilon$,
    we have that for every configuration $\veca, \vecb$ and $\Gamma \in [\median(X, S)/n, 2 \median(X, S)]$, 
    \[g(X', \veca, \vecb, (1 + \epsilon) \Gamma) \leq G(X', \veca, \vecb, \Gamma) \leq g(X', \veca, \vecb, \Gamma).\]

    Recall our algorithm outputs the smallest $\Gamma$, such that there exists $\veca, \vecb$ with $g(X', \veca, \vecb, \Gamma) \leq k$, and $\widehat{F}$ is the solution corresponding to this $\Gamma$.
    Therefore, $G(X', \veca, \vecb, \Gamma) \leq k$.
    By the definition of $G$, this implies the optimal $k$-median value of $X'$ (under portal-respecting distance) is at most $\Gamma$.

    By the minimality of $\Gamma$, for every $\veca, \vecb$, we have $g(X', \veca, \vecb, \Gamma/(1 + \epsilon')) > k$.
    Then $G(X', \veca, \vecb, \frac{\Gamma}{(1 + \epsilon')(1 + \epsilon)}) \geq g(X', \veca, \vecb, \Gamma/(1 + \epsilon')) > k$.
    This implies the optimal $k$-median value of $X'$ (under portal-respecting distance) is greater than $\frac{\Gamma}{(1 + \epsilon')(1 + \epsilon)}$.

    We thus conclude that 
    \begin{align*}
        \sum_{x \in X'} \ccostport(x, \widehat{F}) 
        \leq \Gamma
        < (1 + \epsilon')(1 + \epsilon) \min_{F \subseteq Y, |F| \leq k} \sum_{x \in X'} \ccostport(x, F)
    \end{align*}

    Rescaling $\epsilon$ concludes the proof.
\end{proof}

Finally, we prove \Cref{lemma:G_bounded_by_g}.

\begin{proof}[Proof of \Cref{lemma:G_bounded_by_g}]
    We prove the lemma by induction on $\ell$.
    Throughout, denote by $\tau := |\child(C)|$ for simplicity.
    
    When $\ell = 0$, it corresponds to the base cases of the DP, where the computation of $g$ is accurate.
    Hence, 
    \[
    g(C, \veca_C, \vecb_C, \Gamma_C) = G(C, \veca_C, \vecb_C, \Gamma_C) \geq G(C, \veca_C, \vecb_C, \alpha_0 \Gamma_C) = g(C, \veca_C, \vecb_C, \alpha_0 \Gamma_C).
    \]

    Assume the inequality holds for level $\ell - 1$, and consider any level $\ell$ cluster $C$.
    By the dynamic program (\Cref{eqn:compute_h,eqn:compute_h_base,eqn:compute_g_C_accelerated}), 
    there exists a sequence of values $\Gamma_1, \Gamma_2, \dots, \Gamma_{\tau}$ which are powers of $(1 + \epsilon')$ and satisfy $\sum_i \Gamma_i \leq \Gamma_C'$, together with a sequence of configurations $\veca_1, \vecb_1, \dots, \veca_{\tau}, \vecb_{\tau}$, such that
    \begin{equation}\label{eqn:g_C_w_sum}
        g(C, \veca_C, \vecb_C, \Gamma_C) = \sum_{i = 1}^{\tau} g(D_i, \veca_i, \vecb_i, \Gamma_i).
    \end{equation}
    Since configurations $\{\veca_i, \vecb_i\}$ are consistent with $\veca_C, \vecb_C$, and $\sum_i \Gamma_i \leq \Gamma_C'$, we have 
    \begin{align*}
        G(C, \veca_C, \vecb_C, \Gamma_C) &\leq \sum_{i = 1}^{\tau} G(D_i, \veca_i, \vecb_i, \Gamma_i) && \text{By definition of $G$} \\
        & \leq \sum_{i = 1}^{\tau} g(D_i, \veca_i, \vecb_i, \Gamma_i) && \text{Induction hypothesis} \\
        &= g(C, \veca_C, \vecb_C, \Gamma_C) &&\text{By \eqref{eqn:g_C_w_sum}}.
    \end{align*}

    For the other direction, by the definition of $G$, there exists a sequence of values $\Phi_1, \Phi_2, \dots, \Phi_{\tau}$ which satisfy $\sum_i \Phi_i \leq \Gamma_C'$, together with a sequence of configurations $\veca_1, \vecb_1, \dots, \veca_{\tau}, \vecb_{\tau}$, such that
    \begin{equation}\label{eqn:G_C_w_sum}
        G(C, \veca_C, \vecb_C, \Gamma_C) = \sum_{i = 1}^{\tau} G(D_i, \veca_i, \vecb_i, \Phi_i).
    \end{equation}
    Note that $\Phi_i$'s are no longer necessarily powers of $(1 + \epsilon')$.
    Let $R$ be the operator that rounds every value $v$ to the smallest power of $(1 + \epsilon')$ greater than $v$, i.e., $R(v) = (1 + \epsilon')^{\lceil \log_{(1 + \epsilon')} v \rceil}$.
    Define $\Gamma_i := R(\Phi_i)$ for every $i \in [\tau]$.
    We have
    \begin{align}
        G(C, \veca_C, \vecb_C, \Gamma_C) &= \sum_{i = 1}^{\tau} G(D_i, \veca_i, \vecb_i, \Phi_i) &&\text{By \eqref{eqn:G_C_w_sum}} \notag \\
        & \geq \sum_{i = 1}^{\tau} G(D_i, \veca_i, \vecb_i, \Gamma_i) &&\text{Since $\Gamma_i \geq \Phi_i$} \notag \\
        & \geq \sum_{i = 1}^{\tau} g(D_i, \veca_i, \vecb_i, \alpha_{\ell - 1} \Gamma_i) &&\text{Induction hypothesis}. \label{eqn:G_c_w_prime_sum}
    \end{align}
    Denote $\Gamma_i' := \alpha_{\ell - 1} \Gamma_i$ for short.
    Consider the following sequence of values, which are powers of $(1 + \epsilon')$.
    \begin{align*}
        &\Psi_{\tau - 1} = \Gamma_\tau', \\
        &\Psi_i = R(\Psi_{i + 1} + \Gamma_{i + 1}'), \quad 0 \leq i \leq \tau - 2.
    \end{align*}
    Note that $\{\veca_i, \vecb_i\}$ are consistent with $\veca_C, \vecb_C$, and that $\Psi_{i + 1} + \Gamma_{i + 1}' \leq \Psi_i$, we thus have
    \begin{align*}
        & h(C, D_{\tau - 1}, \veca_C, \vecb_C, \dots, \veca_{\tau - 1}, \vecb_{\tau - 1}, \Psi_{\tau - 1})
        \leq g(D_\tau, \veca_\tau, \vecb_\tau, \Gamma_\tau');\\
        &h(C, D_i, \veca_C, \vecb_C, \dots, \veca_i, \vecb_i, \Psi_i)
        \leq h(C, D_{i + 1}, \veca_C, \vecb_C, \dots, \veca_{i + 1}, \vecb_{i + 1}, \Psi_{i + 1}) \\
        & \hspace{15em} + g(D_{i + 1}, \veca_{i + 1}, \vecb_{i + 1}, \Gamma_{i + 1}'), \qquad 1 \leq i \leq \tau - 2; \\
        &g(C, \veca_C, \vecb_C, \Psi_0 + (\Gamma_C - \Gamma_C'))
        \leq h(C, D_1, \veca_C, \vecb_C, \veca_1, \vecb_1, \Psi_1)
        + g(D_1, \veca_1, \vecb_1, \Gamma_1').
    \end{align*}
    Summing over $0 \leq i \leq \tau - 1$, we have
    \begin{align}
        g(C, \veca_C, \vecb_C, \Psi_0 + (\Gamma_C - \Gamma_C'))
        &\leq \sum_{i = 1}^{\tau} g(D_i, \veca_i, \vecb_i, \Gamma_i') \notag\\
        &\leq G(C, \veca_C, \vecb_C, \Gamma_C) && \text{By \eqref{eqn:G_c_w_prime_sum}.}
        \label{eqn:lb_G_by_g}
    \end{align}
    Finally, note that 
    \begin{align*}
        \Psi_0 &= R(\Psi_1 + \Gamma_1') \\
        &\leq (1 + \epsilon') \Psi_1 + (1 + \epsilon') \Gamma_1' &&\text{By the definition of $R$}\\
        &\leq (1 + \epsilon')^2 \Psi_2 + (1 + \epsilon')^2 \Gamma_2' + (1 + \epsilon') \Gamma_1' \\
        &\leq \dots \\
        & \leq (1 + \epsilon')^{\tau - 1} \Psi_{\tau - 1} + \sum_{i = 1}^{\tau - 1} (1 + \epsilon')^i \Gamma_i' \\
        & \leq (1 + \epsilon)^{\tau - 1} \sum_{i = 1}^{\tau} \Gamma_i' \\
        &= (1 + \epsilon')^{\tau - 1} \alpha_{\ell - 1} \sum_{i = 1}^{\tau} \Gamma_i && \text{Recall $\Gamma_i' = \alpha_{\ell - 1} \Gamma_i$}\\
        & \leq (1 + \epsilon')^{\tau} \alpha_{\ell - 1} \sum_{i = 1}^{\tau} \Phi_i &&\text{Since $\Gamma_i \leq (1 + \epsilon') \Phi_i$}\\
        & \leq (1 + \epsilon')^{\tau} \alpha_{\ell - 1} \Gamma_C' &&\text{Since } \sum_i \Phi_i \leq \Gamma_C'\\
        & \leq (1 + \epsilon')^{\tau} \alpha_{\ell - 1} \Gamma_C - (\Gamma_C - \Gamma_C').
    \end{align*}
    We thus have 
    \[\Psi_0 + (\Gamma_C - \Gamma_C') \leq (1 + \epsilon')^{\tau} \alpha_{\ell - 1}\Gamma_C \leq \alpha_\ell \Gamma_C. \] 
    Combining with \eqref{eqn:lb_G_by_g}, we have $G(C, \veca_C, \vecb_C, \Gamma_C) \geq g(C, \veca_C, \vecb_C, \alpha_\ell \Gamma_C)$, completing the proof.
\end{proof}

\section{Missing Proofs in \Cref{sec:kmedian_bounded_clients}}

\subsection{Proof of \Cref{lemma:kmedian_bounded_clients_good_solution_construction}}
\label{appendix:proof_of_kmedian_bounded_clients_construction}

\lemmakmedianboundedclientsgoodsolutionconstruction*

The proof is essentially the same as \Cref{lemma:kmedian_bounded_centers_good_solution_construction}, so we only provide a sketch here.

\begin{proof}[Proof of \Cref{lemma:fl_bounded_clients_good_solution} (sketch)]
    Let $F^* \subseteq Y$ be the optimal $k$-median solution.
    Wlog, assume $|F^*| = |S| = k$, since otherwise we can add arbitrary facilities to them.
    For $f \in F^*$, recall that $\proj{S}{f}$ is the point in $S$ closest to $f$.
    For $s \in S$, denote $S^{-1}_{F^*}(s) = \{f \in F^* \colon \proj{S}{f} = s\}$.
    Denote $S_0 = \{s \in S \colon |S^{-1}_{F^*}(s)| = 0\}$, 
    $S_1 = \{s \in S \colon |S^{-1}_{F^*}(s)| = 1\}$
    and $S_{\geq 2} = \{s \in S \colon |S^{-1}_{F^*}(s)| \geq 2\}$.
    Denote $S_{\geq 1} := S_1 \cup S_{\geq 2}$.
    For $s \in S_{\geq 1}$, denote $f_s := \argmin_{f \in S^{-1}_{F^*}(s)} \dist(f, s)$, i.e., $f_s$ is the closest point to $s$ among $S^{-1}_{F^*}(s)$, breaking ties arbitrarily.
    Denote $F_S := \{f \in F^* \colon \exists s \in S, \text{ s.t. } f = f_s\}$.

    Following~\cite{Cohen-AddadFS21}, the center set $F$ is constructed as follows:
    \begin{itemize}
        \item \textbf{Step 1:} Among the centers in $F^* \setminus F_S$, remove from $F^*$ a subset of size $1000 c \epsilon |F^* \setminus F_S|$ that yields the minimum cost increase after removal, where $c$ is the same constant in \Cref{lemma:badly_cut_ambient}.\footnote{This step requires us to assume $\epsilon \leq 1/(1000c)$. Since $1000 c$ is a constant, there is no loss of generality to make such assumption.}
        The resulting center set is denoted as $\overline{F}^*$.
        Let $F = \overline{F}^*$.
        \item \textbf{Step 2:} For $s \in S_0$, if $s$ is a bad center, then add $s$ to $F$.
        \item \textbf{Step 3:} For $s \in S_{\geq 1}$, if $s$ is a bad center, then add $s$ to $F$, and remove $f_s$ from $F$.
    \end{itemize}
    Formally, $F := \overline{F}^* \setminus \{f_s \colon s \in \Bad_S \cap S_{\geq 1}\}  \cup \Bad_S$.

    The proofs of properties (a), (b) and (c) are the same as \Cref{lemma:kmedian_bounded_centers_good_solution_construction}.
\end{proof}

\subsection{Proof of \Cref{lemma:kmedian_bounded_clients_discretization_is_ok}: Cost of the Returned Solution}
\label{appendix:proof_of_kmedianboundedclientsdiscretizationisok}

\lemmakmedianboundedclientsdiscretizationisok*

Following \Cref{appendix:proof_of_lemma:discretization_is_ok}, for every entry $(C, \veca_C, \vecb_C, \Gamma_C)$ in the DP table,
Let $g(C, \veca_C, \vecb_C, \Gamma_C)$ be the value computed by the dynamic program \eqref{eqn:kmedian_bounded_clients_g} and \eqref{eqn:kmedian_bounded_clients_h}.
Let $G(C, \veca_C, \vecb_C, \Gamma_C)$ be the true value of that entry, i.e., the
minimum number of facilities required to be placed in $C$, such that the cost inside $C$ is at most $\Gamma_C$, which is formally defined in \eqref{eqn:kmedian_bounded_clients_main_DP}.
$G(C, \veca_C, \vecb_C, \Gamma_C)$ can be different from $g(C, \veca_C, \vecb_C, \Gamma_C)$, mainly because of the discretization of $\Gamma_C$ into powers of $(1 + \epsilon')$, where $\epsilon' = \frac{\epsilon}{2^{O(\ddim)} \log \Delta}$.
The following lemma is an analogy to \Cref{lemma:G_bounded_by_g}, which 
claims that $G$ can be two-sided bounded by $g$.

\begin{lemma}\label{lemma:G_bounded_by_g_clients}
    Let $0 \leq \ell \leq L$ and $C \in \decom_\ell$ be a level $\ell$ cluster.
    Then for every configuration $\veca_C, \vecb_C$, and $\Gamma_C \in [\median(X, S)/n, 2 \median(X, S)]$ being powers of $(1 + \epsilon')$, it holds
    \begin{align*}
        g(C, \veca_C, \vecb_C, \alpha_\ell \Gamma_C) \leq G(C, \veca_C, \vecb_C, \Gamma_C) \leq g(C, \veca_C, \vecb_C, \Gamma_C),
    \end{align*}
    for $\alpha_\ell = (1 + \epsilon')^{\ell \cdot 2^{O(\ddim)}}$.
\end{lemma}

We use \Cref{lemma:G_bounded_by_g_clients} to prove \Cref{lemma:discretization_is_ok}.

\begin{proof}[Proof of \Cref{lemma:discretization_is_ok}]
    Applying \Cref{lemma:G_bounded_by_g_clients} to $X' \in \decom_L$ and noting $\alpha_{L} = (1 + \epsilon')^{2^{O(\ddim)} \log \Delta} \leq 1 + \epsilon$,
    we have that for every configuration $\veca, \vecb$ and $\Gamma \in [\median(X, S)/n, 2 \median(X, S)]$, 
    \[g(X', \veca, \vecb, (1 + \epsilon) \Gamma) \leq G(X', \veca, \vecb, \Gamma) \leq g(X', \veca, \vecb, \Gamma).\]

    Recall our algorithm outputs the smallest $\Gamma$, such that there exists $\veca, \vecb$ with $g(X', \veca, \vecb, \Gamma) \leq k$, and $\widehat{F}$ is the solution corresponding to this $\Gamma$.
    Therefore, $G(X', \veca, \vecb, \Gamma) \leq k$.
    By the definition of $G$, this implies the optimal $k$-median value of $X'$ (under portal-respecting distance) is at most $\Gamma$.

    By the minimality of $\Gamma$, for every $\veca, \vecb$, we have $g(X', \veca, \vecb, \Gamma/(1 + \epsilon')) > k$.
    Then $G(X', \veca, \vecb, \frac{\Gamma}{(1 + \epsilon')(1 + \epsilon)}) \geq g(X', \veca, \vecb, \Gamma/(1 + \epsilon')) > k$.
    This implies the optimal $k$-median value of $X'$ (under portal-respecting distance) is greater than $\frac{\Gamma}{(1 + \epsilon')(1 + \epsilon)}$.

    We thus conclude that 
    \begin{align*}
        \sum_{x \in X'} \dportP(x, \widehat{F}) 
        \leq \Gamma
        < (1 + \epsilon')(1 + \epsilon) \min_{F \subseteq Y, |F| \leq k} \sum_{x \in X'} \dportP(x, F)
    \end{align*}

    Rescaling $\epsilon$ concludes the proof.
\end{proof}

Next, we prove \Cref{lemma:G_bounded_by_g_clients}.
The proof is almost the same as \Cref{lemma:G_bounded_by_g}, except now we need to deal with ornaments in $\Schild(C)$.

\begin{proof}[Proof of \Cref{lemma:G_bounded_by_g_clients}]
    We prove the lemma by induction on $\ell$.
    Throughout, denote by $\tau := |\child(C)|$ for simplicity.
    
    When $\ell = 0$, it corresponds to the base cases of the DP, where the computation of $g$ is accurate.
    Hence, 
    \[
    g(C, \veca_C, \vecb_C, \Gamma_C) = G(C, \veca_C, \vecb_C, \Gamma_C) \geq G(C, \veca_C, \vecb_C, \alpha_0 \Gamma_C) = g(C, \veca_C, \vecb_C, \alpha_0 \Gamma_C).
    \]

    Assume the inequality holds for level $\ell - 1$, and consider any level $\ell$ cluster $C$.
    By the dynamic program \eqref{eqn:kmedian_bounded_clients_g} and \eqref{eqn:kmedian_bounded_clients_h}, 
    there exists a sequence of values $\Gamma_1, \Gamma_2, \dots, \Gamma_{\tau}$ which are powers of $(1 + \epsilon')$ and satisfy $\sum_i \Gamma_i \leq \Gamma_C$, a sequence of configurations $\veca_1, \vecb_1, \dots, \veca_{\tau}, \vecb_{\tau}$, 
    and an integer $K = K(\veca_C, \vecb_C, \veca_1, \vecb_1, \dots, \veca_\tau, \vecb_\tau)$ which corresponds to the number of facilities in $\Schild(C)$, such that
    \begin{equation}\label{eqn:g_C_w_sum_clients}
        g(C, \veca_C, \vecb_C, \Gamma_C) = \sum_{i = 1}^{\tau} g(D_i, \veca_i, \vecb_i, \Gamma_i) + K.
    \end{equation}
    It is important to note that $K$ only depends on the configurations $\veca_C, \vecb_C, \veca_1, \vecb_1, \dots, \veca_\tau, \vecb_\tau$ and does not depend on $\Gamma_C$ and $\{\Gamma_i\}$.

    Since configurations $\{\veca_i, \vecb_i\}$ and $K$ are consistent with $\veca_C, \vecb_C$, and $\sum_i \Gamma_i \leq \Gamma_C$, we have 
    \begin{align*}
        G(C, \veca_C, \vecb_C, \Gamma_C) &\leq \sum_{i = 1}^{\tau} G(D_i, \veca_i, \vecb_i, \Gamma_i) + K&& \text{By definition of $G$} \\
        & \leq \sum_{i = 1}^{\tau} g(D_i, \veca_i, \vecb_i, \Gamma_i) + K&& \text{Induction hypothesis} \\
        &= g(C, \veca_C, \vecb_C, \Gamma_C) &&\text{By \eqref{eqn:g_C_w_sum_clients}}.
    \end{align*}

    For the other direction, by the definition of $G$, there exists a sequence of values $\Phi_1, \Phi_2, \dots, \Phi_{\tau}$ which satisfy $\sum_i \Phi_i \leq \Gamma_C$, 
    a sequence of configurations $\veca_1, \vecb_1, \dots, \veca_{\tau}, \vecb_{\tau}$, 
    and an integer $K = K(\veca_C, \vecb_C, \veca_1, \vecb_1, \dots, \veca_\tau, \vecb_\tau)$ which corresponds to the number of facilities in $\Schild(C)$, 
    such that
    \begin{equation}\label{eqn:G_C_w_sum_clients}
        G(C, \veca_C, \vecb_C, \Gamma_C) = \sum_{i = 1}^{\tau} G(D_i, \veca_i, \vecb_i, \Phi_i) + K.
    \end{equation}
    Note that $\Phi_i$'s are no longer necessarily powers of $(1 + \epsilon')$.
    Let $R$ be the operator that rounds every value $v$ to the smallest power of $(1 + \epsilon')$ greater than $v$, i.e., $R(v) = (1 + \epsilon')^{\lceil \log_{(1 + \epsilon')} v \rceil}$.
    Define $\Gamma_i := R(\Phi_i)$ for every $i \in [\tau]$.
    We have
    \begin{align}
        G(C, \veca_C, \vecb_C, \Gamma_C) &= \sum_{i = 1}^{\tau} G(D_i, \veca_i, \vecb_i, \Phi_i) + K &&\text{By \eqref{eqn:g_C_w_sum_clients}} \notag \\
        & \geq \sum_{i = 1}^{\tau} G(D_i, \veca_i, \vecb_i, \Gamma_i) + K &&\text{Since $\Gamma_i \geq \Phi_i$} \notag \\
        & \geq \sum_{i = 1}^{\tau} g(D_i, \veca_i, \vecb_i, \alpha_{\ell - 1} \Gamma_i) + K &&\text{Induction hypothesis}. \label{eqn:G_c_w_prime_sum_clients}
    \end{align}
    Denote $\Gamma_i' := \alpha_{\ell - 1} \Gamma_i$ for short.
    Consider the following sequence of values, which are powers of $(1 + \epsilon')$.
    \begin{align*}
        &\Psi_{\tau} = 0, \\
        &\Psi_i = R(\Psi_{i + 1} + \Gamma_{i + 1}'), \quad 0 \leq i \leq \tau - 1.
    \end{align*}
    Note that $\{\veca_i, \vecb_i\}$ and $K$ are consistent with $\veca_C, \vecb_C$, and that $\Psi_{i + 1} + \Gamma_{i + 1}' \leq \Psi_i$, we thus have
    \begin{align*}
        & h(C, D_{\tau}, \veca_C, \vecb_C, \dots, \veca_{\tau}, \vecb_{\tau}, \Psi_{\tau}) = K; \\
        &h(C, D_i, \veca_C, \vecb_C, \dots, \veca_i, \vecb_i, \Psi_i)
        \leq h(C, D_{i + 1}, \veca_C, \vecb_C, \dots, \veca_{i + 1}, \vecb_{i + 1}, \Psi_{i + 1}) \\
        & \hspace{15em} + g(D_{i + 1}, \veca_{i + 1}, \vecb_{i + 1}, \Gamma_{i + 1}'), \qquad 1 \leq i \leq \tau - 1; \\
        &g(C, \veca_C, \vecb_C, \Psi_0)
        \leq h(C, D_1, \veca_C, \vecb_C, \veca_1, \vecb_1, \Psi_1)
        + g(D_1, \veca_1, \vecb_1, \Gamma_1').
    \end{align*}
    Summing over $0 \leq i \leq \tau$, we have
    \begin{align}
        g(C, \veca_C, \vecb_C, \Psi_0)
        &\leq \sum_{i = 1}^{\tau} g(D_i, \veca_i, \vecb_i, \Gamma_i') + K \notag\\
        &\leq G(C, \veca_C, \vecb_C, \Gamma_C) && \text{By \eqref{eqn:G_c_w_prime_sum_clients}.}
        \label{eqn:lb_G_by_g_clients}
    \end{align}
    Finally, note that 
    \begin{align*}
        \Psi_0 &= R(\Psi_1 + \Gamma_1') \\
        &\leq (1 + \epsilon') \Psi_1 + (1 + \epsilon') \Gamma_1' &&\text{By the definition of $R$}\\
        &\leq (1 + \epsilon')^2 \Psi_2 + (1 + \epsilon')^2 \Gamma_2' + (1 + \epsilon') \Gamma_1' \\
        &\leq \dots \\
        & \leq (1 + \epsilon')^{\tau - 1} \Psi_{\tau - 1} + \sum_{i = 1}^{\tau - 1} (1 + \epsilon')^i \Gamma_i' \\
        & \leq (1 + \epsilon)^{\tau - 1} \sum_{i = 1}^{\tau} \Gamma_i' \\
        &= (1 + \epsilon')^{\tau - 1} \alpha_{\ell - 1} \sum_{i = 1}^{\tau} \Gamma_i && \text{Recall $\Gamma_i' = \alpha_{\ell - 1} \Gamma_i$}\\
        & \leq (1 + \epsilon')^{\tau} \alpha_{\ell - 1} \sum_{i = 1}^{\tau} \Phi_i &&\text{Since $\Gamma_i \leq (1 + \epsilon') \Phi_i$}\\
        & \leq (1 + \epsilon')^{\tau} \alpha_{\ell - 1} \Gamma_C &&\text{Since } \sum_i \Phi_i \leq \Gamma_C.
    \end{align*}
    We thus have 
    \[\Psi_0 \leq (1 + \epsilon')^{\tau} \alpha_{\ell - 1}\Gamma_C \leq \alpha_\ell \Gamma_C. \] 
    Combining with \eqref{eqn:lb_G_by_g_clients}, we have $G(C, \veca_C, \vecb_C, \Gamma_C) \geq g(C, \veca_C, \vecb_C, \alpha_\ell \Gamma_C)$, completing the proof.
\end{proof}

\end{appendices}

\end{document}